\documentclass[11pt]{article}

\usepackage{amsmath,amsthm,amssymb}
\usepackage{enumerate}
\usepackage{graphicx}
\usepackage{epsfig}

\usepackage{hyperref}
\hypersetup{
    colorlinks=true,
    linkcolor=blue,
    filecolor=magenta,      
    urlcolor=blue,
}

\usepackage[T2A]{fontenc}
\usepackage{inputenc}
\inputencoding{cp1251}
\usepackage{enumitem}
\usepackage{hyperref}

  \usepackage{geometry}
\newcommand*{\mailto}[1]{\href{mailto:#1}{\nolinkurl{#1}}}

\newtheorem{theorem}{Theorem}[section]
\newtheorem{definition}[theorem]{Definition}
\newtheorem{lemma}[theorem]{Lemma}
\newtheorem{example}[theorem]{Example}
\newtheorem{proposition}[theorem]{Proposition}
\newtheorem{corollary}[theorem]{Corollary}

\newtheorem{remark}[theorem]{Remark}
\newtheorem{remarks}[theorem]{Remarks}
\newcommand{\fr}{\frac}

\newcommand{\R}{{\mathbb R}}

\newcommand{\N}{{\mathbb N}}
\newcommand{\Z}{{\mathbb Z}}
\newcommand{\Co}{{\mathbb C}}

\newcommand{\toEF}{\stackrel{\cE_F}{\mathop{\longrightarrow}}}
\newcommand{\toEFN}{\stackrel{\cE_F^N}{\mathop{\longrightarrow}}}

\newcommand{\cT}{{\mathcal T}}

\newcommand{\cA}{{\cal A}}
\newcommand{\bA}{{\bf A}}

\newcommand{\cE}{{\cal E}}

\newcommand{\cH}{{\cal H}}

\def\o{\mathaccent"7017}
\newcommand\Ho{\o{H}}

\newcommand{\cF}{{\cal F}}

\newcommand{\cO}{{\cal O}}
\newcommand{\cP}{{\cal P}}

\newcommand{\cS}{{\cal S}}

\newcommand{\cX}{{\cal X}}
\newcommand{\cZ}{{\cal Z}}

\newcommand{\al}{\alpha}

\newcommand{\om}{\omega}
\newcommand{\vp}{\varphi}

\newcommand{\si}{\sigma}

\newcommand{\De}{\Delta}
\newcommand{\de}{\delta}

\newcommand{\ga}{\gamma}

\newcommand{\ve}{\varepsilon}

\newcommand{\lam}{\lambda}

\newcommand{\Si}{\Sigma}

\newcommand{\ti}{\tilde}
\newcommand{\na}{\nabla}
\newcommand{\pa}{\partial}
\newcommand{\rot}{{\rm rot\5}}
\newcommand{\dv}{{\rm div\5}}
\newcommand{\const}{{\rm const}}

\newcommand{\rRe}{{\rm Re\5\5}}
\newcommand{\rIm}{{\rm Im\5\5}}
\newcommand{\supp}{{\rm supp\5\5}}
\newcommand{\ext}{{\rm ext}}

\newcommand{\ov}{\overline}

\newcommand{\5}{{\hspace{0.5mm}}}

\newcommand{\ds}{\displaystyle}

\date{}

\numberwithin{equation}{section}

\newcommand{\ci}{\cite}
\newcommand{\la}{\label}

\newcommand{\Norm}[1]{\left\Vert #1 \right\Vert}
\newcommand{\norm}[1]{\Vert #1 \Vert}

\newcommand{\lr}{\longrightarrow}
\newcommand{\st}{\stackrel}
\newcommand{\tocEF}{\st{{\cal E}_F}\lr}

\newcommand{\be}{\begin{equation}}
 \newcommand{\ee}{\end{equation}}

 \newcommand{\beqn}{\begin{eqnarray}}
 \newcommand{\eeqn}{\end{eqnarray}}

\newcommand{\ba}{\begin{array}}
 \newcommand{\ea}{\end{array}}

\newcommand{\bd}{\begin{definition}}
 \newcommand{\ed}{\end{definition}}
\newcommand{\bt}{\begin{theorem}}
 \newcommand{\et}{\end{theorem}}

\newcommand{\bp}{\begin{proposition}}
 \newcommand{\ep}{\end{proposition}}

\newcommand{\bl}{\begin{lemma}}
 \newcommand{\el}{\end{lemma}}
\newcommand{\bc}{\begin{corollary}}
 \newcommand{\ec}{\end{corollary}}

\newcommand{\bex}{\begin{example}}
 \newcommand{\eex}{\end{example}}
 
\newcommand{\bexs}{\begin{examples}}
 \newcommand{\eexs}{\end{examples}}

\newcommand{\bexe}{\begin{exercice}}
 \newcommand{\eexe}{\end{exercice}}

\newcommand{\br}{\begin{remark} }
 \newcommand{\er}{\end{remark}}
\newcommand{\Verts}{\begin{remarks}}
 \newcommand{\ers}{\end{remarks}}

\newcommand{\bce}{\begin{center}}
\newcommand{\ece}{\end{center}}

\date{}

\numberwithin{equation}{section}

\begin{document}

\bce
{\huge\bf   Attractors   of  Hamilton nonlinear 
 \bigskip
 
partial differential equations}
 \bigskip \bigskip

 {\Large A.I. Komech} \footnote{
 Supported partly by Austrian Science Fund (FWF) P28152-N35
 }
 \medskip
 \\
{\it
  \centerline{Institute for Information Transmission Problems RAS}
     }
 
  \centerline{akomech@iitp.ru}
\medskip

 {\Large E.A. Kopylova } \footnote{
 Supported partly  by grant of RFBRa 18-01-00524}
 \medskip
 \\
{\it
\centerline{Institute for Information Transmission Problems RAS}
  }
  \centerline{ek@iitp.ru}

 \bigskip\bigskip
\bigskip

%\maketitle
\centerline{{\it To the memory of Mark Vishik}}
\bigskip

\ece
\bigskip

%\end{document}
\begin{abstract}
We survey the theory of attractors of nonlinear Hamiltonian partial differential  equations since its 
appearance in 1990. These are  results  on global attraction to stationary states, to solitons and to stationary orbits,  on adiabatic effective dynamics of solitons and their asymptotic stability.  Results of numerical simulation are given.

The obtained results allow us to formulate a new general  conjecture  on attractors of $G$ -invariant nonlinear Hamiltonian partial differential equations.
 
This conjecture   suggests a novel  dynamical interpretation  of basic quantum phenomena: Bohr's transitions between quantum stationary states, de Broglie's wave-particle duality and 
Born's probabilistic interpretation.

 \end{abstract}
 
  \bigskip \bigskip

\bigskip
\bigskip
\bigskip
\tableofcontents

%%%%%%%%%%%%%%%%%%%%%%%%%%%%%%%%%%%%%%%%%%%%%%%%%%%%%%%%%%%%%%%
%%%%%%%%%%%%%%%%%%%%%%%%%%%%%%%%%%%%%%%%%%%%%%%%%%%%%%%%%%%%%%%
\section{Introduction} 
%%%%%%%%%%%%%%%%%%%%%%%%%%%%%%%%%%%%%%%%%%%%%%%%%%%%%%%%%%%%%%%

This paper is a survey of the results on long time behaviour and attractors for nonlinear Hamilton
partial differential equations that appeared since 1990.

Theory of attractors for nonlinear PDEs originated from the seminal paper of Landau \cite{L1944} published in 1944, where he
suggested the first mathematical interpretation of the onset of turbulence as the growth of the dimension of attractors of the Navier--Stokes equations 
when the Reynolds number increases.

The foundation for  corresponding mathematical theory was laid in 1951 by Hopf who established for the first time the existence of
global solutions to the 3D Navier--Stokes equations \cite{Hopf1951}. He introduced the ``method of compactness'' which is a nonlinear version
of the Faedo-Galerkin approximations. This method relies on a priori estimates and Sobolev embedding theorems. It has strongly influenced
the development of the theory of nonlinear PDEs, see \cite{Lions1969}.

Modern development of the theory of attractors for general \textit{dissipative systems}, i.e. systems with friction 
(the Navier--\allowbreak Stokes equations,  nonlinear parabolic equations, reaction-diffusion equations, wave equations with friction, etc.), 
as originated in the 1975--1985's in the works of Foias, Hale, Henry, Temam, and others \cite{FMRT2001, H1988, H1981},
was developed further in the works of Vishik, Babin, Chepyzhov, and others \cite{BV1992, CV2002}.
A typical result of this theory in the absence of external excitation is  global convergence to stationary states:
for any finite energy solution to  dissipative {\it autonomous} equation in a~region $\Omega\subset\mathbb R^n$, there is a convergence
\begin{equation}\label{at11}
	\psi(x,t) \to S(x), \qquad t\to + \infty.
\end{equation}
Here $S(x)$ is a stationary  solution with suitable boundary conditions, and this convergence holds as a~rule in the $L^2(\Omega)$-metric.
In particular, the relaxation to an equilibrium regime in chemical reactions is due to the energy dissipation.
\medskip

A development of a similar theory for  \textit{Hamiltonian PDEs} seemed unmotivated and impossible in view of energy conservation 
and time reversal for these equations. However, as it turned out, such a~theory is possible and its shape was suggested
by a~novel mathematical interpretation of fundamental postulates of quantum theory:
\medskip\\
I. Transitions between quantum stationary orbits (Bohr 1913).
\smallskip\\
II. Wave-particle duality (de Broglie 1924).
\smallskip\\
III. Probabilistic interpretation (Born 1927).
\smallskip\\
\noindent Namely, postulate I can be interpreted as  global attraction 
 (\ref {atU})
of all quantum trajectories to an attractor formed by stationary orbits
(see Appendix),
and postulate II can be interpreted as decay into solitons (\ref {attN}).
The probabilistic interpretation  also can be justified by the asymptotics 
(\ref {attN}). More details can be found in \ci{Kjumps2019}.
\smallskip

Investigations of the 1990--2019's suggest that such long time asymptotics of solutions are in fact typical 
for nonlinear Hamiltonian PDEs. 
These results are presented in this article. This theory differs significantly from the theory of attractors of dissipative systems
where the attraction to stationary states is due to an energy dissipation caused by a friction.
For Hamiltonian equations the friction and energy dissipation are absent, 
and the attraction is caused by radiation which irrevocably  brings the energy to infinity. 
\medskip

The modern development of the theory of nonlinear Hamiltonian equations dates back to J\"orgens~\cite{Jor1961},
who has established the existence of global solutions for nonlinear wave equations of the form
\begin{equation}\label{Jw}
	\ddot\psi(x,t) = \Delta \psi (x,t) + F(\psi (x,t)), \qquad x \in \mathbb R^n,
\end{equation}
developing the Hopf method of compactness.
The subsequent studies  in this direction were well reflected by J.-L. Lions \cite{Lions1969}.

First results on the long time asymptotics of solutions to nonlinear Hamiltonian PDEs were obtained by
Segal \cite{Segal1966, Segal1968}, Morawetz and Strauss \cite{Mor1968, 
MS1972,
St68}. In these papers
 {\it local energy decay} is proved for solutions to equations (\ref{Jw}) with {\it defocusing type} nonlinearities 
 $F(\psi) = -m^2\psi-\kappa | \psi |^p \psi$, where $m^2\ge 0$, $\kappa>0$, and $p>1$. 
 Namely, for sufficiently smooth and small initial states, one has
\begin{equation}\label{ledec}
	\int_{|x|<R} [|\dot\psi(x,t)|^2+|\nabla\psi(x,t)|^2+|\psi(x,t)|^2]dx\to 0, \qquad t\to\pm\infty
\end{equation}
for any finite $R>0$. Moreover, the corresponding nonlinear wave and  scattering operators are constructed.
In the works of Strauss \cite{St81-1, St81-2}, the completeness of scattering is established for small solutions to more general equations.
The decay (\ref{ledec}) means that the energy escapes each bounded region
for large times.
\medskip

For convenience,  characteristic properties of all finite energy solutions to an equation will be referred to as {\it global}, 
in order to distinguish them from the corresponding {\it local}
properties for solutions with initial data sufficiently close to an attractor.
\medskip

All the above-mentioned results  on  local energy decay (\ref{ledec}) mean that the
corresponding {\it local attractor} of small initial states consists of the zero point only.
First results on  {\it global attraction} for nonlinear Hamiltonian PDEs 
were obtained by one of the authors
in the 1991--1995's for 1D models \cite{K1991, K1995a, K1995b}, and were later extended to nD equations.
Let us note that global attraction to a~(proper) attractor is impossible for any finite-dimensional
Hamiltonian  system because of energy conservation.

Global attraction  for Hamiltonian PDEs is derived from an analysis of the irreversible energy radiation to infinity, 
which plays the role of the  dissipation. Such  analysis requires subtle methods of harmonic analysis: the Wiener Tauberian theorem,
the Titchmarsh convolution theorem, the theory of quasi-measures, the Paley-Wiener estimates, 
eigenfunction expansions for nonselfadjoint Hamiltonian operators based on M.G.~Krein theory of $J$-selfadjoint operators, and others.

The results obtained so far indicate a~certain dependence of long-time asymptotics of solutions on  symmetry group of an equation:
for example, it may be the trivial group $G = \{e\} $, or the unitary group $G = U(1) $, or the group of translations $G = \mathbb R^n$.
Namely, the  results suggest 
the conjecture
that for  ``generic'' nonlinear Hamilton {\it autonomous} PDEs with a Lie symmetry group $G$,
any finite energy solution admits the asymptotics
\begin{equation}\label{at10}
	\psi(x,t) \sim e^{g_\pm t} \psi_\pm (x) , \qquad t \to \pm \infty.
\end{equation}
Here, $e^{g_\pm t}$ is a representation of  one-parameter subgroup of the symmetry group $G$ which corresponds to the generators 
$g_\pm$ from the corresponding Lie algebra,
while $\psi_\pm(x) $ are some
``scattering states'' depending on the considered trajectory $\psi(x,t)$.
Both  pairs $(g_+, \psi_+)$ 
and $(g_-, \psi_-)$
are solutions to the corresponding nonlinear eigenfunction problem.
\medskip

 In the case of the  trivial symmetry group,  the conjecture (\ref{at10}) means  global attraction to the corresponding stationary states
\begin{equation}\label{ate}
  \psi(x,t) \to S_\pm(x) , \qquad t\to\pm\infty
\end{equation}
(see Fig.~\ref{fig-1}),
where  $S_\pm(x) $  depend on  considered trajectory $\psi(x,t)$,
and the convergence holds in local seminorms, i.e., in norms  of type $L^2(|x|<R)$ with any $R>0$.
The convergence (\ref{ate}) in global norms (i.e., corresponding to $R=\infty$) cannot hold due to the energy conservation.

In particular, the asymptotics (\ref{ate}) can be easily demonstrated for the d'Alembert equation, see (\ref{dal})-- (\ref{dal3}).
In this example the convergence (\ref{ate}) in global norms obviously fails due to the presence of travelling waves $f(x\pm t)$.
Similarly,   a solution  to  3D wave equation  with a unit propagation velocity 
is concentrated in spherical layers $ | t | -R <| x | <| t | + R $ if  initial data has a support in the ball $ | x | \le R $.
Therefore, the solution converges to zero when $ t \to \pm \infty $, although its energy remains constant. 
This convergence  corresponds to the well-known {\it  strong  Huygens principle}. 
Thus, attraction to stationary states (\ref {ate}) is a generalization of the strong Huygens principle to non-linear equations.
The difference is that for  linear wave equation the limit is always zero, while for non-linear equations
the limit can be any stationary solution.
\medskip

Further, in the case of symmetry group of translations $G = \mathbb R ^ n $ asymptotics (\ref {at10}) means global attraction
to solitons (traveling waves)
\begin {equation} \label {att}
  \psi (x, t) \sim \psi_ \pm (x-v_ \pm t), \qquad t \to \pm \infty,
\end {equation}
for {\it generic} translation-invariant equation. In this case the convergence holds in local seminorms
{\it in the comoving frame of reference}, that is, in $L ^ 2 (| x-v_ \pm t | < R) $ for any $ R> 0$.
The validity of such local asymptotics in comoving reference systems suggests that
there may be several such solitons, which provide the refined asymptotics
\be \la {attN}
  \psi (x, t) \sim \sum_k \psi_ \pm (x-v ^ k_ \pm t) + w_ \pm (x, t), \qquad t \to \pm \infty,
\ee
where $w_ \pm $ are  some dispersion waves, being  solutions to  corresponding free equation, and
convergence holds now  in some {\it global norm}.
A trivial example gives the d'Alembert equation (\ref {dal}) with solutions $\psi (x, t) = f (x-t) + g (x + t)$.

Asymptotics with several solitons (\ref{attN}) were discovered first in 1965 by Kruskal and Zabusky
in numerical simulations of the Korteweg--de Vries equation (KdV). Later on, global asymptotics of this type
were proved   for nonlinear {\bf   integrable} 
translation-invariant equations (KdV and others) by Ablowitz, Segur, Eckhaus, van Harten, and others 
using the method of {\it inverse scattering problem}
\cite{EvH}.
\smallskip

Finally, for the unitary symmetry group $G=U(1)$,  asymptotics (\ref{at10}) mean  global attraction to ``stationary orbits'' (or {\it ``solitary waves''})
\begin{equation}\label{atU}
  \psi(x,t)\sim\psi_\pm(x) e^{-i \omega_\pm t} , \qquad t \to \pm \infty
\end{equation}
in the same local seminorms (see Fig.~\ref{fig-3}).
These asymptotics were inspired by Bohr's postulate on transitions between quantum stationary states (see Appendix for details).
Our results confirm such asymptotics for generic $U(1)$-invariant nonlinear equations of type (\ref{KG1}) and (\ref{KGN})--(\ref{Dn}).
More precisely, we have proved  global attraction {\it to the manifold of all stationary orbits}, though
the attraction to a particular  stationary orbitы, with fixed $\omega_\pm$, is still open problem.
\smallskip

The existence of stationary orbits $\psi(x)e^{i\omega t}$ for a~broad class of  $U(1)$-invariant
nonlinear wave equations (\ref{Jw}) was extensively studied in the 1960--1980's. 
The most general results were obtained by Strauss, Berestycki and P.-L. Lions \cite{BL83-1, BL83-2, St77}.
Moreover, Esteban, Georgiev and S\'er\'e  constructed stationary orbits
for  nonlinear relativistically-invariant Maxwell--\allowbreak Dirac equations (\ref{DM}).
The orbital stability of stationary orbits has been studied by Grillakis, Shatah, Strauss and others \cite{GSS87, GSS90}.
\smallskip

Let us emphasize that we conjecture  asymptotics (\ref{atU}) for {\it generic} $U(1)$-invariant equations.
This means that  long time behavior 
of solutions
may be quite different for $U(1)$-invariant equations of ``positive codimension''. 
In particular, for solutions to linear Schr\"odinger equation
$$
  i\dot\psi(x,t)=-\Delta\psi(x,t)+V(x)\psi(x,t),\qquad x\in\R^n
$$  
the asymptotics (\ref{atU}) generally fail. Namely, any finite energy solution admits the spectral representation
$$
\psi(x,t)=\sum C_k\psi_k(x)e^{-i\omega_k t}+\int_0^\infty C(\omega)\psi(\omega,x)e^{-i\omega t}d\omega,
$$
where $\psi_k$ and $\psi(\omega,\cdot)$ are  corresponding eigenfunctions of  discrete and  continuous spectrum, respectively. 
The last integral is a dispersion wave which decays to zero in the 
norms $L^2(|x|< R)$ with any $R>0$ 
(under appropriate conditions on the potential $V(x)$).
Respectively, the attractor is the linear span of the eigenfunctions $\psi_k$. 
Thus, the long-time
asymptotics does not reduce to a single term like (\ref{atU}), so the linear case is degenerate in this sense.
Let us note that our results for equations (\ref{KG1}) and (\ref{KGN})--(\ref{Dn}) are established for {\it strictly nonlinear case}: 
see  condition (\ref{C4}) below, which eliminates linear equations.
\smallskip

For more sophisticated symmetry groups $G=U(N)$,  asymptotics (\ref{at10}) mean the attraction to $N$-frequency trajectories, which
can be quasi-periodic. In particular,
the symmetry groups $SU(2)$, $SU(3)$ and others were suggested in 1961 by Gell-Mann and Ne'eman 
for  strong interaction of baryons \cite{GM1962, Ne1962}.
The suggestion relies on  discovered parallelism between empirical data for the baryons, and the ``Dynkin scheme''
of Lie algebra $su(3)$ with $8$ generators (the famous ``eightfold way'').
This theory resulted in the scheme of quarks and in the development of the quantum chromodynamics \cite{AF1963,HM1984},
and in the prediction of a new baryon with prescribed values of its mass and decay products. 
This particle, the $\Omega^-$-hyperon, was promptly discovered experimentally \cite{omega-1964}.

This empirical correspondence between  Lie algebra generators and elementary particles presumably gives an
evidence in favor of the general conjecture (\ref{at10}) for equations with  Lie symmetry groups.
\smallskip

%%%%%%%%%%%%%%%%%%%%%%
Let us note that our conjecture (\ref{at10}) specifies the concept of ``localized solution/coherent structures'' from ``Grande Conjecture''
and ``Petite Conjecture'' of Soffer \cite[p.460]{soffer2006} in the context of $G$-invariant equations.
The Grande Conjecture is proved in \cite{KM2009a} for 1D wave equation coupled to  nonlinear oscillator (\ref{w1m}).
Moreover, a suitable versions of the Grande Conjecture are also proved in \cite{IKS2004a, IKS2004b} for 3D wave, Klein--Gordon 
and Maxwell equations coupled to  relativistic particle with sufficiently small charge \eqref{rosm}; see Remark \ref{rGC}.
Finally, for any matrix symmetry group $G$, (\ref{at10}) implies the Petite Conjecture since the localized solutions 
$e^{g_\pm t}\psi_\pm(x)$ are quasiperiodic then.
\medskip
%%%%%%%%%%%%%%%%%%%%%%%%

Below we dwell upon  available results on the asymptotics \eqref{ate}--\eqref{atU}.
In Sections \ref {s1} and  \ref {s2} we review  results on global attraction  to stationary states and  to solitons, respectively.
Section  \ref {s3} concerns adiabatic effective dynamics of solitons, and Section \ref {s4} concerns the mass-energy equivalence.
In  Section \ref {s5} we give a concise complete proof of the attraction to stationary orbits.
Sections \ref {s6} and \ref {s7} concern asymptotic stability of stationary orbits and solitons, and Section \ref {s8} -- various generalizations.
In  Section \ref {s9} we present results of numerical simulation of soliton asymptotics for relativistic-invariant equations.
In Appendix we comment on the relations between general conjecture (\ref{at10})
and Bohr's postulates in Quantum Mechanics.

\smallskip

In conclusion let us comment on  previous related surveys in this area. The survey \ci{K2000} presents the   results only for  1D equations.
The results on asymptotic stability of solitons were described  in detail in \ci{Im2013} for linear equations coupled to a particle,  
and in \ci{Kumn2013} -- for relativistic-invariant Ginzburg--Landau equations. 
In present article we  give only a short statement of these results (Sections \ref{s2.1}, \ref{s2.2} and \ref{s8}). 
Finally,  present survey gives much more information on our methods than \ci{K2016}. Our main novelties are as  follows: 
\smallskip

i) Streamlined and simplified  proofs of 
the results \ci{KSK1997,KS2000,KS1998}
on
global attraction  to stationary states and to solitons for systems of relativistic particle 
coupled  to scalar wave equation and to the Maxwell equation. These results give  the first rigorous justification of  famous  {\it radiation damping}
in Classical Electrodynamics. We omit  unessential  technical details, 
but explain carefully 
%%%
our approach which relies on the Wiener Tauberian Theorem
%%%
 in Sections \ref{sWP},
 \ref{sML} and \ref{sTIWP}.
\smallskip

ii) The complete proof of  nonlinear analog of the Kato theorem  on the absence of embedded eigenvalues 
(Section \ref{sKT})
which is a crucial point 
in the proof of global attraction  to stationary orbits for  $U(1)$-invariant equations in \ci{K2003,KK2006,KK2007,KK2010a,KK2008,KK2009,KK2010b,KK2013t,
K2017, K2018,KK2019b,KK2019,C2012,C2013}.
  \smallskip
  
  iii)  The informal arguments on the dispersion radiation and the nonlinear 
  spreading of spectrum 
  (Section \ref{Sec-4.8})
  which 
  mean the nonlinear energy transfer  from lower to higher harmonics
  and
  lie behind our application of the Titchmarsh
  Convolution theorem.
   \smallskip
 
 iv) Recent  results \ci{K2017,K2018,KK2019b,KK2019}
 on global attraction for nonlinear wave, Klein-Gordon and Dirac equations with concentrated nonlinearities. 
 We give a detailed survey of the methods and results in Section \ref{sCN}.
 \medskip\\
These methods and ideas  are presented here  for the first time in review literature.  
 
%%%%%%%%%%%%%%%%%%%%%%%%%%%%%%%%%%%%%%%%%%%%%%%%%%%%%%%%%%%%%%%%%%%%%%
\section*{Acknowledgments}
%%%%%%%%%%%%%%%%%%%%%%%%%%%%%%%%%%%%%%%%%%%%%%%%%%%%%%%%%%%%%%%%%%%%%
The authors  express a deep gratitude to H.~Spohn and B.~Vainberg for long-time collaboration on attractors of Hamiltonian PDEs, 
as well as to A.~Shnirelman for many useful long-term discussions. We are also grateful to V.~Imaykin and A.\,A.~Komech
for collaboration lasting many years.

%%%%%%%%%%%%%%%%%%%%%%%%%%%%%%%%%%%%%%%%%%%%%%%%%%%%%%%%%%%%%%%%%%%%%%
%%%%%%%%%%%%%%%%%%%%%%%%%%%%%%%%%%%%%%%%%%%%%%%%%%%%%%%%%%%%%%%%%%%%%%
\setcounter{equation}{0}
\section {Global attraction to stationary states} \la {s1}
%%%%%%%%%%%%%%%%%%%%%%%%%%%%%%%%%%%%%%%%%%%%%%%%%%%%%%%%%%%%%%%%%%%%%%
In this section we review  the results on  global attraction to stationary states \eqref {ate} that were received in 1991-1999 
for nonlinear Hamiltonian PDEs. The first results of this type were obtained for one-dimensional nonlinear wave equations 
\cite {K1991,K1995a,K1995b,K1999,K2000}. Later on these results were  extended to three-dimensional wave equations and Maxwell’s equations 
coupled to a charged relativistic particle \cite {KSK1997, KS2000},
and also to the three-dimensional wave equations with concentrated
nonlinearity. 
In \ci{D2012, T2010} the attraction \eqref {ate}
was established for finite systems of oscillators coupled to an infinite-dimensional thermostat. 
\smallskip

The global attraction \eqref {ate} can be easily demonstrated on trivial (but instructive) example of the d'Alembert equation
\be\la{dal}
   \ddot \psi (x, t) = \psi '' (x, t), \qquad x \in \R.
\ee
All derivatives here and below are understood in the sense of distributions.
This equation is formally equivalent to the Hamilton system
\be\la{hasys}
\dot \psi (t) = D_ \pi \cH, \, \, \dot \pi (t) = - D_ \psi \cH
\ee
with Hamiltonian
\be\la{hamdal}
\cH (\psi, \pi) = \frac12 \int [| \pi (x) | ^ 2 + | \psi '(x) | ^ 2] \, dx, \qquad (\psi, \pi) \in \cE_c: = H ^ 1_c (\R ) \oplus [L ^ 2 (\R )\cap L^1(\R)],
\ee
where $H ^ 1_c (\R )$ is the space of continuous functions  $\psi (x) $ with finite norm 
\be\la{Ec}
 \Vert \psi \Vert_ {H^1_c(\R)}:=\Vert \psi '\Vert_ {L ^ 2 ( \R )} + | \psi (0) |.
\ee
Let moreover,
\be\la{lim}
\psi(x) \to C_ \pm, \quad x \to \pm \infty.
\ee
For such initial data $ (\psi (x, 0), \dot \psi (x, 0)) = (\psi (x), \pi (x)) \in \cE_c $  the d'Alembert formula gives
\begin {equation} \label {dal2}
\psi (x, t)\to S_ \pm (x) =\frac {C_ + + C _-} 2 \pm \frac12 \int _ {- \infty} ^ {\infty} \pi (y) dy, \qquad t \to \pm \infty,
\end {equation}
where the convergence is uniform on every finite interval $ | x| < R $. Moreover,
\be\la{dal3}
\dot \psi (x, t) = \frac {\psi '(x + t) - \psi' (x-t)} 2+ \frac {\pi (x + t) + \pi (x-t)} 2
\to 0, \qquad t \to \pm \infty,
\ee
where convergence holds in $ L ^ 2 (-R, R) $ for each $ R> 0 $. Thus, the attractor is the set of 
states
$(\psi (x), \pi (x)) = (C, 0)$, where $C \in \R $ is any constant. Note that the limits (\ref {dal2})
for positive and negative times may be different.

%%%%%%%%%%%%%%%%%%%%%%%%%%%%%%%%%%%%%%%%%%%%%%%%%%%%%%%%%%%%%%%%%%%%%%
%%%%%%%%%%%%%%%%%%%%%%%%%%%%%%%%%%%%%%%%%%%%%%%%%%%%%%%%%%%%%%%%%%%%%%
\subsection{1D nonlinear wave equations}\la{s2.1}
%%%%%%%%%%%%%%%%%%%%%%%%%%%%%%%%%%%%%%%%%%%%%%%%%%%%%%%%%%%%%%%%%%%%%%
In \cite {K1999}, global attraction to stationary states has been proved for
nonlinear wave equations
of type
\begin {equation} \label {neli}
\ddot \psi (x, t) = \psi '' (x, t) + \chi (x) F (\psi (x, t)), \qquad x \in \R,
\end {equation}
where
\begin {equation} \label {FU-}
\chi \in C_0 ^ \infty (\R), \quad \chi (x) \ge 0, \quad \chi (x) \not \equiv 0; \qquad
F (\psi) = - \nabla U (\psi), \quad \psi \in \R ^ N; \quad U (\psi) \in C ^ 2 (\R ^ N).
\end {equation}
The equation (\ref {neli}) can be formally written as Hamilton’s system
(\ref{hasys})
with Hamiltonian
$$
\cH (\psi, \pi) = \frac12 \int [| \pi (x) | ^ 2 + | \psi '(x) | ^ 2 + \chi (x) U (\psi (x, t) )] \, dx, \qquad (\psi, \pi) \in \cE_c^N=\cE_c\otimes\R^N.
$$
We assume that the potential is confining, i.e.
\begin {equation} \label {conf}
U (\psi) \to\infty, \qquad | \psi | \to \infty.
\end {equation}
In this case, it is easy to prove that a finite energy solution  $Y (t) = (\psi (t), \pi (t)) \in C (\R, \cE_c^N)$
exists and is unique for any initial state $ Y (0) \in \cE_c^N $, and the energy is conserved:
\be \la {hama}
\cH (Y (t)) = \const, \qquad t \in \R.
\ee
\bd
i) $\cE_F^N$ denote the space $\cE_c^N$ endowed with the seminorms
\begin {equation} \label {cER}
\Vert (\psi, \pi) \Vert _ {\cE_c^N, R} = \Vert \psi '\Vert_R + | \psi (0) | + \Vert \pi \Vert_R, \qquad R=1,2,\dots,
\end {equation}
where $ \Vert \cdot \Vert_R $ denotes the norm in $L ^ 2_R: = L ^ 2 (- R, R)$.
\smallskip\\
ii) 
The convergence in $\cE_F^N$ is equivalent to the convergence in every seminorm (\ref{cER}). 

\ed
The space $\cE_F^N$ is not complete, and the convergence in $\cE_F^N$ is equivalent to the convergence in the metric
\begin {equation} \label {metr}
{\rm dist} [Y_1, Y_2] =
\sum_1 ^ \infty 2 ^ {- R} \displaystyle \frac {\Vert Y_1-Y_2 \Vert _ {\cE_c^N, R}}
{1+ \Vert Y_1-Y_2 \Vert _ {\cE_c^N, R}}, \quad Y_1, Y_2 \in \cE_c^N.
\end {equation}
The main result of \cite {K1999} is the following theorem, which is illustrated by the Figure \ref {fig-1}.
Denote by $ \cS $ the set of stationary states $(\psi (x), 0) \in \cE_c^N$, where $\psi (x)$ is the solution of the stationary equation
$$
\psi '' (x) + \chi (x) F (\psi (x)) = 0, \qquad x \in \R.
$$
%%%%%%%%%%%%%%%%%%%%%%%%%%%%%%%%%%
\begin {theorem} \label {t1}
{\rm i)}
Let conditions \eqref {FU-} and \eqref {conf} hold. Then any finite energy solution  $Y (t) = (\psi (t), \pi (t)) \in C (\R, \cE_c^N)$ attracts to  $\cS$:
\begin {equation} \label {Z}
Y (t) \toEFN \cS, \qquad t \to \pm \infty
\end {equation}
in the metric \eqref {metr}. It means that
\begin {equation} \label {conv}
{\rm dist} [Y (t), \cS]: =\inf_ {S \in \cS} {\rm dist} [Y (t), S]\to 0, \qquad t \to \pm \infty.
\end {equation}
{\rm ii)}
Suppose additionally that the function $ F (\psi) $ is real-analytic for $ \psi \in \R ^ N $. Then
$ \cS $ is a discrete subset of $ \cE_c^N $, and for any finite energy solution $ Y (t) = (\psi (t), \pi (t)) \in C (\R, \cE_c^N) $
\begin {equation} \label {Zd}
Y (t) \toEFN S_ \pm \in \cS, \qquad t \to \pm \infty.
\end {equation}
\end {theorem}

%%%%%%%%%%%%%%%%%%%%%%%%%%%%%%%%%%%%%%%%%%%%%%%%%%%%%%
%%%%%%%%%%%%%%%%%%%%%%%%%%%%%%%%%%%%%%%%%%%%%%%%%%%%%%
\begin{figure}[htbp]
\begin{center}
\includegraphics[width=0.7\columnwidth]{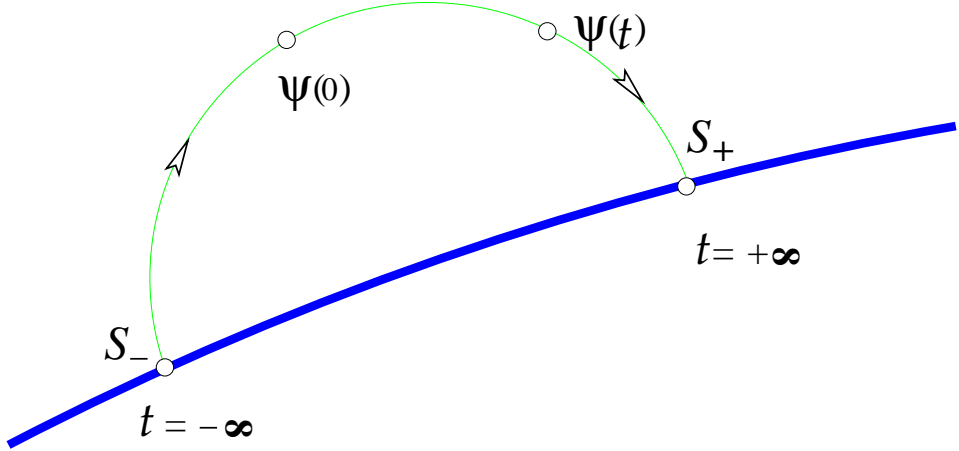}
\caption{Convergence to stationary states}
\label{fig-1}
\end{center}
\end{figure}
%%%%%%%%%%%%%%%%%%%%%%%%%%%%%%%%%%%%%%%%%%%%%%%%%%%%%%%	

\vspace{20mm}
\noindent {\bf Sketch of the proof.}
It suffices to consider only the case $ t \to \infty $. Our  proof of (\ref {Z}) and (\ref {Zd}) in \ci{K1999}
relies on new method of {\it omega-limit trajectories}  which is a development of the method of omega-limit points used in  \cite {K1995b} .
Later on this method played an essential role in the theory of global attractors for $U(1)$-invariant  PDEs \ci{K2003,KK2006,KK2007,KK2010a,KK2008,KK2009,KK2010b,KK2013t,
K2017, KK2019b,KK2019,C2012,C2013}.

First, we note that the finiteness of energy radiated from the segment $ [- a, a] \supset \supp \chi $
implies the finiteness of the ``dissipation integral'' \cite [(6.3)] {K1999}:
$$
\int_0^\infty [| \dot \psi (-a, t) | ^ 2 + | \psi '(- a, t) | ^ 2 + | \dot \psi (a, t) | ^ 2 + | \psi' (a, t) | ^ 2] dt <\infty.
$$
This means, roughly, that
\begin {equation} \label {idis2}
\psi (\pm a, t) \sim C_ \pm, \qquad \psi '(\pm a, t) \sim 0, \qquad t \to \infty.
\end {equation}
More precisely, the functions $ \psi (\pm a, t) $ and $ \psi '(\pm a, t) $ are slowly varying for large times, so
their shifts form {\it omega-compact families}. 
Namely, from an arbitrary sequence $ s_k \to \infty $, one can choose a subsequence $ s_ {k '} \to \infty $ for which
\be\la{bav}
\psi (\pm a, t + s_ {k '}) \to C_ \pm, \qquad k' \to \infty,
\ee
where the constants $C_\pm$ depend on the subsequence, and the convergence holds in $C [0, T]$ {\it for any} $ T> 0 $.
It remains to prove that 
\be\la{ins}
\psi (x, t + s_ {k '}) \to S _ + (x)\in \cS, \qquad k' \to \infty,
\ee
in $C(0,T;H^1[-a,a])$  for any $ T> 0 $. In other words, {\it each omega limit trajectory} is a stationary state.

Roughly speaking we need to justify the correctness of the boundary value problem for a nonlinear differential equation (\ref {neli})
in the half-strip $ -a \le x \le a $, $ t> 0 $, with the Cauchy boundary conditions (\ref {idis2}) on the sides $ x = \pm a $. 
Then the convergence (\ref{bav})  of  boundary values  implies the convergence  (\ref{ins}) of the solution inside
the strip.

Our main idea is to use evident symmetry of  wave equation with respect to interchange of variables
$ x $ and $ t $ with a simultaneous change of the sign of the potential $ U $. However, in this equation with the ``time''  $x $
the condition (\ref {conf}) makes new potential $ -U $ unbounded from below! Consequently,
this dynamics with $ x $ as the time variable is not correct on the interval $ | x | \le a $.
For example, in the case  $ U(\psi) = \psi ^ 4 $,  the equation \eqref{neli} for  solutions of type $\psi(x,t)=\psi(x)$ reads
$ \psi '' (x) -4\psi^3(x) = 0 $. Solutions of this ordinary equation with finite
Cauchy's initial data at $ x = -a $ can become infinite at any point  $ x \in (-a, a) $. 
However,  in our situation a  local correctness  is sufficient due to \textit {a priori bounds},
which follow from energy conservation (\ref {hama}) by the conditions (\ref {FU-}) and (\ref {conf}).

\br
{\rm
i) The energy of the limit states $ S_ \pm $ may be less than the conserved energy of the corresponding solution.
This limit jump of energy is similar to the well-known property of the norm for {\it weak convergence} of a sequence in the Hilbert space.
\smallskip \\
ii) The discreteness of the set $ \cS $ is essential for the asymptotics (\ref {Zd}).
For example, convergence (\ref {Zd}) fails for the solution $ \psi (x, t) = \sin [\log (| x-t | +2)] $ in the case when $ F (\psi) = 0 $ for $ | \psi | \le 1 $.
}
\er

%%%%%%%%%%%%%%%%%%%%%%%%%%%%%%%%%%%%%%%%%%%%%%%%%%%%%%%%%%%%%%%%%%%%%%%%%
%%%%%%%%%%%%%%%%%%%%%%%%%%%%%%%%%%%%%%%%%%%%%%%%%%%%%%%%%%%%%%%%%%%%%%%%%
\subsection {String coupled to nonlinear oscillators}\la{s2.2}
%%%%%%%%%%%%%%%%%%%%%%%%%%%%%%%%%%%%%%%%%%%%%%%%%%%%%%%%%%%%%%%%%%%%%%%%%
{\bf I.} First results on global attraction to stationary states (\ref {Z}) and (\ref {Zd}) were established in \cite {K1991, K1995a, KM2009a}
for the case of point nonlinearity (``Lamb system''):
\begin {equation} \label {w1m}
(1 + m \delta (x)) \ddot \psi (x, t) = \psi '' (x, t) + \delta (x) F (\psi (0, t)), \qquad x \in \mathbb R.
\end {equation}
This equation describes transversal oscillations of a  string with vector displaysments
$ \psi (t) \in \R ^ N $ coupled to an oscillator attached at $ x = 0 $, and
acting on the string with a force $ F (\psi (0, t)) $ orthogonal to the string;
$ m> 0 $ is the mass of a particle attached to the string at the point $ x = 0 $.
For  linear force function $ F (\psi) = - k \psi $, such  system was first considered by H. Lamb \cite {Lamb1900}.

The conserved energy reads
\be\la{hamdal2}
\cH (\psi, \pi,p) = \frac12 \int [| \pi (x) | ^ 2 + | \psi '(x) | ^ 2]  dx+\fr{m p^2}2+U(\psi(0)).
\ee
We denote $ Z: = \{z \in \mathbb R ^ N:  F (z) = 0 \} $. 
Obviously, every 
finite energy
stationary solution of the equation (\ref {w1m}) is a constant function $ \psi_z (x) = z \in  Z $. Let us denote by $ \cS$ the  manifold of all finite energy stationary states,
$$
\cS: = \{S_z = (\psi_z, 0): z \in Z \}.
$$
This set  is discrete in $ \cE_c $, if $ Z $ is discrete in $ \mathbb R ^ N $.
Now the proof of attractions (\ref {Z}) and (\ref {Zd}) relies on the  reduced equation for the oscillator
$$
m \ddot y (t) = F (y (t)) - 2 \dot y (t) + 2 \dot w _ {\rm in} (t), \qquad t> 0,
$$
where $ \dot w _ {\rm in} \in L ^ 2 (0, \infty) $. This equation follows  from the d'Alembert representation for the solution $\psi(x,t)$ at $x>0$ and $x<0$.

In \cite {KM2009a}  a stronger asymptotics in  the global norm of the  Hilbert space $ \cE_c $ are obtained  instead of asymptotics (\ref {Zd}) 
in local seminorms. This is achieved by identifying  the corresponding d'Alembert outgoing and incoming waves. 
In \cite {KM2009b, KM2013} the   asymptotic completeness of the corresponding nonlinear scattering operators  has been proved.
\smallskip \\
{\bf II.}
In \cite {K1995b} we have extended the results \cite {K1991, K1995a} on global attraction to stationary states 
to the case of a string with several oscillators:
$$
\ddot \psi (x, t) = \psi '' (x, t) + \sum_1 ^ M \delta (x-x_k) F_k (\psi (x_k, t)).
$$
This equation reduces to a system of $ M $ ordinary equations with delay. Its study required new approach
relying on a special analysis of  {\it omega-limit points} of trajectories.
\medskip \\
Note that detailed proofs of all results \cite {K1991,K1995a,K1995b,K1999}  are available in the survey \cite {K2000}.

%%%%%%%%%%%%%%%%%%%%%%%%%%%%%%%%%%%%%%%%%%%%%%%%%%%%%%%%%%%%%%%%%
%%%%%%%%%%%%%%%%%%%%%%%%%%%%%%%%%%%%%%%%%%%%%%%%%%%%%%%%%%%%%%%%%
\subsection {Wave-particle system}\la{sWP}
%%%%%%%%%%%%%%%%%%%%%%%%%%%%%%%%%%%%%%%%%%%%%%%%%%%%%%%%%%%%%%%%%
In \cite {KSK1997}, the first result on global attraction  
to stationary states \eqref {ate} is obtained  for three-dimensional
real scalar wave field coupled to a relativistic particle.
The scalar field satisfies  3D wave equation
\begin {equation} \label {w3}
\ddot \psi (x, t) = \Delta \psi (x, t) - \rho (x-q (t)), \qquad x \in \mathbb R ^ 3,
\end {equation}
where $ \rho \in C ^ \infty_0 (\mathbb R ^ 3) $ is a fixed function representing the charge density of the particle,
and $ q (t) \in \mathbb R ^ 3 $ is the particle position. The particle motion obeys the 
Hamilton equation with relativistic kinetic energy $ \sqrt {1 + p ^ 2} $:
\begin {equation} \label {q3}
\dot q (t) = \frac {p (t)} {\sqrt {1 + p ^ 2 (t)}}, \qquad \dot p (t) = - \nabla V (q (t)) - \int \nabla \psi (x, t) \rho (x-q (t)) \, dx,\quad p^2:=p\cdot p.
\end {equation}
Here $ - \nabla V (q) $ is  external force corresponding to  real potential $ V (q) $, and the integral term is a self-force.
Thus,  wave function $ \psi $ is generated by  charged particle, and 
plays the role of a potential acting on the particle, along with the external potential $ V (q) $.

The system \eqref {w3}--\eqref {q3} can formally be represented in Hamiltonian form
\begin {equation} \label {HAM}
\dot \psi = D_ \pi \cH, \quad \dot \pi = -D_ \psi \cH, \quad
\dot q (t) = D_p \cH, \quad \dot p = -D_q \cH
\end {equation}
with Hamiltonian (energy)
\begin {equation} \label {Ham}
\cH (\psi, \pi, q, p) = \displaystyle \frac12 \! \int [| \pi (x) | ^ 2 + | \nabla \psi (x) | ^ 2] \, dx +  \int \psi (x) \rho (x -  q) \, dx +
\sqrt {1 +  p ^ 2} +  V (q).
\end {equation}

By $ \Vert \cdot \Vert $ we denote the norm in the Hilbert space $ L ^ 2: = L ^ 2 (\R ^ 3) $,
and $ \Vert \cdot \Vert_R $ denotes the norm in $ L ^ 2 (B_R) $, where $ B_R $ being the ball
$ | x | \le R $. Let  $\Ho^1:=\Ho^1(\R^3)$ be the completion of the  space $C_0^\infty(\R^3)$ in the norm $\Vert\nabla\psi(x)\Vert$.
%%%%%%%%%%%%%%%%
\begin {definition} \la {dcE}
i) $ \cE: = \Ho^ 1  \oplus L ^ 2 \oplus \mathbb R ^ 3 \oplus \mathbb R ^ 3 $
 is the Hilbert phase space of tetrads $ (\psi, \pi, q, p) $ with finite norm
$$
\Vert (\psi, \pi, q, p) \Vert _ {\cE} = \Vert \nabla \psi \Vert + \Vert \pi \Vert + | q | + | p |.
$$
ii) $ \cE_ \si $ for $ \si \in \R $ is the space of $ Y = (\psi, \pi, q, p) \in \cE $ with $ \psi \in C ^ 2 (\R ^ 3) $ 
and $ \pi \in C ^ 1 (\R ^ 3) $ satisfying the estimate
\be \la {WP8}
| \nabla \psi (x) | + | \pi (x) | + | x | (| \nabla \nabla \psi (x) | + | \nabla \pi (x) |)
= {\cal O} (| x | ^ {- \sigma}), \quad | x | \to \infty.
\ee
iii) $ \cE_F $ is the space $ \cE $ with metric of type (\ref {metr}), where the corresponding seminorms are defined as
\begin {equation} \label {cER3}
\Vert (\psi, \pi, q, p) \Vert _ {\cE,R} = \Vert \nabla \psi \Vert_R + \Vert \psi \Vert_R + \Vert \pi \Vert _R+ | q | + | p |.
\end {equation}
\end {definition}
%%%%%%%%%%%%%%%%%%%%%

Obviously,  the energy (\ref {Ham}) is a continuous functional on $ \cE $, and $ \cE_ \si \subset \cE $
for $ \si> 3/2 $. The convergence in $\cE_F$ is equivalent to the convergence in every seminorm (\ref{cER3}).
We assume the  external potential be  confining:
\begin {equation} \label {V}
V (q) \to \infty, \qquad | q | \to \infty.
\end {equation}
In this case the Hamiltonian (\ref {Ham}) is bounded below:
 \be \la {WPHm}
  \inf_ {Y \in {\cal E}} {\cal H} (Y) = V_ {0} + \fr 12 (\rho, \Delta ^ {- 1} \rho),
\ee
where
\be \la {V0}
 V_0: = \inf_ {q \in \R ^ 3} V (q)> - \infty.
 \ee
 The following lemma is proved in \ci [Lemma 2.1] {KSK1997}.

%%%%%%%%%%%%%%%%%%%%%%%%%%%%%%
\bl \la {lex}
 Let   $ V (q) \in C ^ 2 (\mathbb R ^ 3) $  satisfies the condition (\ref {V0}).
 Then for any initial state $ Y (0) \in \cE $ there exists a unique  finite energy solution
$ Y (t) = (\psi (t), \pi (t), q (t), p (t)) \in C (\mathbb R, \cE) $, and
\\
i) for every $ t \in \R $ the map $ W_t: Y_0 \mapsto Y (t) $ is continuous both on $ {\cal E} $ and on $ {\cal E} _F $;
\\
ii) the energy $ {\cal H} (Y (t)) $ is conserved, i.e.
 \be \la {WPEC}
  {\cal H} (Y (t)) = {\cal H} (Y_0) ~~~ for ~~ t \in \R;
 \ee
iii) a priori estimates hold
 \be \la {apr}
   \sup_ {t \in \R} [\Vert \na \psi (t) \Vert + \Vert \pi (t) \Vert] <\infty, \quad
   \quad  \sup_ {t \in \R} | \dot q (t) | =\ov v<1;
 \ee
iv) if (\ref {V}) holds, then also
 \be \la {aprq}
     \sup_ {t \in \R} | q (t) | =\ov q_0<\infty.
 \ee
\el
%%%%%%%%%%%%%%%%%%%%%%%%%%%%%%%%%%%%%%%%%%%%%%%%%%%%%%%%%%
\br
{\rm
In the case of  point particle $ \rho (x) = \delta (x) $, the system \eqref {w3}--\eqref {q3} is incorrect, 
since in this case any solution of the wave equation \eqref {w3} is singular at the point  $ x = q (t) $,
and, accordingly, the integral in \eqref {q3} is not defined. Energy functional (\ref {Ham}) in this
 case is not bounded from below, because the integral in (\ref {WPHm}) diverges and is equal to $ - \infty $.
 Indeed, in the Fourier transform, this integral has the form
 $$
 (\rho, \Delta ^ {- 1} \rho) = -\int \fr {| \hat \rho (k) |^2} {k ^ 2} dk,
$$
where $ \hat \rho (k) \equiv 1 $. This is the famous ``ultraviolet divergence.''
Thus, the self-energy of  point charge is infinite, that suggested Abraham to introduce the model of an ``extended  electron''
with a continuous charge density $ \rho (x) $.
  }
 \er
%%%%%%%%%%%%%%%%%%%%%%%%%%%%%%%%%%%%%%%%%%%%%%%%%%%%%%%

Denote $ Z = \{q \in \R ^ 3: \nabla V (q) = 0 \} $.
It is easy to verify that  stationary states of the system \eqref {w3}--\eqref {q3}  have the form
$
S_q = (\psi_q, 0, q, 0)$, where $q \in Z$
and  $ \Delta \psi_q (x) = \rho (x-q) $. Therefore, $ \psi_q (x) $ is the Coulomb potential
$$
\psi_q (x): = - \displaystyle \frac1 {4 \pi} \int \frac {\rho (y-q) \, dy} {| x-y |}
$$
Respectively, the set of all stationary states of this system is
$$
\cS: = \{S_q: q \in Z \}.
$$
If the set $ Z \subset \R ^ N $ is discrete, then the set $ \cS $ is also discrete in $ \cE $ and in $ \cE_F $.
Finally, assume that the ``form-factor'' $ \rho $ satisfies Wiener's {\it condition}
\begin {equation} \label {W1}
\hat \rho (k): = \int e ^ {ikx} \rho (x) \, dx \ne 0, \qquad k \in \mathbb R ^ 3.
\end {equation}

%%%%%%%%%%%%%%%%%%%%%%%%%%%%%
\br
{\rm The Wiener condition  means a strong coupling of  scalar wave field $ \psi (x) $ to the particle. 
It is a suitable version of the ``Fermi Golden Rule'' for the system  \eqref {w3}--\eqref {q3}: 
the perturbation $\rho(x-q)$ is not orthogonal to all eigenfunctions of continuum spectrum of the Laplacian $\De$. 
}
\er
%%%%%%%%%%%%%%%%%%%%%%%%%%%%%
 For simplicity of the exposition we assume that
\be\la{C}
\rho\in C_0^\infty(\R^3),
 ~~~~~\rho(x)=0{\rm ~~for~~}|x|\geq R_\rho,~~~~\rho(x)=\rho_r(|x|).
\ee
The main result of \cite {KSK1997} is as follows.

%%%%%%%%%%%%%%%%%%%%%%%%%%
\begin {theorem} \label {t4}
{\rm i)} Let the conditions \eqref {V} and \eqref {W1} hold, and $ \si> 3/2 $. Then for any initial state $ Y (0)=(\psi_0,\pi_0,q_0,p_0) \in \cE_ \si $
the corresponding solution $ Y (t) = (\psi (t), \pi (t), q (t), p (t)) \in C(\R,\cE)$ to the system \eqref {w3}--\eqref {q3} 
attracts to the set of stationary states:
\begin {equation} \label {WP9}
Y (t) \toEF \cS, \qquad t \to \pm \infty,
\end {equation}
where attraction holds in the  metric \eqref {metr} defined by the seminorms \eqref {cER3}.
\medskip \\
{\rm ii)}
Let, additionally,  the set $ Z $ be discrete in $ \mathbb R ^ 3 $. Then
\begin {equation} \label {WP10}
Y (t) \toEF S_ \pm \in \cS, \qquad t \to \pm \infty.
\end {equation}
\end {theorem}
%%%%%%%%%%%%%%%%%%%%%%%%%%%%%%%%%%%%%%
The key point in the proof of this  is the  relaxation of the acceleration
\begin {equation} \label {rel}
\ddot q (t) \to 0, \qquad t \to \pm \infty.
\end {equation}
This relaxation has long been known in classical electrodynamics as ``radiation damping''.
Namely, the  Li\'enard--Wiechert   formulas for retarded potentials suggest that  a particle with a non-zero acceleration 
radiates energy to infinity. The radiation cannot last forever, because the total energy of the solution is finite. 
These arguments result in the conclusion (\ref{rel}) that can be found in any textbook on classical electrodynamics.

However,  rigorous proof is not so obvious and it  was done for the first time in \ci {KSK1997}. 
The proof relies on  calculation of  total energy amount radiated to infinity using the Li\'enard-Wiechert   formulas. 
The central point is the representation  of this amount in the form of a convolution  
and subsequent application of the Wiener Tauberian theorem.

Below we give a streamlined version of this proof for $ t \to + \infty $.

\begin {remark} \label {WC1}
{\rm
i) The condition \eqref {V} is not necessary for relaxation ~ \eqref {rel}.
The relaxation also takes place under the condition \eqref {V0} (see Remark \ref {rVW}).
\medskip \\
ii) The Wiener condition ~ \eqref {W1} also is not necessary for relaxation ~ \eqref {rel}.
For example, ~ \eqref {rel} obviously holds
in the case when
$ V(x) \equiv 0 $ and $\rho (x) \equiv 0 $. More generally, such relaxation also holds
when $ V(x) \equiv 0 $ and the norm $ \Vert \rho \Vert $ is sufficiently small, see  (\ref{rosm}). }

\end {remark}

%%%%%%%%%%%%%%%%%%%%%%%%%%%%%%%%%%%%%%%%%%%%%%%%%%%%%%%%%%%%%%%%%%%%%%
 %%%%%%%%%%%%%%%%%%%%%%%%%%%%%%%%%%%%%%%%%%%%%%%%%%%%%%%
\subsubsection{Li\'enard-Wiechert asymptotics}
%%%%%%%%%%%%%%%%%%%%%%%%%%%%%%%%%%%%%%%%%%%%%%%%%%%%%%%%%%%%%%%%%%%%%
Let us recall  long range asymptotics of the  Li\'enard-Wiechert potentials \ci{KSK1997, KS2000}.
Denote by $\psi_r(x,t)$ the retarded potential 
\be\la{WPretp}
  \psi_r(x,t)=-\fr 1{4\pi}\int\fr{d^3y~\theta(t-|x-y|)}{|x-y|}\rho(y-q(t-|x-y|)),
\ee
and set $ \pi_r(x,t)=\dot\psi_r(x,t)$. Denote $T_r:=\ov q_0+R_\rho$.
%%%%%%%%%%%%%%%%%%%%%%%%%%%%%%%%%%%%%%%%%%%%%%%%%%%%%%%%%%%%%%%%%%%%%%%%%%%%%%%%%%%%%%
\begin{lemma}\la{WPKR} 
The following asymptotics hold 
\begin{equation}\label{WP3.7}
\left\{ \begin{array}{ll}
\pi_r(x,|x|+t)=\ov \pi(\om(x),t)|x|^{-1}+{\cal O}(|x|^{-2})\\\\
\nabla\psi_r(x,|x|+t)=-\om(x)\ov \pi(\om(x),t)|x|^{-1}+{\cal O}(|x|^{-2})
\end{array}\right|,\quad |x|\to\infty
\end{equation}
uniformly in $t\in[T_r,T]$ for any  $T>T_r$. Here $\om(x)=x/|x|$, and $\ov \pi(\om(x),t)$ is given in (\ref{WP3.17}).
\end{lemma}
%%%%%%%%%%%%%%%%%%%%%%%%%%%%%%%%%%%%%%%%%%%%%%%%%%%%%%%%%%%%%%%%%%%%%%%%%%%%%%%%%%%%%%%%%
\begin{proof}
The integrand of (\ref{WPretp}) vanishes for $|y|>T_r$.
Then $|x-y|\le t$ for $t-|x|>T_r$, and  (\ref{WPretp}) implies
\beqn\nonumber
\nabla\psi_r(x,t)
&=&\int \fr {d^3y} {4\pi |x-y|} n\nabla\rho(y-q(t-|x-y|))\cdot\dot q(t-|x-y|)
+{\cal O}(|x|^{-2})\\
\nonumber
&=&-\om(x) \pi_r(x,t)+{\cal O}(|x|^{-2}),\qquad  t-|x|>T_r,
\eeqn
because $\ds n=\fr{x-y}{|x-y|}=\om(x)+{\cal O}(|x|^{-1})$ for bounded $|y|$. 
Hence, it suffices to prove asymptotics (\ref {WP3.7}) for $\pi_r$ only. We have
\be\la{WP3.12}
\pi _r(x,t)=-\int d^3y~\fr 1 {4\pi |x-y|}\nabla\rho(y-q(\tau))\cdot\dot q(\tau),\quad  \tau:=t-|x-y|.
\ee 
Replacing  $t$ by $|x|+t$  in  definition of  $\tau$, we obtain  
$$
\tau=|x|+t-|x-y|=t+\omega(x) \cdot y+{\cal O}(|x|^{-1})
=\ov\tau+{\cal O}(|x|^{-1}),\qquad \ov\tau=t+\om\cdot y,
$$
since
$$
|x|-|x-y|=|x|-\sqrt{|x|^2-2x\cdot y+|y|^2}\sim
|x|\Big(\fr {x\cdot y}{|x|^2}-\fr{|y|^2}{2|x|^2}\Big)
=\omega(x) \cdot y+{\cal O}(|x|^{-1}).
$$
 Hence (\ref{WP3.12}) implies (\ref{WP3.7}) with
\be\la{WP3.17}
\ov \pi(\om,t):=-\fr 1 {4\pi}\int d^3y~ \nabla\rho(y-q(\ov\tau))\cdot\dot q(\ov\tau).
\ee
\end{proof}

%%%%%%%%%%%%%%%%%%%%%%%%%%%%%%%%%%%%%%%%%%%%%%%%%%
%%%%%%%%%%%%%%%%%%%%%%%%%%%%%%%%%%%%%%%%%%%%%%%%%%%%
\subsubsection{Free wave equation}
%%%%%%%%%%%%%%%%%%%%%%%%%%%%%%%%%%%%%%%%%%%%%%%%%%%%
Consider now the solution  $\psi_K(x,t)$ of  free wave equation
with initial conditions  
\be\la{pK}
\psi_K(x,0)=\psi_0(x),\quad \dot\psi_K(x,0)=\pi_0(x),\qquad x\in\R^3.
\ee
The Kirchhoff formula gives
\be\la{kf}
\psi_K(x,t)=\fr 1{4\pi t}\int_{S_t(x)}~d^2y~\pi_0(y)+\fr \pa {\pa t}~\Big(\fr 1{4\pi t}\int_{S_t(x)}~d^2y~\psi_0(y)\Big).
\ee
Here $S_t(x)$ is the sphere $\{y:~|y-x|=t\}$. 
Denote $\pi_K(x,t)=\dot\psi_K(x,t)$. 
%%%%%%%%%%%%%%%%%%%%%%%%%%%%%%%%%
\begin{lemma}\la{WPI00}
Let $Y_0\in{\cal E}_\sigma$. Then for any $R>0$ and  any $T_2>T_1\ge0$
\be\la{WP3.10}
 \int_{R+T_1}^{R+T_2}dt \int_{\pa B_R}d^2x~\Big(|\pi_K(x,t)|^2+|\nabla\psi_K(x,t)|^2\Big) \leq I_0<\infty.
\ee
\end{lemma}
%%%%%%%%%%%%%%%%%%
\begin{proof}
Formula (\ref{kf}) implies 
$$
\nabla\psi_K(x,t)=\fr{t}{4\pi}\int_{S_1}\!\! d^2z~\na\pi_0(x+tz)+
\fr{1}{4\pi}\int_{S_1}\!\! d^2z~\na\psi_0(x+tz)+\fr{t}{4\pi}\int_{S_1}\!\! d^2z~\na_x(\na\psi_0(x+tz)\cdot z).
$$
Here $S_1:=S_1(0)$. From (\ref{WP8}) it follows that
\beqn\nonumber
|\nabla\psi_K(x,t)|&\le&
 C \sum_{s=0}^1 t^{s}\int_{S_1}d^2z~|x+tz|^{-\sigma-1-s}\\
 \nonumber
& =& C \sum_{s=0}^1 \fr{2\pi t^{s-1}}{(\si+s-1)|x|}\Big((t-|x|)^{-\si-s+1}-(t+|x|)^{-\si-s+1}\Big).
\eeqn
Therefore, 
\beqn\nonumber
 \int\limits_{ R+T_1}^{R+T_2}dt \int\limits_{\pa B_R}d^2x |\nabla\psi_k(x,t)|^2 
\!\!\!&\le&\!\!\! C 
 \int\limits_{R+T_1}^{R+T_2} \Big[\fr{(t+R)^{2-2\si}+(t-R)^{2-2\si}}{t^2}+(t-R)^{-2\sigma}\Big]dt\\
\nonumber
\!\!\! &\le& \!\!\! C_1\int\limits_{R+T_1}^{R+T_2}dt \Big[\Big(1+\fr Rt\Big)^{2}+\Big(1-\fr Rt\Big)^{2}+1\Big](t-R)^{-2\sigma}\le I_0<\infty.
 \eeqn
The integral with  $\nabla\pi_K(x,t)$ can be estimated similarly.
\end{proof}

%%%%%%%%%%%%%%%%%%%% Section 3    %%%%%%%%%%%%%%%%%%%%%%%%%%%
 \subsubsection{Scattering of energy to infinity}
%%%%%%%%%%%%%%%%%%%%%%%%%%%%%%%%%%%%%%%%%%%%%%%%%%%%%%%%%%%%%%%%%%%%%%%%%%%%
Now we obtain a  bound on the total energy radiated to infinity which we will represent as a   ``radiation integral''.
This integral has to be bounded a priori by (\ref{apr}).
Indeed,
the energy ${\cal H}_R(t)$ at time $t\in\R$ in the ball $B_R$ is defined by
$$
{\cal H}_R(t)=\fr 12\int_{B_R}d^3x ~\Big(|\pi(x,t)|^2+|\nabla\psi(x,t)|^2\Big)+ \sqrt{1+p^2(t)}+V(q(t))+\int d^3x\, \psi(x,t)\rho(x-q(t))\,.
$$
Consider  the energy $I_R(T_1,T_2)$ radiated from the ball $B_R$ during the time interval $[T_1, T_2]$ with $T_2>T_1>0$:
$$
I_R(T_1,T_2)={\cal H}_R(T_1)-{\cal H}_R(T_2).
$$
This energy is bounded a priori, because 
by (\ref{apr})
the energy
${\cal H}_R(T_1)$ is bounded from above,  while ${\cal H}_R(T_2)$ is bounded from
 below. Thus,
\be\la{WP3.4}
I_R(T_1,T_2)\le I<\infty,
\ee
where  $I$ does not depend on $T_1$, $T_2$  and $R$.  Further, one has
$$
\fr d{dt}{\cal H}_R(t)
=\int_{\pa B_R}d^2 x ~\omega (x)\cdot \pi(x,t)\nabla\psi(x,t),\quad t>R.
$$
Hence, (\ref{WP3.4}) implies
$$
\int_{R+T_1}^{R+T_2}~dt~\int_{\pa B_R}d^2 x ~\omega (x)\cdot \pi(x,t)\nabla\psi(x,t)\le I.
$$
The solution admits the splitting
$\pi=\pi_r+\pi_K$, $\psi=\psi_r+\psi_K$, and hence,
$$
\int_{R+T_1}^{R+T_2}~dt~\int_{\pa B_R}d^2 x ~\omega (x)\cdot
(\pi_r \nabla\psi_r+\pi_K \nabla\psi_r+\pi_r \nabla\psi_K+\pi_K \nabla\psi_K)\leq I.
$$
Lemmas \ref{WPKR} and \ref{WPI00}  together with  the Cauchy-Schwarz inequality imply
$$
\int_{T_r}^T~dt~\int_{S_1}d^2 \om ~|\ov \pi(\om,t)|^2\le I_1+T{\cal O}(R^{-1}),\quad T>T_r,
$$
where $I_1<\infty$ does not depend on $T$ and $R$.
 Taking the limit $R\to\infty$ and then $T\to\infty$ we obtain
 the finiteness of   the energy radiated to infinity:
 \be\la{WP3.1}
\int_{0}^\infty dt\int_{S_1} d^2\om |\ov \pi(\om,t)|^2<\infty.
\ee
%%%%%%%%%%%%%%%%%%%%%%%%%   Section 5    %%%%%%%%%%%%%%%%%%%%%%%%%%%%%%%%%%%%
 \subsubsection{Convolution representation and relaxation of   acceleration and velocity}
%%%%%%%%%%%%%%%%%%%%%%%%%%%%%%%%%%%%%%%%%%%%%%%%%%%%%%%%%%%%%%%%%%%%%%
%%%%%%%%%%%%%%%%%%%%%%%%%%%%%%%%%%%%%%%%%%%%%%%%%%%%%%%%%%%%%%%%%%%%%%
Applying  a partial integration in (\ref{WP3.17}), we obtain 
\beqn\la{WP3.0}
\!\!\!\!\!\!\ov \pi(\om,t)\!\!\!&=&\!\!\! \int d^3y~\nabla\rho(y-q(\ov\tau))\cdot\dot q(\ov\tau)
=\int d^3y~\nabla_y\rho(y-q(\ov\tau))\cdot\dot q(\ov\tau)\fr 1{1-\om\cdot\dot q(\ov\tau)}\nonumber\\
\!\!\!&=&\!\!\!-\int d^3y~\rho(y-q(\ov\tau))\fr{\pa}{\pa y_\al}\fr{\dot q_\al(\ov\tau)}{1-\om\cdot\dot q(\ov\tau)}
=\fr 1 {4\pi}\int d^3y~\rho(y-q(\ov\tau))\frac{\om\cdot\ddot q(\tau)}{(1-\om\cdot\dot q(\tau))^2}. 
\eeqn
The function
 $\ov\pi(\omega,t)$ is globally Lipschitz continuous in $\omega$ and $t$ due to  (\ref{apr})
 Hence, (\ref{WP3.1}) implies 
 \be\la{WP4.3}
 \lim_{t\to\infty} \ov\pi(\omega,t)=0
 \ee
 uniformly in $\omega\in S_1$.
 Denote $r(t)=\omega\cdot q(t)$, $s=\omega\cdot y$, $\tilde\rho(q_3)=\ds\int dq_1dq_2\rho(q_1,q_2,q_3)$
 and  decompose  the $y$-integration  in (\ref{WP3.0}) along and transversal to $\omega$. Then
 \beqn
\ov\pi(\omega,t) & = & \int ds \,\tilde\rho(s-r(t+s))\,\frac{\ddot r(t+s)}
 {{(1-\dot r(t+s))}^2}\nonumber\\
 \nonumber& = &
 \int d\tau\, \tilde\rho(t-(\tau-r(\tau)))\,
 \frac{\ddot r(\tau)}{{(1-\dot r(\tau))}^2}
 = \int d\theta\, \tilde\rho(t-\theta) g_\omega(\theta)=\tilde\rho*g_\omega(t).
\eeqn
 Here  $\theta=\theta(\tau)=\tau-r(\tau)$ is a nondegenerate diffeomorphism 
 of $\R$ since $\dot r\leq \ov r<1$ due to (\ref{apr}), and 
 \be\la{Larmor}
 g_{\omega}(\theta)=\fr {\ddot r(\tau(\theta))}
{(1-\dot r(\tau(\theta)))^3}.
 \ee
 Let us extend $q(t)=0$  for $t<0$. Then $\tilde\rho*g_\omega\,(t)$ is defined for all $t$, and coincides  
 with $\ov\pi(\omega,t)$ for sufficiently  large $t$. Hence, (\ref{WP4.3}) reads as a convolution limit
 \be\la{WP4.6}
 \lim_{t\to\infty} \tilde\rho*g_{\omega}(t)=0.
 \ee
 Moreover, $g^\prime_\omega(\theta)$ is bounded by (\ref{apr}). 
 Therefore, (\ref{WP4.6}) and the Wiener condition (\ref{W1}) imply  \be\la{WP4.7}
 \lim_{\theta\to\infty} g_{\omega}(\theta)=0,\qquad \om\in S_1
 \ee
 by Pitt's extension of the Wiener Tauberian Theorem, cf.~\cite[Thm.~9.7(b)]{Rudin}. Hence, (\ref{Larmor}) implies
 \be\la{WP4.7'}
 \lim_{t\to\infty} \ddot q(t)=0.
  \ee
since $\theta(t)\to\infty$ as $t\to\infty$.
Finally, 
\be\la{WP4.7''}
\lim_{t\to\infty}\dot q(t)=0,
\ee
 since   $|q(t)|\le \ov q_0$ due to (\ref{apr}).
 \begin{remark}\la{rVW}
 %%%%%%%%%%%%%%%%%%%%%%%%%%%%%%%%%%%%%%%%%%%%%%%%%%%%%%%%%%%%%%%%%%%%%%%%%%%%%%%%%%
(i) We have used condition (\ref{V}) in the proof of (\ref{WP3.4}).
However, (\ref{V0}) at this point is also sufficient. Hence, the relaxation  (\ref{WP4.7'})  holds also under condition (\ref{V0}).
\\
(ii) For  point charge $\rho(x)=\delta(x)$, (\ref{WP4.6}) implies (\ref{WP4.7}) directly.
\\
(iii) Condition (\ref{W1}) is necessary for the implication (\ref{WP4.7})$\Rightarrow$(\ref{WP4.7'}). 
Indeed, if (\ref{W1}) is violated, then $\hat\rho_a(\xi)=0$ for some $\xi\in\R$, and with the choice
 $g(\theta)=\exp(i\xi \theta)$ we have $\rho_a*g(t)\equiv 0$ whereas $g$ does not decay to zero.
 \end{remark}

%%%%%%%%%%%%%%%%%%%%%%%%%%%%%%%%%%%%%%%%%%%%%%%%%%%%%%%%%%%%%%%%%%%%%%
%%%%%%%%%%%%%%%%%%%%%%%%%%%%%%%%%%%%%%%%%%%%%%%%%%%%%%%%%%%%%%%%%%%%%%
\subsubsection{A compact attracting set}
%%%%%%%%%%%%%%%%%%%%%%%%%%%%%%%%%%%%%%%%%%%%%%%%%%%%%%%%%%%%%%%%%%%%%
Here we show that the set
\be\la{cA}
{\cal A} = \{S_q:~q\in\R^3,~ |q|\le \ov q_0\}
\ee 
is an attracting subset.
It is compact in ${\cal E}_F$ since ${\cal A}$ is homeomorphic to a closed ball in $\R^3$.
%%%%%%%%%%%%%%%%%%%%%%%%%%%%%%%%%
\begin{lemma}\la{WPCAT}
The following attraction holds,
\be\la{WPcA} 
Y(t)\tocEF {\cal A},\quad t\longrightarrow \pm\infty.
\ee
\end{lemma}
%%%%%%%%%%%%%%%%%%%%%%%%%%%%%%%
\begin{proof}
We need to check that for every $R>0$ 
\begin{eqnarray}\la{WP5.1}
{\rm dist}_R(Y(t),{\cal A})\!=\!|p(t)|\!+\!\Vert\pi(t)\Vert_R\!+\!\!\inf_{S_q\in{\cal A}}\!\Big(|q(t)\!-\!q|\!+\!\Vert\psi(t)\!-\!\psi_{q}\Vert_R
\!+\!\Vert\nabla(\psi(t)\!-\!\psi_{q})\Vert_R\Big)
\to 0,~~ t\to\infty.
\end{eqnarray}
We estimate each summand separately.\\
i) $|p(t)|\to 0$ as $t\to\infty$ by  (\ref{WP4.7'}).\\
ii) $\inf\limits_{|q|\le\ov q_0}|q(t)-q|=0$ for any $t\in\R$ by  (\ref{apr}).\\
iii) (\ref{WPretp}) implies for $t>R+T_r$ and $|x|<R$
\[
 |\pi_r(x,t)|\le~C\max\limits_{t-R-T_r\le \tau\le t}|\dot q(\tau)|\int_{|y|<T_r} d^3y \,\fr 1{|x-y|} |\nabla\rho (y-q(t-|x-y|))|\,.
 \]
 The integral in the RHS is bounded uniformly in $t>R+T_r$ and $x\in B_R$. Hence,
 $\Vert\pi_r(t)\Vert_{R}\to 0$ as $t\to\infty$ by  (\ref{WP4.7''}). Then  also $\Vert\pi(t)\Vert_{R}\to 0$.
 \\
iv) Obviously,  we can replace  $q$ with  $q(t)$ in the last summand in (\ref{WP5.1}). Then for $t>R+T_r$ and $|x|<R$, one has
$$
 \psi_r(x,t)-\psi_{q(t)}(x) =-\int_{|y|<T_r} d^3y\fr 1{4\pi|x-y|}\Big(\rho(y-q(t-|x-y|))-\rho(y-q(t))\Big)
$$
by  (\ref{WPretp}). Moreover, $\rho(y-q(t-|x-y|))-\rho(y-q(t))\to 0$ as $t\to\infty$ uniformly in $x\in B_R$ due to (\ref{WP4.7''}).
 Hence, $\Vert\psi_r(t)-\psi_{q(t)}\Vert_R\to 0$ as $t\to\infty$. Then also  $\Vert\psi(t)-\psi_{q(t)}\Vert_{R}\to 0$.
 Finally,  $\Vert\nabla(\psi(t)- \psi_{q(t)})\Vert_R$ can be estimated in a similar way.
 \end{proof}

%%%%%%%%%%%%%%%%%%%%%%%%%%%%%%%%%%%%%%%%%%%%%%%%%%%%%%%%%%%%%%%%%%%%%%%%%%%%%%%%%%%%%%%%%%%%%%%%%%%%%%%%%%%%%%%%%%%%%%%%%%%%%%%%%%%%%%%%%%%
\subsubsection{Global attraction}
%%%%%%%%%%%%%%%%%%%%%%%%%%%%%%%%%%%%%%%%%%%%%%%%%%%%%%%%%%%%%%%%%%%%%
Now we complete the proof of Theorem \ref{t4}.\\
{\it  i)}
Let $ Y (t) \in C (\R, \cE) $ be any  finite energy solution to the system \eqref {w3}--\eqref {q3}.
If  the attraction (\ref {WP9}) does not hold,  there is a sequence
$ t_k \to \infty $ for which
\be \la {tk}
{\rm dist} (Y (t_k), \cS) \ge \de> 0, \qquad k = 1,2, \dots
\ee
 Since $\cA$ is a compact set in $ {\cal E}_F $, (\ref{WPcA}) implies that
\be \la {tk2}
Y (t_ {k '}) \toEF \ov Y\in\cA, \qquad k' \to \infty
\ee
 for some subsequence $ k '\to \infty $.
It remains to check that $ \ov Y = S_ {q _ *} \in \cS $ with some $ q _ * \in Z $, since this contradicts (\ref {tk}).

 
First, $ \ov Y = S_q  $ with some
$ | q | \le \ov q_0 $
 by the definition (\ref{cA}). Similarly,
 by the continuity of the map $ W_t $ in $ {\cal E} _F $,
\be \la {tk3}
W_tY (t_ {k '}) = Y (t_ {k'} + t) \toEF W_t \ov Y = S_ {Q (t)}, \qquad k '\to \infty,
\ee
where $ Q (\cdot) \in C ^ 2 (\R, \cE) $, since $  W_t \ov Y\in C (\R, \cE) $ is a solution
to the system \eqref {w3}--\eqref {q3}.
Finally,
for $ S_ {Q (t)} $ to be a solution to the system (\ref {w3})--(\ref {q3}), there must be $ \dot Q (t) \equiv 0 $. Therefore, $ Q (t) \equiv q _ * \in Z $ and
$ \ov Y = S_ {q _ *} \in \cal S $.
\medskip\\
{\it  ii)} If  the set $Z$  is discrete in $\R^3$,  then solitary  manifold  $\cal S$ is discrete in ${\cal  E}_F$.
\hfill$\Box$

%%%%%%%%%%%%%%%%%%%%%%%%%%%%%%%%%%%%%%%%%%%%%%%%%%%%%%%%%%%%%%%%%%%%%
%%%%%%%%%%%%%%%%%%%%%%%%%%%%%%%%%%%%%%%%%%%%%%%%%%%%%%%%%%%%%%%%%%%%%
\subsection {Maxwell--Lorentz equations: radiation damping}\la{sML}
%%%%%%%%%%%%%%%%%%%%%%%%%%%%%%%%%%%%%%%%%%%%%%%%%%%%%%%%%%%%%%%%%%%%%
In \cite {KS2000}  global attraction to stationary states
similar to
(\ref {WP9}), (\ref {WP10}) was extended to the
Maxwell--Lorentz equations with  charged relativistic particle:
\begin {equation} \label {ML}
\!\!\!\!\!\!\!\!\left \{\! \!\begin {array} {l}
\dot E (x, t) \! = \! \rot B (x, t) - \dot q \rho (x \! - \! q), ~\dot B (x, t) \! = \! - \rot E (x, t), ~ \dv E (x, t) \! = \! \rho (x \! - \! q), ~ \dv B (x, t) \! = \! 0
\\
\\
\dot q (t) \! = \! \displaystyle \frac {p (t)} {\sqrt {1 \! + p ^ 2 (t)}}, ~~
\dot p (t) \! = \! \displaystyle \int [E (x, t) \! + E ^ {\rm ext} (x, t) \! + \dot q (t) \wedge (B (x, t) \! + B ^ {\rm ext} (x, t))] \rho (x \! - q (t) ) \, dx
\end {array} \right |.
\end {equation}
Here
$ \rho (x-q) $ is the particle charge density, $ \dot q \rho (x-q) $ is the corresponding current density,
and $ E ^ {\rm ext} = - \nabla \phi ^ {\rm ext} (x) $ and $ B ^ {\rm ext} = - \rot A ^ {\rm ext} (x) $ are external static Maxwell  fields. 
Similarly to \eqref {V}, we assume that {\it effective scalar potential} is confining:
\begin {equation} \label {VML}
V (q): = \int \phi ^ {\rm ext} (x) \rho (x-q) \, dx \to \infty, \qquad | q | \to  \infty.
\end {equation}
This system describes classical electrodynamics with “extended electron” introduced by Abraham \cite {A1902, A1905}.
In the case of a point electron, when $ \rho (x) = \delta (x) $, such system is not well defined. 
Indeed, in this case, any solutions $ E (x, t) $ and $ B (x, t) $ of Maxwell's equations (the first line of \eqref {ML}) 
are singular for $ x = q (t) $, and, accordingly, the integral in the last equation \eqref {ML} does not exist.

This system may be formally
presented in Hamiltonian form, if the fields are expressed in terms of potentials
$ E (x, t) = - \nabla \phi (x, t) - \dot A (x, t) $, $ B (x, t) = - \rot A (x, t) $, \ci {IKM2004}.
The corresponding Hamilton functional reads
\begin {equation} \label {HAMml}
\cH = \frac12 [\langle E, E \rangle + \langle B, B \rangle] + V (q) + \sqrt {1 + p ^ 2} =
\frac12 \int [E ^ 2 (x) + B ^ 2 (x)] \, dx + V (q) + \sqrt {1 + p ^ 2}.
\end {equation}
The Hilbert phase space  of finite energy states is defined as $ \cE: = L ^ 2  \oplus L ^ 2  \oplus \R ^ 3 \oplus \R ^ 3 $.
Under the condition (\ref {VML}) a solution $ Y (t) = (E (x, t), B (x, t), q (t), p (t)) \in C (\R, \cE) $ of finite  energy  
exists and is unique for any initial state $ Y (0) \in \cE $.

The Hamiltonian (\ref {HAMml}) is conserved along  solutions, what provides \textit {a priori estimates}, 
which play an important role in proving an attraction of the type (\ref {WP9}), (\ref {WP10}) in \cite {KS2000}.
Key role in the proof of  relaxation of an acceleration plays again (\ref {rel}), 
which is derived by a suitable generalization of our  methods \cite {KSK1997}: 
the expression of energy  radiated  to infinity  via Li\'enard--Wiechert  retarded potentials, its
representation in the form of a convolution and the use of Wiener's Tauberian theorem.
\medskip

In classical electrodynamics the relaxation \eqref {rel} known as {\bf radiation damping}.
It is traditionally derived from the Larmor and Li\'enard formulas
for  radiation power of a point particle  (see formulas (14.22) and  (14.24) of \cite{Jackson}),
but this approach ignores field feedback although it plays the key role in the  relaxation.
The main problem is that this reverse field reaction  for point particles is infinite.
A rigorous  sense of these classical calculations  was first found in \cite {KSK1997, KS2000}
for the Abraham model of  ``extended electron'' under the  Wiener condition \eqref {W1}.
Details can be found in \cite {S2004}.

%%%%%%%%%%%%%%%%%%%%%%%%%%%%%%%%%%%%%%%%%%%%%%%%%%%%%%%%%%%%%%%%
%%%%%%%%%%%%%%%%%%%%%%%%%%%%%%%%%%%%%%%%%%%%%%%%%%%%%%%%%%%%%%%%
\subsection{Wave equation with concentrated nonlinearities}\la{sCN}
%%%%%%%%%%%%%%%%%%%%%%%%%%%%%%%%%%%%%%%%%%%%%%%%%%%%%%%%%%%%%%%%

Here we prove  
the result of \ci{NP} on
global attraction to solitary manifold for 3D wave equation with point coupling to an $\mathbf{U}(1)$-invariant  nonlinear oscillator.  
This goal  is inspired by fundamental mathematical  problem of an  interaction of point particles with the fields. 

Point interaction models were  first considered since 1933  in the papers of Wigner, Bethe and Peierls, Fermi and others
(see \cite{AF2018} for a detailed survey) and of Dirac  \cite{Dirac1938}. Rigorous mathematical results were obtained  since 1960 
by   Zeldovich, Berezin, Faddev,  Cornish, Yafaev,  Zeidler and others \cite{BF1961, Cornish1965, GKZ1998, Y1992, Z1960},  
and since 2000 by Noja, Posilicano, Yafaev and others \cite{NP, Y2017, ANO}.
\smallskip

We consider real wave field $\psi(x,t)$ coupled to a nonlinear oscillator 
\begin{equation}\label{W}
\left\{\begin{array}{c}
\ddot \psi(x,t)=\Delta\psi(x,t)+\zeta(t)\delta(x)\\\\
\lim\limits_{x\to 0}(\psi(x,t)-\zeta(t)G(x))=F(\zeta(t))
\end{array}\right|\quad x\in\R^3,\quad t\in\R,
\end{equation} 
where $G(x)=\ds\frac{1}{4\pi|x|}$ is the Green's function of the operator $-\Delta$ in $\R^3$. 
Nonlinear function  $F(\zeta)$ admits a potential:
\begin{equation}\label{FU}
 F(\zeta)=U'(\zeta), \quad\zeta\in\R,\quad U\in C^2(\R).
\end{equation} 
We assume that the potential is confining, i.e., 
\begin{equation}\label{bound-below}
U(\zeta)\to\infty, \quad \zeta\to\pm\infty.
\end{equation}
Еhe system (\ref{W}) admits stationary solutions  $\psi_q=qG(x)\in L^2_{loc}(\R^3)$, where $q\in Q:=\{q\in R: F(q)=0\}$.
We assume that the set $Q$ is nonempty and does not contain intervals, i.e.,
\begin{equation}\label{ab}
[a,b]\not\subset Q
\end{equation}
for any $a<b$.

As before, $ \Vert \cdot \Vert $ and $ \Vert \cdot \Vert_R $ denote the norms in $ L ^ 2= L ^ 2 (\R ^ 3)$
and  in $ L ^ 2 (B_R) $ respectively, and
$\Ho^1=\Ho^1(\R^3)$ is the completion of the  space $C_0^\infty(\R^3)$ in the norm $\Vert\nabla\psi(x)\Vert$.
Denote
\[
\Ho^2=\Ho^2(\R^3):=\{ f\in \Ho^1,~~\Delta f\in L^2\},\quad t\in\R.
\]
We  define the function sets
$$
D=\{\psi\in L^2:\psi(x)=\psi_{reg}(x)+\zeta G(x),~~\psi_{reg}\in \Ho^2,~~
 \zeta\in\R,~~\lim\limits_{x\to 0}\psi_{reg}(x)=F(\zeta)\}
$$
and
\[
\dot D=\{\pi\in L^2(\R^3):\pi(x)=\pi_{reg}(x)+\eta G(x),
~~\pi_{reg}\in \Ho^1, ~~\eta\in\R\}.
\]
Obviously, $D\subset\dot D$.
%%%%%%%%%%%%%%%%%%%%%%%%%%%%%%%%%%%%%%%%%%%%%%%%%%%%%%%%%%%%%%%%%%%%%%
\begin{definition}\label{cDdef}
${\cal D}$ is the Hilbert manifold  of  states $\Psi=(\psi,\pi)\in D\times\dot D$. 
\end{definition}
 First, we  prove  global well-posedness for the system (\ref{W}).
%%%%%%%%%%%%%%%%%%%%%%%%%%%%%%%%%%%%%%%%%%%%%%%%%%%%%%%%%%%%%
\begin{theorem}\label{theorem-well-posedness}
Let conditions (\ref{FU}) and (\ref{bound-below}) hold. Then 
\begin{enumerate}
\item
For every initial data $\Psi_0=(\psi_0,\pi_0)\in {\cal D}$  the system
(\ref{W}) has a unique  solution 
$\Psi(t)=(\psi(t),\dot\psi(t))\in C(\R,{\cal D})$.
\item
The energy is conserved:
\be\la{enerc}
{\cal H}(\Psi(t)):=\frac 12 \Big(\Vert\dot\psi(t)\Vert^2+\Vert\nabla\psi_{reg}(t)\Vert^2\Big)+U(\zeta(t))=\const, \quad t\in\R.
\ee
\item
The following a priori bound holds
\begin{equation}\label{zeta-bound}
|\zeta(t)|\le C(\Psi_0), \quad t\in \R. 
\end{equation}
\end{enumerate}
\end{theorem}
%%%%%%%%%%%%%%%%%%%%%%%%%%%%%%%%%%%%%%%%%%%%%%%%%%%
\begin{proof}
It suffices to prove the theorem  for $t\ge 0$.
\smallskip

{\it Step i)} 
First we consider  free wave equation with initial data from $\cal D$:
\begin{equation}\label{CP1}
\ddot{\psi}_f(x,t) = \Delta\psi_f(x,t),
\quad (\psi_f(0), \dot\psi_f(0) ) =( \psi_0,  \pi_0)=(\psi_{0,reg},\pi_{0,reg})+(\zeta_0 G,\eta_0 G)\in{\cal D},
\end{equation}
where  $(\psi_{0,reg},\pi_{0,reg})\in \Ho^2\oplus\Ho^1$.
%%%%%%%%%%%%%%%%%%%%%%%%%%%%%%%%%%%%%
 \begin{lemma}\label{wdl}
There exists a unique solution $\psi_f(t)\in C([0;\infty), L^2_{loc})$ to (\ref{CP1}). Moreover, for any $t>0$
there exists the limit
$$
\lambda(t):=\lim\limits_{x\to 0}\psi_f(x,t)\in C[0,\infty),
$$ 
and 
\begin{equation}\label{dot_lam} 
\dot\lambda(t)\in L^2_{loc}[0,\infty).
\end{equation} 
\end{lemma}
 %%%%%%%%%%%%%%%%%%%%%%%%%%%%%%%%%%%%%%
 \begin{proof}
We split $\psi_f(x,t)$ as 
 \[
 \psi_f(x,t)=\psi_{f,reg}(x,t)+g(x,t),
 \]
 where $\psi_{f,reg}$ and $g$ are  solutions to  free wave equation with  initial data 
 $(\psi_{0,reg},\pi_{0,reg})$ and $(\zeta_0 G,\eta_0 G)$, respectively.
 First, $\psi_{f,reg}\in C([0,\infty),\Ho^2)$ by the energy conservation.  
 Hence,   $\lim\limits_{x\to 0}\psi_{f,reg}(x,t)$ exists for any $t\ge 0$ since $\Ho^2(\R^3)\subset C(\R^3)$. 
 
 Let us obtain an explicit formula for $g$. Note, that the function
 $h(x,t)=g(x,t)-(\zeta_0+\eta_0 t)G(x)$ satisfies
 \begin{equation}\label{Smir}
 \ddot h(x,t)=\Delta h(x,t)-(\zeta_0+\eta_0 t)\delta(x),\quad h(x,0)=0, ~~\dot h(x,0)=0.
 \end{equation} 
 The unique solution to (\ref{Smir}) is   spherical wave : 
\begin{equation}\la{Smir1}
h(x,t)=-\frac{\theta(t-|x|)}{4\pi|x|}(\zeta_0+\eta_0(t-|x|)),\quad t\ge 0.
\end{equation}
Here $\theta$  is the Heaviside function.  Hence,
$$
g(x,t)=h(x,t)+(\zeta_0+\eta_0 t)G(x)=-\frac{\theta(t-|x|)(\zeta_0+\eta_0(t-|x|))}{4\pi|x|}+\frac{\zeta_0+\eta_0 t}{4\pi|x|}\in C([0,\infty),L^2_{loc}(\R^3)),
$$
and then
$$
\lim\limits_{x\to 0}g(x,t)=\fr{\eta_0}{4\pi},\quad t>0.
$$
Finally, $\dot\psi_{f,reg}(0,t)\in L^2_{loc}([0,\infty))$ by \cite[Lemma 3.4]{NP}. Hence, (\ref{dot_lam}) follows.
\end{proof}
{\it Step ii)} Now we prove  local well-posedness. We modify the nonlinearity $F$ so that it becomes Lipschitz-continuous.
Define
$$
\Lambda(\Psi_0)=\sup\{|\zeta|: \zeta\in\R,\, U(\zeta)\le {\cal H}(\Psi_0)\}.
$$
We may pick a modified potential function $\tilde U(\zeta)\in C^2(\R)$, so that
\begin{equation}\label{Lambda1}
\left\{\begin{array}{ll}
\tilde U(\zeta)= U(\zeta),\quad |\zeta|\le\Lambda(\Psi_0),\\\\
\tilde U(\zeta)>{\cal H}(\Psi_0),\quad |\zeta|>\Lambda(\Psi_0),
\end{array}\right.
\end{equation}
and  the function  $\tilde F(\zeta)=\tilde U'(\zeta)$ is Lipschitz-continuous:
$$
|\tilde F(\zeta_1)-\tilde F(\zeta_2)|\le C|\zeta_1-\zeta_2|,\quad\zeta_1,\zeta_2\in\R.
$$
The following lemma  is trivial.
%%%%%%%%%%%%%%%%%%%%%%%%%%%%%%
\begin{lemma}\label{LLWP}
For small $\tau>0$   the Cauchy problem
\begin{equation}\label{delay1}
\frac {1}{4\pi}\dot\zeta(t)+\tilde F(\zeta(t))=\lambda(t),\quad \zeta(0)=\zeta_{0}
\end{equation}
has a unique solution $\zeta\in C^1([0,\tau])$.
\end{lemma}
Denote 
$$
\psi_S(t,x):=\frac{\theta(t-|x|)}{4\pi|x|}\zeta(t-|x|), \quad t\in [0,\tau],
$$
with $\zeta$ from Lemma \ref{LLWP}.
%%%%%%%%%%%%%%%%%%%%%%%%
\begin{lemma}\label{TLWP}
The function $\psi(x,t):= \psi_f(x,t)+\psi_S(x,t)$ is a unique  solution to the system
\begin{equation}\label{CP}
\left\{\begin{array}{c}
\ddot \psi(x,t)=\Delta\psi(x,t)+\zeta(t)\delta(x)\\\\
\lim\limits_{x\to 0}(\psi(x,t)-\zeta(t)G(x))=\tilde F(\zeta(t))\\\\
\psi(x,0)=\psi_0(x),~~\dot\psi(x,0)=\pi_0(x)
\end{array}\right|\quad x\in\R^3,\quad t\in [0,\tau],
\end{equation}
satisfying the condition
\begin{equation}\label{CDD}
(\psi(t),\dot\psi(t))\in {\cal D},\quad t\in [0,\tau].
\end{equation}
\end{lemma}
%%%%%%%%%%%%%%%%%%%%%%%%%%%%%%%%%%%%%%%%%%%%%%%%%%%%%%%%%%%%%%%%%%%%%%%%%%%%%%%%%%%%%%%%
\begin{proof}
Initial conditions of  (\ref{CP}) follow from (\ref{CP1}).
Further, 
$$
\lim_{x \to 0}\, ( \psi(t,x)\!-\!\zeta(t) G(x))=\lambda(t)+
\lim_{x\to 0}\Big(\frac{\theta(t-|x|)\zeta(t-|x|)}{4\pi|x|}-\frac{\zeta(t)}{4\pi|x|}\Big)
=\lambda(t)-\frac{1}{4\pi}\dot\zeta(t)=\tilde F(\zeta(t)).
$$
Thus, the second equation of (\ref{CP}) is satisfied.  At last,
$$
\ddot\psi=\ddot\psi_f+\ddot\psi_S=\Delta\psi_f
+\Delta\psi_S+\zeta\delta=\Delta\psi+\zeta\delta
$$
and $\psi$ solves the first equation of (\ref{CP}) then.  
 
It remains to check (\ref{CDD}).
Note, that  the function  
$\varphi_{reg}(x,t)=\psi(x,t)-\zeta(t) G_1(x)=\psi_{reg}(x,t)+\zeta(t)(G(x)-G_1(x))$, where $G_1(x)=G(x)e^{-|x|}$,
satisfies
\[
\ddot\varphi_{reg}(x,t)=\Delta\varphi_{reg}(x,t)+(\zeta(t)-\ddot\zeta(t)) G_1(x)
\]
with  initial data from $H^2\oplus H^1$. Moreover, (\ref{dot_lam}) and  (\ref{delay1}) imply that
$\ddot\zeta\in L^2([0,\tau])$. Hence, 
\[
(\varphi_{reg}(x,t),\dot\varphi_{reg}(x,t))\in  H^2\oplus H^1,\quad t\in [0,\tau] 
\] 
by \cite[Lemma 3.2]{NP}. Therefore,
\[
\psi_{reg}(x,t)=\psi(x,t)-\zeta(t) G(x)=\varphi_{reg}(x,t)+\zeta(t)(G_1(x)-G(x))
\]
satisfies
$(\psi_{reg}(t),\dot\psi_{reg}(t))\in \Ho^2\oplus\Ho^1$, $t\in [0,\tau]$,
and (\ref{CDD}) holds then.
%%%%%

It remains to prove the uniqueness. 
Suppose now that there exists another solution
$\tilde\psi=\tilde\psi_{reg}+\tilde\zeta G$ to the system (\ref{CP}), 
with $(\tilde\psi,\dot {\tilde\psi})\in {\cal D}$. 
Then, by reversing the above argument, the second equation of (\ref{CP})
 implies that $\tilde\zeta$ solves the Cauchy problem (\ref{delay1}). 
 The uniqueness of the solution of (\ref{delay1}) implies that $\tilde\zeta=\zeta$. Then, defining
$$
\psi_S(t,x):=\frac{\theta(t-|x|)}{4\pi|x|}\zeta(t-|x|), \quad t\in [0,\tau],
$$
for $\tilde\psi_f=\tilde\psi-\psi_S$ one obtains
$$
\ddot{\tilde\psi}_f=\ddot{\tilde\psi}-\ddot\psi_S=\Delta\tilde\psi_{reg}-(\Delta\psi_S+\zeta\delta)=
\Delta(\tilde\psi_{reg}-(\psi_S-\zeta G))=\Delta\tilde\psi_f\,,
$$
i.e. $\tilde\psi_f$ solves the Cauchy problem (\ref{CP1}). 
Hence, $\tilde\psi_f=\psi_f$ by the uniqueness of the solution to (\ref{CP1}), and hence, $\tilde\psi=\psi$.
\end{proof}
According to \cite[Lemma 3.7]{NP}
\begin{equation}\label{cHFT}
{\cal H}_{\tilde F}(\Psi(t))=\Vert\dot\psi(t)\Vert^2+\Vert\nabla\psi_{reg}(t)\Vert^2
+\tilde U(\zeta(t))=const,\quad t\in [0,\tau].
\end{equation}
%%%%%%%%%%%%%%%%%%%%%%%%%%%%%
{\it Step iii)}  Now we are able to prove  the globall well-posedness.
%%%%%%%%%%%%%%%%%%%%%%%%%%%%%%%%%%%%%%%%%%%
First, note that
\begin{equation}\label{UtU}
\tilde U(\zeta(t))=U(\zeta(t)), \quad t\in [0,\tau].
\end{equation}
%%%%%%%%%%%%%%%%%%%%%%%%%%%%%%%%%%%%%%%%%%%%
Indeed, 
${\cal H}_{F}(\Psi_0)\ge  U(\zeta_{0})$ by the definition of energy in (\ref{enerc}).
Therefore, $|\zeta_0|\le\Lambda(\Psi_0)$, and then  $\tilde U(\zeta_0)=U(\zeta_0)$, 
${\cal H}_{\tilde F}(\Psi_0)={\cal H}_{F}(\Psi_0)$.
Further, 
\[
{\cal H}_{F}(\Psi_0)={\cal H}_{\tilde F}(\Psi(t))\ge \tilde U(\zeta(t)),\quad t\in [0,\tau],
\]
and (\ref{Lambda1}) implies that
\begin{equation}\label{zeta_bound}
|\zeta(t)|\le\Lambda(\Psi_0),\quad t\in [0,\tau].
\end{equation}

Now we can replace $\tilde F$ by $F$ in  Lemma \ref{TLWP} and in (\ref{cHFT}).
The solution $\Psi(t)=(\psi(t),\dot\psi(t))\in {\cal D}$ constructed in Lemma  \ref{TLWP}
exists for $0\le t\le\tau$, where the time span $\tau$ in Lemma \ref{LLWP} depends only on $\Lambda(\Psi_0)$.
Hence, the bound (\ref{zeta_bound}) at $t=\tau$ allows us to extend the solution $\Psi$ to the time
interval $[\tau, 2\tau]$. We proceed by induction to obtain the solution for all $t\ge 0$.
Theorem \ref {theorem-well-posedness} is proved.
\end{proof}
 %%%%%%%%%%%%%%%%%%%%%%%%%%%%%%%%%%%%%%%%%%%%%%%%%%%%%%%%%%%%%%
The main result of \ci{NP} is as follows.
%%%%%%%%%%%%%%%%%%%%%%%%%%%%%%%%%%%%%%%%%%%%%%%%%%%%%%%%%%%%%%%%%%%%%%%%%%%%%%%%
\begin{theorem}\label{main-theorem}

Let $\Psi(x,t)=(\psi(x,t),\dot\psi(x,t))$ be a solution to  (\ref{W})   with initial data  from ${\cal D}$.
Then 
\[
\Psi(x,t)\to (\psi_{q_{\pm}},\,0),\quad t\to\pm\infty,
\]
where $q_{\pm}\in Q$ and the convergence holds 
in $L^2_{loc}(\R^3)\oplus L^2_{loc}(\R^3)$.
\end{theorem}
\begin{proof}
It suffices to  prove this theorem for $t\to+\infty$ only.
 By Lemma \ref{TLWP},  the solution $\psi(x,t)$ to (\ref{W}) with initial data $(\psi_0,\pi_0)\in {\cal D}$, 
can be represented as the sum
\begin{equation}\label{sol_sum}
\psi(x,t):= \psi_f(x,t)+\psi_S(x,t), \quad t\ge 0,
\end{equation}
where  {\it dispersive component} $\psi_f(x,t)$ is a unique solution to  (\ref{CP1}), and {\it singular component}
 $\psi_S(x,t)$ is a unique solution to  the following Cauchy problem 
\begin{equation}\label{CP2}
\ddot\psi_S(x,t)= \Delta\psi_S(x,t) +\zeta(t)\delta(x),
\quad \psi_S(x,0) = 0,\quad\dot\psi_S(x,0)=0.
\end{equation}
Here $\zeta(t)\in C^1_b([0,\infty))$  is a unique solution to 
\begin{equation}\label{delay}
\frac {1}{4\pi}\dot\zeta(t)+ F(\zeta(t))=\lambda(t),
\quad \zeta(0)=\zeta_0.
\end{equation}
Now we can prove local decay of $\psi_f(x,t)$.
%%%%%%%%%%%%%%%%%%%
\begin{lemma}\label{kk1}
For any $R>0$, the following convergence holds
\begin{equation}\label{psif-dec}
\Norm{(\psi_f(t),\dot\psi_f(t))}_{H^2(B_R)\oplus H^1(B_R)}\to 0,\quad t\to\infty. 
\end{equation}
Here $B_R$ is the ball of radius $R$.
\end{lemma}
%%%%%%%%%%%%%%%%%%%%%%%%%
\begin{proof}
We represent the initial data  $(\psi_0,\pi_0)=(\psi_{0,reg},~\pi_{0,reg})+(\zeta_0G, ~\eta_0 G)\in {\cal D}$ as
\[
(\psi_0,~\pi_0)=(\varphi_{0},~p_{0})+(\zeta_0\chi G, ~\eta_0 \chi G),
\]
where a cut-of function $\chi \in C_0^\infty(\R^3)$ satisfies
\begin{equation}\label{G1}
\chi (x)=\left\{\begin{array}{ll} 1,\quad |x|\le1,\\
0,\quad |x|\ge 2.
\end{array}\right.
\end{equation}
Let us show that
$$
(\varphi_{0},~p_{0})\in H^2\oplus H^1. 
$$
Indeed,
\[
(\varphi_{0},~p_{0})=(\psi_0-\zeta_0\chi G,~\pi_0-\eta_0\chi G)\in L^2\oplus L^2.
\]
On the other hand,
\[
(\varphi_{0},~p_{0})=(\psi_{0,reg}+\zeta_0(1-\chi)G,~\pi_{0,reg}+\eta_0(1-\chi)G)
\in \Ho^2\oplus\Ho^1.
\]
Now we split the dispersion component $\psi_f(x,t)$  as
$$
\psi_f(x,t)=\varphi(x,t)+\varphi_{G}(x,t), \quad t\ge 0,
$$
where $\varphi$ and  $\varphi_{G}$  are defined
as solutions to the free wave equation with initial data $(\varphi_{0},~p_{0})$ and  $(\zeta_0\chi G, \eta_0\chi G)$, respectively,
and study the decay properties of  $\varphi_{G}$ and $\varphi$.

%%%%%%%%%%%%%%%%%%%%%%%%%%%%%%%%%%%%%%%%%%%%%%%%%%%%%%%%%%%%%%%%%%%%%%%%%%%%%%%%
First, by the strong Huygens principle
$$
\varphi_{G}(x,t)=0 ~~{\rm for}~~t\ge |x|+2.
$$
Indeed, 
$\varphi_G(x,t)=\zeta_0\dot\psi_G(x,t)+\eta_0\psi_G(x,t)$, where
$\psi_G(x,t)$ is the solution to the free wave equation with initial data 
$(0, \chi G)\in H^1\oplus L^2$, and   $\psi_G(x,t)$ satisfies the strong Huygens principle  by \cite[Theorem XI.87]{RS3}.
\medskip

%%%%%%%%%%%%%%%%%%%%%%%%%%%%%%%%%%%%%%%%%%%%%%%%%%%%%%%%%%%%%%
It remains to check that
\begin{equation}\label{en-dec}
\Norm{(\varphi(t),\dot \varphi(t))}_{H^2(B_R)\oplus H^1(B_R)}\to 0,\quad t\to\infty,\quad\forall R>0,
\end{equation} 
For  $r\ge 1$ denote  $\chi_r=\chi(x/r)$, where $\chi(x)$ is  a cut-off function (\ref{G1}).
Denote $\phi_0=(\varphi_0,\pi_0)$. Let $u_r(t)$ and $v_r(t)$  be  solutions to  free wave
equations with the initial data $\chi_r \phi_0$ and $(1-\chi_r) \phi_0$, respectively, so that
$\varphi(t)=u_r(t)+v_r(t)$. By  the strong Huygens principle 
\[
u_r(x,t)=0 ~~{\rm for}~~t\ge |x|+2r.
\]
To conclude (\ref{en-dec}), it remains to note that
\begin{eqnarray}\nonumber
\Vert(v_r(t),\dot v_r(t))\Vert_{H^2(B_R)\oplus H^1(B_R)}&\le& C(R)
\Vert(v_r(t),\dot v_r(t))\Vert_{\Ho^2\oplus H^1}
= C(R)\Vert (1-\chi_r) \phi_0\Vert_{\Ho^2\oplus H^1}\\
\label{en-dec1}
&\le& C(R)\Vert (1-\chi_r) \phi_0\Vert_{H^2\oplus H^1}
\end{eqnarray}
by the energy conservation for the free wave equation. We also use the embedding
$\Ho^1(\R^3)\subset L^6(\R^3)$. The right-hand side of (\ref{en-dec1})
could be made arbitrarily small if $r\ge 1$ is  sufficiently large.
\end{proof}
%%%%%%%%%%%%%%%%%%%%%%%%%%%%%%%%%%%%%%%%%%%%%%%%%%%%%%%%%%%%%%%%%%%%%%%%%%%%%%%%%%%%%%%%
Due to (\ref{sol_sum}) and  (\ref{psif-dec}), for the proof of Theorem~\ref{main-theorem}  
it suffices to verify
 the convergence of  $\psi_S(x,t)$ to stationary states:
%%%%%%%%%%%%%%%%%%%%%%%%%%%%%%%%%
\begin{lemma}\label{propfin}
Let $\psi_S(x,t)$ and $\zeta(t)$ be  solutions to  (\ref{CP2}) and (\ref{delay}), respectively. Then
\[
(\psi_S(t),\dot\psi_S(t))\to (\psi_{q_{\pm}},~0),\quad t\to\infty,
\]
 where $q_\pm\in Q$ and the convergence holds in $L^2_{loc}(\R^3)\oplus L^2_{loc}(\R^3)$.
\end{lemma}
%%%%%%%%%%%%%%%%%%%%%%%%
\begin{proof}
The unique solution to (\ref{CP2}) is the  spherical wave 
\begin{equation}\label{psiSf}
\psi_S(x,t)=\frac{\theta(t-|x|)}{4\pi|x|}\zeta(t-|x|),\quad t\ge 0,
\end{equation}
cf. (\ref{Smir})--(\ref{Smir1}).
Then a priori bound (\ref{zeta-bound}) and equation (\ref{delay}) imply that 
\[
(\psi_S(t),\dot\psi_S(t))\in   L^2(B_R)\oplus L^2(B_R),\quad 0\le  R<t.
\]
First, we prove  the convergence of $\zeta(t)$.
%%%%%%%%%%%%%%%%%%%%%%%%%%%%%%%%%%%%%%%%%%%%%%%%%%%%%%%%%
From (\ref{zeta-bound}) it follows that $\zeta(t)$ has the upper and lower limits:
\begin{equation}\label{low-upper}
\underline{\lim}_{t\to\infty}\zeta(t)=a,\quad
\overline{\lim}_{t\to\infty}\zeta(t)=b.
\end{equation}
Suppose  that $a<b$.
Then the  trajectory $\zeta(t)$ oscillates between $a$ and $b$.
Assumption (\ref{ab}) implies that   $F(\zeta_0)\not =0$ for some $\zeta_0\in (a,b)$. 
For the concreteness, let us assume that $F(\zeta_{0})>0$.
The convergence (\ref{psif-dec}) implies that
\begin{equation}\label{lam-dec}
\lambda(t)=\psi_f(0,t)\to 0,\qquad t\to\infty.
\end{equation}
Hence, for  sufficiently large $T$ we have
\[
-F(\zeta_{0})+\lambda(t)<0, \quad t\ge T.
\]
Then for $t\ge T$ the transition of the trajectory from left to right through the point $\zeta_0$ is impossible by (\ref{delay}).
Therefore, $a=b=q_+$, where $q_+\in Q$ since  $F(q_+)=0$ by (\ref{delay}). Hence (\ref{low-upper}) implies
 \begin{equation}\label{zetalim}
\zeta(t)\to q_+,\quad t\to\infty,
\end{equation}
%%%%%%%%%%%%%%%%%%%%%%%%%%%%%%%%%%%%%%%%%%%%%%%
Further, 
\begin{equation}\label{tetalim}
\theta(t-|x|)\to 1, \quad t\to\infty
\end{equation}
uniformly in  $|x|\le R$.
Then (\ref{psiSf}) and (\ref{zetalim})  imply that
\[
\psi_S(t)\to q_+G, \quad t\to\infty,
\]
where the convergence holds in $L^2_{loc}(\R^3)$.
It remains to verify the convergence of $\dot\psi_S(t)$.
We have
\[
\dot\psi_S(x,t)=\frac{\theta(t-|x|)}{4\pi|x|}\dot\zeta(t-|x|), \quad |x|<t.
\]
From  (\ref{delay}), (\ref{lam-dec}) and (\ref{zetalim})  it follows  that
$\dot\zeta(t)\to 0$ as $ t\to\infty$.
Then
$$
\dot\psi_S(t)\to 0,\qquad t\to\infty
$$
in $L^2_{loc}(\R^3)$ by (\ref{tetalim}).
\end{proof}

This completes the proof of Theorem~\ref{main-theorem}.
\end{proof}
%%%%%%%%%%%%%%%%%%%%%%%%%%%%%%%%%%%%%%%%%%%%%%%%%%%%%%%%%%%%%%%
%%%%%%%%%%%%%%%%%%%%%%%%%%%%%%%%%%%%%%%%%%%%%%%%%%%%%%%%%%%%%%%%%%%%%%
\subsection {Remarks}
%%%%%%%%%%%%%%%%%%%%%%%%%%%%%%%%%%%%%%%%%%%%%%%%%%%%%%%%%%%%%%%
All above results on global attraction to stationary states
refer to “generic” systems with a trivial symmetry group, which are characterized by
a suitable
discreteness of attractors, by Wiener condition, etc.
\smallskip

Global attraction to stationary states (\ref {ate}) resembles similar asymptotics (\ref {at11}) for dissipative systems.
 However, there are a number of fundamental differences:
 \smallskip \\
I. In dissipative systems attractor always consists of {\it stationary states},  the attraction (\ref {at11})  {\it holds
only as $ t \to + \infty $}, and this attraction is due to the absorption of energy and can  be in global norms.
Such attraction also {\it holds for all finite-dimensional dissipative systems}.
\medskip \\
II. On the other hand, in Hamiltonian systems {\it attractor may differ from the set of stationary states},
as will be seen later. In addition,  energy absorption in these systems is absent, and the attraction \eqref {ate} to
stationary states  is due to the {\it radiation of energy to infinity}, which plays the role of energy absorption.
This attraction takes place both {\it as $ t \to \infty $, and as $ t \to - \infty $}, and it holds {\it only in local seminorms}.
 Finally, it {\it cannot hold for any finite-dimensional Hamiltonian systems} 
(except for the case when the Hamiltonian is an identical constant).
\newpage
%%%%%%%%%%%%%%%%%%%%%%%%%%%%%%%%%%%%%%%%%%%%%%%%%%%%%%%%%%%%%%%
%%%%%%%%%%%%%%%%%%%%%%%%%%%%%%%%%%%%%%%%%%%%%%%%%%%%%%%%%%%%%%%%%%%%%%
\setcounter{equation}{0}
\section {Global attraction to solitons} \la {s2}
%%%%%%%%%%%%%%%%%%%%%%%%%%%%%%%%%%%%%%%%%%%%%%%%%%%%%%%%%%%%%%%
%%%%%%%%%%%%%%%%%%%%%%%%%%%%%%%%%%%%%%%%%%%%%%%%%%%%%%%%%%%%%%%%%%%%%%

As already mentioned in the introduction,  the soliton asymptotics  \eqref {attN} with several solitons
were discovered for the first time numerically in 1965 for KdV  by Kruskal and   Zabusky. 
Later on such asymptotics were proved by the method of {\it   inverse scattering problem} 
for nonlinear {\it integrable} Hamiltonian translation-invariant equations 
by Ablowitz, Segur, Eckhaus, Van Harten and others (see \cite {EvH}).

Here we present the   results on global attraction to one soliton \eqref {att} for nonlinear translation-invariant  {\it non-integrable} 
Hamiltonian equations. Such attraction  was  proved first in \ci {KS1998} 
and in \ci {IKM2004}
for  charged relativistic particle coupled  to  the scalar wave field
and to the  Maxwell field respectively.
%%%%%%%%%%%%%%%%%%%%%%%%%%%%%%%%%%%%%%%%%%%%%%%%%%%%%
%%%%%%%%%%%%%%%%%%%%%%%%%%%%%%%%%%%%%%%%%%%%%%%%%%%%%%
\subsection {Translation-invariant wave-particle system}\la{sTIWP}
%%%%%%%%%%%%%%%%%%%%%%%%%%%%%%%%%%%%%%%%%%%%%%%%%%%%%
In \cite {KS1998}
the system \eqref {w3}--\eqref {q3} was considered  in the case of zero potential $ V(x)\equiv 0 $:
\begin {equation} \label {wq3}
\ddot \psi (x, t) = \Delta \psi (x, t) - \rho (x-q), \, \, \, x \in \mathbb R ^ 3
; \quad
\dot q = \displaystyle \frac {p} {\sqrt {1 + p ^ 2}}, \, \, \dot p = - \displaystyle \int \nabla \psi (x, t) \rho (x-q) \, dx,
\end {equation}
which can be written in the Hamilton form (\ref{HAM}).
The Hamiltonian of this system is given by  (\ref {Ham}) with $ V = 0 $, and it is conserved along trajectories.
By Lemma \ref {lex} with $ V (x) \equiv 0 $, global solutions exist for all initial data
$ Y (0) \in \cE $, and a priori estimates (\ref {apr}) hold.

This system is translation-invariant,  so the corresponding full momentum
\begin {equation} \label {P}
P = p- \int \pi (x) \nabla \psi (x) \, dx
\end {equation}
is also conserved.
Respectively, the system   \eqref {wq3} admits traveling-wave type solutions (solitons)
\begin {equation} \label {solit}
\psi_ {v} (x-a-v t), \, \, \,
q (t) = a+vt, \, \, \, p_v = v / \sqrt {1-v ^ 2},
\end {equation}
where $ v, a \in \R ^ 3 $, and $ | v | <1 $.
These functions  are easily determined: for $ | v | <1 $ there is a unique function $ \psi_v $ which makes (\ref {solit})
a solution to (\ref {wq3}),  
\be \la {3 '}
\psi_v (x) = - \int d ^ 3y (4 \pi | (y-x) _ \| + \lambda (y-x) _ \bot |) ^ {- 1} \rho (y), 
\ee
where we set $\lambda = \sqrt {1-v ^ 2}$ and  $x=x _ \|+x_\bot $, where $x _ \|  \| v$ and $x_\bot \bot v$ for $x\in \R^3$.
Indeed, 
substituting (\ref{solit}) into the wave equation of (\ref{wq3}), we get the stationary equation 
\be\la{3.9'} 
(v\cdot\nabla)^2\psi_v(x)=\Delta\psi_v(x)-\rho(x).
\ee 
Through the Fourier transform  
\be\la{A1} 
\hat\psi_v(k)=-\hat\rho(k)/(k^2-(v\cdot k)^2),
\ee 
which implies \eqref{3 '}.
The set of all solitons forms  $ 6 $ -dimensional \textit {soliton submanifold} in the Hilbert phase space  $ \cE $:
\begin {equation} \label {cS}
\cS = \{S_ {v, a} = (\psi_ {v} (x-a),\, \pi_ {v} (x-a),\, a, \,p_v): \quad v, \,a \in \mathbb R ^ 3, \quad | v | <1 \},
\end {equation}
where $ \pi_v: = - v \na \psi_v $.
Recall that the spaces $\cE$ and $\cE_ \si $ are defined in Definition \ref{dcE}.
The following theorem is the main result of \cite {KS1998}.
%%%%%%%%%%%%%%%%%%%%%%%%%%%%%%%
\begin {theorem} \label {t6}
Let the Wiener condition \eqref {W1} hold  and $ \si> 3/2 $. Then for any initial state $ Y (0) \in \cE_ \si $,
the correspoding solution $ Y(t) = (\psi (t), \pi (t), q (t), p (t)) $ of
the system \eqref {wq3} converges to the soliton manifold $ \cS $ in the following sense:
\begin {equation} \label {dq}
\ddot q(t)\to 0,\quad\dot q (t) \to v_ \pm, \qquad t \to \pm \infty,
\end {equation}
\begin {equation} \label {ssol}
(\psi (x, t), \dot \psi (x, t))=(\psi_ {v_ \pm} (x-q (t)), \pi_ {v_ \pm} (x-q (t))) +(r_ \pm (x, t), s_ \pm (x, t)),
\end {equation}
where  the remainder  decreases locally in the {\bf comoving frame}: for each $ R> 0 $
\begin {equation} \label {ssolh}
\Vert \nabla r_ \pm (q (t) + x, t) \Vert_R + \Vert r_ \pm (q (t) + x, t) \Vert_R + \Vert s_ \pm (q (t) + x, t ) \Vert_R \to 0, \qquad t \to \pm \infty.
\end {equation}
\end {theorem}
%%%%%%%%%%%%%%%%%%%%%%%%%%%%%
The theorem means that, in particular,
\be \la {tsol}
\psi (x, t) \sim \psi_v (x-v_ \pm t + \vp_ \pm (t)), \quad {\rm where}\quad \dot \vp_ \pm (t) \to 0, \quad t \to \pm \infty.
\ee
The proof  \cite {KS1998} relies on a) relaxation of acceleration  (\ref {rel}) in the case $ V = 0 $ (see Remark \ref {rVW} i)), 
and b) on the \textit {canonical change of variables} to the comoving frame.
The key role is played by the fact that the  soliton $ S_ {v, a} $ minimizes the Hamiltonian (\ref {Ham}) (in the case $ V = 0 $)
with a fixed total momentum \eqref {P}, which implies \textit {orbital stability of solitons} \cite {GSS87, GSS90}.
In addition, the proof essentially relies on the \textit {strong Huygens principle} for the three-dimensional wave equation.
\smallskip

Before entering into more precise and technical discussion,  it may be useful to give  general idea of our strategy. 
As was mentioned above, the total  momentum (\ref{P}) is conserved
because of translation invariance. 

We transform the system (\ref{wq3})
to new variables $(\Psi(x),\Pi(x),Q,P)=(\psi(q+x),\pi(q+x),q,P(\psi,q,\pi,p))$.
The key role in our strategy is played by
the fact that this transformation is canonical, which is  proved in
Section \ref{secan}.
Through this canonical transformation 
one obtains the new Hamiltonian 
\begin{eqnarray*} 
&&\cH_P(\Psi,\Pi)=\cH(\psi,\pi,q,p) 
\\
&&=\int\,d^3x\,\Big(\fr 12 |\Pi(x)|^2+\fr 12|\nabla\Psi(x)|^2+ \Psi(x)\rho(x)\Big) 
+\Big[1+\Big( P+\int\,d^3x\,\Pi(x)\,\nabla\Psi(x)\Big)^2\,\Big]^{1/2}\!\!\!\!\!. 
\nonumber 
\end{eqnarray*} 
Since $Q$ is the cyclic coordinate (i.e., the Hamiltonian $\cH_P$ does not depend on $Q$), we may regard $P$ as a fixed parameter and consider 
the reduced system for $(\Psi,\Pi)$ only. Let us define 
\be\la{6'} 
\pi_v(x)=-v\cdot\nabla\psi_v(x),\quad P(v)=p_v+\int\,d^3x\,v\cdot\nabla\psi_v(x)\,\nabla\psi_v(x)\,, 
\,\,p_v=v/(1-v^2)^{1/2}\,\,. 
\ee 
We will prove that  $(\psi_v,\pi_v)$ is the unique  critical point and global minimum of $H_{P(v)}$\,. Thus,  if initial data 
is close to $(\psi_v,\pi_v)$, then corresponding solution must remain close forever by conservation of energy, which translates into the orbital stability of the solitons. 
Here we follow the ideas of the Bambusi and  Galgani paper \ci{BG},
were the orbital stability of  solitons for the Maxwell--Lorentz
equations was proved for the first time.
For a general class of nonlinear wave equations with symmetries such 
approach to
orbital stability of the 
solitons was developed in the well known work \ci{GSS}.

However, the orbital stability by itself is not enough. It only ensures that initial states, close to a soliton, remain so, but does not 
yield the convergence of $\dot q(t)$ in (\ref{dq}), and even less the  asymptotics (\ref{ssol}), (\ref{ssolh}). Thus we need an additional, 
not quite obvious argument which combines the relaxation (\ref{rel}) with the orbital stability in order 
to establish the soliton-like  asymptotics (\ref{dq}), (\ref{ssol}), (\ref{ssolh}).
As one essential input we will use the strong Huygens principle for  wave equation. 

%%%%%%%%%%%%%%%%%%%%%%%%%%%%%%%%%%%%%%%%%%%%%%%%%%%%%%%
%%%%%%%%%%%%%%%%%%%%%%%%%%%%%%%%%%%%%%%%%%%%%%%%%%%%%%%%%
\subsubsection{Canonical transformation and reduced system} 
%%%%%%%%%%%%%%%%%%%%%%%%%%%%%%%%%%%%%%%%%%%%%%%%%%%%%%%%
Since the total momentum is conserved, it is natural to use  $P$ as a new coordinate. To maintain the symplectic structure we 
have to complete this coordinate  to
 a canonical transformation of the Hilbert phase space $\cE$. 
\begin{definition}\la{TD} 
Let the transform $T:{\cal E}\to{\cal E}$ be defined by 
\be\la{3.1} 
T:Y=(\psi,\pi,q,p)\mapsto Y^T=( \Psi(x),\Pi(x),Q,P) =(\psi(q+x),\pi(q+x),q,P(\psi,q,\pi,p))\,\,, 
\ee 
where $P(\psi,q,\pi,p)$ is the total momentum (\ref{P}). 
\end{definition} 
%%%%%%%%%%%%%%%%%%%%
\begin{remarks} 
i) $T$ is continuous on ${\cal E}$ and Fr\'echet differentiable at points $Y=(\psi,q,\pi,p)$ with sufficiently 
smooth $\psi(x),\pi(x)$, but it is not everywhere differentiable. 
 \\
ii) 
In  the $T$-coordinates the solitons $Y_{v,a}(t)=(\psi_v(x-a-vt), \pi_v(x-a-vt), q=a+vt, p_v)$  are stationary except for the coordinate $Q$, 
\be\la{3.2} 
TY_{v,a}(t)=(\psi_v(x),\pi_v(x), a+vt, P(v)) 
\ee 
with  the total momentum $P(v)$ of the soliton  defined in (\ref{6'}). 
\end{remarks} 
%%%%%%%%%%%%%%%%%%%%%%
Denote  $\cH^T(Y)=\cH(T^{-1}Y)$ for $Y=(\Psi,\Pi,Q,P)\in{\cal E}$\,. Then 
\begin{eqnarray*} 
&&\cH^T(\Psi,\Pi,Q, P)=\cH_P(\Psi,\Pi)=\cH(\Psi(x-Q),\Pi(x-Q),Q, P+\int\,d^3x\,\Pi(x)\,\nabla\Psi(x)) 
\\
&&=\int\,d^3x\,\big[\fr 12 |\Pi(x)|^2+\fr 12|\nabla\Psi(x)|^2+ \Psi(x)\rho(x)\big] +\Big(1+\big[P+\int\,d^3x\,\Pi(x)\,\nabla\Psi(x)\big]^2\Big)^{1/2}\!\!. 
\end{eqnarray*} 
The functionals $\cH^T$ and $\cH$ are Fr\'echet-differentiable on the phase space ${\cal E}$. 
%%%%%%%%%%%%%%%%%%%%%%%%
\bp\la{CTS} 
Let $Y(t)\in C(\R,{\cal E})$  be a solution to the system (\ref{wq3}). Then 
$$Y^T(t):=TY(t)=(\Psi(t),\Pi(t), q(t), p(t))\in C(\R,{\cal E})$$
 is a solution to the Hamiltonian system 
\be\la{3.4} 
\left\{\ba{rcl}
\dot\Psi=D_\Pi \cH^T, & \dot\Pi=-D_\Psi \cH^T\\
\dot Q=D_P \cH^T, & \dot P=-D_Q \cH^T
\ea\right|.
\ee
\ep
%%%%%%%%%%%%%%%%%%%%%%%%
\begin{proof} 
The equations for $\dot \Psi$, $\dot \Pi$ and $\dot Q$ can be checked by direct computation, while the one for $\dot P$ 
follows from conservation of the total momentum  (\ref{P}) since the Hamiltonian $\cH^T$ does not depend on $Q$.
\end{proof}

%%%%%%%%%%%%%%%%%%%%%%%%%%%%%%% 
\br
Formally, Proposition \ref{CTS} follows from the fact that  $T$ is a canonical transform, see Section \ref{secan}.
\er
%%%%%%%%%%%%%%%%%%%%
Recall that $Q$ is a cyclic coordinate. Hence, the system (\ref{3.4}) is equivalent to a reduced Hamiltonian system for 
$\Psi$ and $\Pi$ only, which can be written as 
\be\la{3.5} 
\dot\Psi=D_\Pi \cH_P, \qquad \dot\Pi=-D_\Psi \cH_P.
\ee
Due to (\ref{3.2}), the soliton $(\psi_v,\pi_v)$ is a stationary solution to  (\ref{3.5})
with $P=P(v)$. Moreover,
for every $P\in\R^3$, the functional $\cH_P$ is Fr\'echet differentiable on the Hilbert space 
${\cal F}=\Ho^1\oplus L^2$\,. 
Hence, (\ref{3.5})  implies that the soliton is a critical point of  $\cH_{P(v)}$ on ${\cal F}$\,. 
 The next lemma demonstrates that $(\psi_v,\pi_v)$ is a global minimum
of $\cH_{P(v)}$ on ${\cal F}$.
%%%%%%%%%%%%%%%%%%%%%%%%%%%%%%%
\begin{lemma}\la{SCP} 
i) For every $v\in \R^3$ with $|v|<1$ the functional $\cH_{P(v)}$ has the lower bound 
\be\la{3.6} 
\cH_{P(v)}(\Psi,\Pi)-\cH_{P(v)}(\psi_v,\pi_v)\geq \fr{1-|v|}2\Big(\Vert\Psi-\psi_v\Vert^2+\Vert\Pi-\pi_v\Vert^2\Big),
\quad  (\Psi,\Pi)\in {\cal F}.
\ee 
ii) $\cH_{P(v)}$ has no other critical points on ${\cal F}$  except  point $(\psi_v,\pi_v)$.
\end{lemma}
%%%%%%%%%%%%%%%%%%%%%%%%%%%%% 
\begin{proof}
{\it Step i)} Denoting  $\Psi-\psi_v=\psi$ and $\Pi-\pi_v=\pi$, we have 
\beqn\la{3.9} 
\cH_{P(v)}(\psi_v+\psi,\pi_v+\pi)-\cH_{P(v)}(\psi_v,\pi_v)= 
\int d^3x(\pi_v(x)\pi(x)+\nabla\psi_v(x)\cdot\nabla\psi(x)+\rho(x)\psi(x)) 
\nonumber\\
+\fr 12 \int\,d^3x\,(|\pi(x)|^2+|\nabla\psi(x)|^2 ) + (1+(p_v+m)^2)^{1/2}-(1+p_v^2)^{1/2}\,\,,\,\,\,\,\,\,\,\,\,\,\,\,\, 
\eeqn 
where $p_v=P(v)+\int\,d^3x\,\pi_v(x)\,\nabla\psi_v(x)$, and 
$$ 
m =\int\,d^3x\,(\pi(x)\,\nabla\psi_v(x)+\pi_v(x)\,\nabla\psi(x)+ \pi(x)\,\nabla\psi(x)). 
$$ 
Taking into account that  $v=(1+ p_v^2)^{-1/2} p_v$, we obtain 
\begin{eqnarray*} 
&&\cH_{P(v)}(\psi_v+\psi,\pi_v+\pi)-\cH_{P(v)}(\psi_v,\pi_v)\\
&&= \fr 12 \int\,d^3x\,(|\pi(x)|^2+|\nabla\psi(x)|^2 ) +(1+ p_v^2)^{-1/2} \int\,d^3x\,\pi(x)\,p_v\cdot\nabla\psi(x)\\
&&-(1+ p_v^2)^{-1/2} p_v\cdot m+(1+(p_v+m)^2)^{1/2}-(1+p_v^2)^{1/2}\,. 
\end{eqnarray*} 
It is easy to check that the expression in the third line  is nonnegative. Then the lower bound (\ref{3.6}) follows by using 
 $|(1+ p_v^2)^{-1/2} p_v|=|v|$\,. 
 \smallskip\\
 {\it Step ii)}
If $(\Psi,\Pi)\in{\cal F}$ is a critical point for $\cH_{P(v)}$, then it satisfies 
$$ 
0=\Pi(x)+(1+\ti p^2)^{-1/2}\ti p \cdot\nabla\Psi(x)\,, 
\,\,\,0=-\Delta\Psi(x)+\rho(x)-(1+\ti p^2)^{-1/2}\ti p\cdot\nabla\Pi(x)\,, 
$$ 
where $\ti p=P(v)+\ds \int\,d^3x\,\Pi(x)\,\nabla\Phi(x)$\,. This system is equivalent to equation
 (\ref{3.9'}) for solitons in the case of  the  velocity 
 $\ti v=(1+\ti p^2)^{-1/2}\ti p$\,. 
Hence, $\Psi=\psi_{\ti v}$\,, $\Pi=\pi_{\ti v}$ and $P(\ti v)=P(v)$\,. 

It remains to check that  $\ti v=v$.
Indeed,  for the total momentum $P(v)$ of the soliton-like solution (\ref{solit}), 
the Parseval identity and (\ref{A1}) imply 
$$ 
P(v)=p_v+\int\,d^3x\,\,v\cdot\nabla\psi_v(x)\,\nabla\psi_v(x)
=\fr v{\sqrt{1-v^2}}+(2\pi)^{-3}\int\,d^3k \fr{(v\cdot k)\hat\rho(k) \ov{k\hat\rho(k)}} 
{(k^2-(v\cdot k)^2)^2}\,\,.
$$
Hence, $P(v)=\kappa(|v|) v$ with $\kappa(|v|)\geq 0$, and for $v\not= 0$  one has
$$ 
|P(v)| 
=\fr {|v|}{\sqrt{1-v^2}}+\fr 1{(2\pi)^3|v|} 
\int\,d^3k 
\fr{|(v\cdot k) \hat\rho(k)|^2} 
{(k^2-(v\cdot k)^2)^2}\,\,. 
$$ 
Since $|P(v)|= \kappa(|v|)|v|$ 
is a monotone increasing function of $|v|\in[0,1[$, we conclude that $v=\ti v$. 
\end{proof}
%%%%%%%%%%%%%%%%%%%%%%%%%%%%%
\br 
{\rm
Proposition \ref{CTS} is not really needed for the proof of Theorem \ref{t6}. 
However, the Proposition together with (\ref{3.2}) and (\ref{3.5}) show  that $(\psi_v,\pi_v)$ is a critical point 
and suggest  an investigation of the stability through a lower bound as in (\ref{3.6}). 
In  Section \ref{secan} we sketch the derivation of Proposition \ref{CTS} for sufficiently smooth solutions 
based only on the invariance of  symplectic structure. 
We expect that a similar proposition  holds for other translation invariant systems similar to (\ref{wq3}). 
}
\er
%%%%%%%%%%%%%%%%%%%%%%%%%%%%%%%%%%%%%%%
%%%%%%%%%%%%%%%%%  Section 4    %%%%%%%%%%%%%%%% 
\subsubsection{Orbital stability of solitons} 
%%%%%%%%%%%%%%%%%%%%%%%%%%%%%%%%%%%%%%

We follow \ci{BG} deducing orbital stability from the conservation of the Hamiltonian 
$\cH_P$ together with its lower bound (\ref{3.6}).  For   $|v|<1$ denote 
\be\la{4.2} 
\de=\de(v)=\Vert\psi^0(x)-\psi_v(x-q^0)\Vert+ \Vert\pi^0(x)-\pi_v(x-q^0)\Vert+|p^0-p_v|\,. 
\ee 

 %%%%%%%%%%%%%%%%%%%%%%%%%%%%
\begin{lemma}\la{OSS}  
Let  $Y(t)=(\psi(t),\pi(t),q(t), p(t))\in C(\R,{\cal E})$ be a solution to  (\ref{wq3}) with an initial 
state 
$Y(0)=Y^0=(\psi^0,\pi^0,q^0,p^0)\in{\cal E}$.Then for every $\ve>0$ there exists a $\de_\ve>0$ such that 
\be\la{4.3} 
\Vert\psi(q(t)+x,t)-\psi_v(x)\Vert+ 
\Vert\pi(q(t)+x,t)-\pi_v(x)\Vert+|p(t)-p_v|\leq\ve,\qquad t\in\R 
\ee 
provided $\de\leq\de_\ve$. 
\end{lemma} 
%%%%%%%%%%%%%%%%%%%%%%%%%%%%%%%
\begin{proof} 
Denote by $P^0$ the total momentum of the considered solution $Y(t)$. 
There exists a soliton-like solution (\ref{solit}) corresponding to some velocity 
$\ti v$ with the same total momentum $P(\ti v)=P^0$\,. Then (\ref{4.2}) implies that
$|P^0-P(v)|=|P(\ti v)-P(v)| ={\cal O}(\de)$.
Hence also $|\ti v-v|={\cal O}(\de)$ and 
$$ 
\Vert\psi^0(x)-\psi_{\ti v}(x-q^0)\Vert+ \Vert\pi^0(x)-\pi_{\ti v}(x-q^0)\Vert+|p^0-p_{\ti v}|={\cal O}(\de)\,. 
$$ 
 Therefore, denoting $(\Psi^0,Q^0,\Pi^0,P^0)=TY^0$, we have 
\be\la{4.5} 
\cH_{P(\ti v)}(\Psi^0\,,\Pi^0)-\cH_{P(\ti v)}(\psi_{\ti v}\,,\,p_{\ti v})= {\cal O}(\de^2)\,. 
\ee 
Total momentum and energy conservation  imply that for $(\Psi(t),Q(t),\Pi(t),P^0)=TY(t)$ 
$$ 
\cH_{P(\ti v)}(\Psi(t),\,\Pi(t))=\cH(TY(t))= \cH_{P(\ti v)}(\Psi^0\,,\Pi^0){\rm \,\,\,for\,\,}t\in\R\,. 
$$ 
Hence (\ref{4.5}) and (\ref{3.6}) with $\ti v$ instead of $v$ imply 
\be\la{4.6} 
\Vert \Psi(t)-\psi_{\ti v}\Vert+\Vert \Pi(t)-\pi_{\ti v}\Vert={\cal O}(\de) 
\ee 
uniformly in $t\in\R$\,. On the other hand, total momentum conservation 
implies 
$$ 
p(t)=P(\ti v)+\langle\Pi(t),\nabla\Psi(t)\rangle{\rm \,\,\,for\,\,}t\in\R\,. 
$$ 
Therefore (\ref{4.6}) leads to  
\be\la{4.7} 
|p(t)-p_{\ti v}|={\cal O}(\de) 
\ee 
uniformly in $t\in\R$\,. 
Finally (\ref{4.6}), (\ref{4.7}) together imply (\ref{4.3}) because 
$|\ti v-v|={\cal O}(\de)$\,. 
\end{proof}
%%%%%%%%%%%%%%%%%%% Section 5    %%%%%%%%%%%%%%%%%%% 
%%%%%%%%%%%%%%%%%%%%%%%%%%%%%%%%%%%%%%%%%%%% 
\subsubsection{Strong Huygens principle and soliton asymptotics} 
%%%%%%%%%%%%%%%%%%%%%%%%%%%%%%%%%%%%%%%%%%%
We combine the relaxation of the acceleration and orbital stability  
with the Strong Huygens principle
to prove Theorem \ref{t6}. 
%%%%%%%%%%%%%%%%%%%%%%
\bp\la{5.1} 
Let the assumptions of Theorem \ref{t6} be fulfilled. 
Then for every $\de>0$ there exist a $t_*=t_*(\de)$ and a solution 
$Y_*(t)=(\psi_*(x,t),\pi_*(x,t),q_*(t),p_*(t)) 
\in C([t_*,\infty),{\cal E})$ to the system (\ref{wq3}) such that 
\smallskip\\
i) $Y_*(t)$ coincides with $Y(t)$ in  the future cone, 
\beqn 
q_*(t)=&q(t)\,\,&for\,\,\,t\geq t_*\,,\la{5.2'}\\
\psi_*(x,t)=&\psi(x,t)\,\,&for\,\,\,|x-q(t_*)|<t-t_*\,.\la{5.2} 
\eeqn 
ii) $Y_*(t_*)$ is close to  a soliton $Y_{v,a}$ with some $v$ and $a$, 
 \be\la{5.3} 
 \Vert Y_*(t_*)-Y_{v,a}\Vert_{\cal E}\leq\de\,. 
 \ee 
 \ep
 %%%%%%%%%%%%%%%%%%%%%%%%%%%%%%%%
\begin{proof}
The Kirchhoff formula gives
$$ 
\psi(x,t)=\psi_r(x,t)+\psi_0(x,t),\qquad x\in\R^3,\,\, t> 0,
$$ 
where 
\beqn 
\psi_r(x,t)&=&-\int\,\fr{d^3y}{4\pi|x-y|}\rho(y-q(t-|x-y|))\,,\la{5.6}\\
\psi_0(x,t)&=&\fr 1{4\pi t}\int_{S_t(x)}\,d^2y\,\pi(y,0)+ 
\fr\pa{\pa t} 
\left( \fr 1{4\pi t}\int_{S_t(x)}\,d^2y\,\psi(y,0) \right).
\la{5.6'} 
\eeqn 
Here $S_t(x)$ denotes the sphere $|y-x|=t$.
Let us assume for simplicity that initial fields vanish.
General case can be easily reduced to this situation using  the strong Huygens principle.
We will comment on this reduction at the end of the proof.

In the case  of zero initial data the solution reduces to the retarded potential:
$$
\psi(x,t)=\psi_r(x,t),\quad x\in\R^3,\quad t>0.
$$
We construct the solution $Y_*(t)$ as a modification of $Y(t)$. First, we modify the trajectory $q(t)$. 
The relaxation of acceleration (\ref{dq}) means that  for any $\ve>0$ there exist  $t_\ve >0$ such that
$$
|\ddot q(t)|\le \ve,\qquad t\ge t_\ve.
$$
Hence, the trajectory for large times  locally tends to a straight line, i.e., for any  fixed $T>0$
$$
q(t)= q(t_\ve)+(t-t_\ve)\dot q(t_\ve)+r(t_\ve,t),\quad \mbox {where}
\quad\max_{t\in [t_\ve,t_\ve+T]} |r(t_\ve,t)|\to 0,\quad t_\ve\to\infty.
$$
Denote $\lam_\ve(t):=q(t_\ve)+\dot q(t_\ve)(t-t_\ve)$ and define  modified trajectory as 
\be\la{qT2}
q_*(t)=\left\{ \ba{l} \lam_\ve(t),\!\qquad t\le t_\ve
\medskip\\
q(t),~\qquad t\ge t_\ve \ea\right|,
\ee
Then
$$
\ddot q_*(t)=\left\{ \ba{l} 0,~~\qquad t< t_\ve
\medskip\\
\ddot q(t),\!\qquad t> t_\ve\ea\right|.
$$
The next step we define the modified field as  retarded potential of type
(\ref{5.6})
\be\la{5.62}
\psi_*(x,t)=-\int\,\fr{d^3y}{4\pi|x-y|}\rho(y-q_*(t-|x-y|)),\quad x\in\R^3, \quad t\in\R.   
\ee
%%%%%%%%%%%%%%%%%%%
\bl\la{lTT}
The right hand side of (\ref{5.62}) depends on the trajectory $q_*(\tau)$
only from a bounded interval of time $\tau \in[t-T(x,t),t]$, where
\be\la{tau}
T(x,t):=\fr{R_\rho+|x-q(t)|}{1-\ov v}.
\ee
Here  $\ov v=\sup\limits_{t\in\R}|\dot q(t)|<1$ by (\ref{apr}).
\el
%%%%%%%%%%%%%%%%%%%%
\begin{proof}
This lemma is obvious geometrically, and its formal proof also is easy.
The inegrand of  (\ref{5.62}) vanishes for $|y-q_*(t-|x-y|)|\ge R_\rho$ by (\ref{C}).
Therefore, the integral is spreaded over the region $|y-q_*(t-|x-y|)|\le R_\rho$, which implies
$|y-q_*(t)+q_*(t)-q_*(t-|x-y|)|\le R_\rho$. Hence, 
$$
|y-q_*(t)|\le R_\rho+\ov v |x-y|.
$$
On the other hand, $ |x-y|\le |x-q_*(t)|+|y-q_*(t)|$, and hence,
$$
|y-q_*(t)|\ge -|x-q_*(t)|+|x-y|.
$$
Therefore,
$$
 -|x-q_*(t)|+|x-y| \le R_\rho+\ov v |x-y|,
$$
which implies 
$$
 |x-y| \le \fr{R_\rho+|x-q_*(t)|}{1-\ov v}.
$$
Now the lemma is proved.
\end{proof}
%%%%%%%%%%%%%%%%%%%%%%%%
The potential (\ref{5.62}) satisfies the wave equation
$$
\ddot \psi_*(x,t)=\De\psi_*(x,t)-\rho(x-q_*(t)),\quad x\in\R^3,\quad t\in\R.
$$
We should still   prove  equations for the trajectory $q_*(t)$:
\begin {equation} \label {wq32}
\dot q_*(t) = \displaystyle \frac {p_*(t)} {\sqrt {1 + p_* ^ 2(t)}}, 
\qquad\dot p_*(t) = - \displaystyle \int \nabla \psi_* (x, t) \rho (x-q_*(t)) \, dx,\qquad t>t_*
\end {equation}
with sufficiently large $t_*\ge t_\ve$.
Let us note that the integral here is spreaded over the ball $|x-q_*(t)|\le R_\rho$.
Now Lemma \ref{lTT} implies that $\psi_* (x, t)$ depends on the trajectory $q_*(\tau)$ 
only from a bounded interval $\tau \in[t-\ov T,t]$, where
$$
\ov T:=\fr{2R_\rho}{1-\ov v}.
$$
Let us define $t_*:=t_\ve+\ov T$. Then by Lemma \ref{lTT}
$$
\psi_*(x,t)=\psi(x,t),\quad t>t_*,\quad |x-q_*(t)|\le R_\rho
$$
since $q_*(t)\equiv q(t)$ for $t>t_*-\ov T=t_\ve$ by (\ref{qT2}). 
Hence, equations (\ref{wq32}) hold for $q_*(t)$ as well as for $q(t)$.

It remains to prove  (\ref{5.3}). 
The key observation is that outside the cone $K_\ve:=\{(x,t)\in\R^4:|x-q(t_\ve)|<t-t_\ve\}$
 the retarded potential (\ref{5.62}) coincides with the  soliton $\psi_{v,a}(x,t)$, 
 where $v=\dot q(t_\ve)$ and $a= q(t_\ve)$ by our definition (\ref{qT2}). In particular,
$$
\psi(x,t_*)=\psi_{v,a}(x-a-vt_*),\qquad |x-q(t_\ve)|>t_*-t_\ve=\ov T.
$$
In the  ball $|x-q(t_*)|<\ov T$ the coincidence generally does not hold, but
the difference of the left hand side with the right hand side
converges to zero as $\ve\to 0$ uniformly for $|x-q(t_*)|<\ov T$,
and such uniform convergence holds for the gradient of the difference.
This follows from the integral representation (\ref{5.62}) by Lemma \ref{lTT} since 
 $$
 \max_{t\in(t_*-T(x,t_*),t_*)}[|q_*(t)-\lam_\ve(t)|+|\dot q_*(t)-\dot \lam_\ve(t)|]\to 0,
 \qquad \ve\to 0
 $$ 
 by the relaxation of acceleration (\ref{dq}).  
 It is important that $T(x,t_*)$ is bounded for $|x-q(t_*)|<\ov T$ by  (\ref{tau}).
 This proves Proposition \ref{5.1} in the case of zero initial data.
 
 The next step is the proof for initial data with bounded support:
 $$
 \psi(x,0)=\pi(x,0)=0,\qquad |x|> R_0.
 $$
 Now we apply the strong Huygens principle: in this case the potential (\ref{5.6'}) vanishes in a future cone,
 $$
 \psi_0(x,t)=0,\qquad |x|<t-R_0.
 $$
However, the estimate $|\dot q(t)|\le\ov v<1$ implies that the trajectory $(q(t),t)$ lies in this cone for all $t>t_0$.
Hence, the solution for $t>t_0$ again reduces to the retarded potential and the needed conclusion follows. 

Finally, arbitrary finite energy initial data admits a splitting in two summands:
the first vanishing for $|x|>R_0$ and the second vanishing for $|x|<R_0-1$.
The energy of the second summand is arbitrarily small for large $R_0$,
and the energy of the corresponding  potential (\ref{5.6'}) is conserved in time
since it is a solution to  free wave equation. Hence, its role  is negligible for sufficiently large $R_0$. 
\end{proof}
Now we can prove our main result. 
\\
{\bf Proof of Theorem \ref{t6}} 
For every $\ve >0$ there exists  $\de>0$ such that  (\ref{5.3}) implies by Lemma \ref{OSS}, 
$$ 
\Vert\psi_*(q_*(t)+x,t)-\psi_v(x)\Vert+ \Vert\pi_*(q_*(t)+x,t)-\pi_v(x)\Vert+|\dot q_*(t)-v| \leq\ve{\rm \,\,for\,\,}t>t_*\,. 
$$ 
Therefore, (\ref{5.2'}) and (\ref{5.2}) imply that for every $R>0$ and   $t>t_*+\ds\fr{R}{1-\ov v}$ %with $\ov v<1$ from (\ref{apr}), 
\begin{eqnarray*} 
&&\Vert\psi (q(t)+x,\,t)-\psi_v(x)\Vert_R\,+\, \Vert\pi (q(t)+x,\,t)-\pi_v(x)\Vert_R\,+\,|\dot q(t)-v|\,\\
&&=\Vert\psi_*(q_*(t)+x,t)-\psi_v(x)\Vert_R+\, \Vert\pi_*(q_*(t)+x,t)-\pi_v(x)\Vert_R+\,|\dot q_*(t)-v|\le\ve. 
\end{eqnarray*}  
Since $\ve >0$ is arbitrary, we conclude (\ref{ssolh}). Theorem \ref{t6} is proved.

%%%%%%%%%%%%%%%%%%%%%%%%%%%%%%%%%%%%%%%%%%%%%%%%%%%%%%%%%%%%%
%%%%%%%%%%%%%%%%%%%%%%%%%%%%%%%%%%%%%%%%%%%%%%%%%%%%%%%%%%%%%
\subsection{Invariance of symplectic structure} \la{secan}
%%%%%%%%%%%%%%%%%%%%%%%%%%%%%%%%%%%%%%%%%%%%%
The canonical equivalence of the Hamiltonian systems (\ref{wq3}) and (\ref{3.4}) can be seen from the Lagrangian viewpoint. 
We remain at the formal level. For a complete mathematical justification we would have to develop some theory 
of infinite dimensional Hamiltonian systems which is beyond the scope of this paper. 
 
By definition we have $\cH^T(\Psi,\Pi,Q,P)=\cH(\psi,\pi,q,p)$ with the arguments related through the transformation $T$. 
To each Hamiltonian we associate a Lagrangian through the Legendre transformation 
\begin{eqnarray*} 
L(\psi,\dot\psi,q,\dot q)=& 
\langle\pi,\dot\psi\rangle+p\cdot\dot q-\cH(\psi,\pi,q,p)\,\,,\,\,\,\,\,\, 
&\,\,\,\dot\psi =D_\pi \cH\,\,,\,\,\,\,\dot q=D_p \cH\,\,,\\
L^T(\Psi,\dot\Phi,Q,\dot Q)=& 
\langle \Pi,\dot\Psi\rangle+P\cdot\dot Q-\cH^T(\Psi,\Pi,Q,P)\,\,, 
&\,\,\,\dot\Psi =D_\Pi \cH^T\,\,,\,\,\,\dot Q=D_P \cH^T\,\,. 
\end{eqnarray*} 
These Legendre transforms are well defined because the Hamiltonian functionals are convex in the momenta. 
%%%%%%%%%%%%%%%%%%%%%%%%%%%%%
\bl\la{lCT}
The following indentity holds,
$$
L^T(\Psi,\dot\Psi,Q,\dot Q)=L(\psi,\dot\psi,q, \dot q).
$$
\el
%%%%%%%%%%%%%%%%%%%%%%%%%%%%%% 
\begin{proof}
Clearly we have to check the invariance of the canonical 1-form, 
\be\la{AA2} 
\langle \Pi,\dot\Psi\rangle+P\cdot\dot Q=\langle \pi,\dot\psi\rangle +p\cdot\dot q\,. 
\ee 
For this purpose we substitute 
$$ 
\left\{
\ba{rclrcl}
\Pi(x)&=&\pi(q+x),       \qquad &\dot\Psi(x)                    &=&\dot\psi(q+x)+\dot q\cdot\nabla\psi(q+x)\\
       P&=&p-\ds\int \dot\psi\cdot\nabla\psi \,dx,&\dot Q&=&\dot q 
\ea\right|
$$ 
Then the left hand side of (\ref{AA2}) becomes 
$$ 
\langle \pi(q+x),\dot\psi(q+x)+\dot q\cdot\nabla\psi(q+x)\rangle+ 
(p - \langle \pi(x),\nabla\psi(x)\rangle)\cdot \dot q=\langle \pi,\dot\psi\rangle+p\cdot \dot q\,. 
$$ 
The lemma is proved.
\end{proof}
%%%%%%%%%%%%%%%%%%%%%%%%%
This lemma implies that he corresponding action functionals are identical when transformed by $T$. 
Hence, finally, 
the two Hamiltonian systems (\ref{wq3}) and (\ref{3.4}) are equivalent
since dynamical trajectories are stationary points of the respective action functionals. 

%%%%%%%%%%%%%%%%%%%%%%%%%%%%%%%%%%%%%%%%%%%%%%%%%%%%%%%%%%%%%
%%%%%%%%%%%%%%%%%%%%%%%%%%%%%%%%%%%%%%%%%%%%%%%%%%%%%%%%%%%%%
\subsection {Translation-invariant Maxwell-Lorentz system}
%%%%%%%%%%%%%%%%%%%%%%%%%%%%%%%%%%%%%%%%%%%%%%%%%%%%%%%%%%%%%
In \cite {IKM2004} asymptotics of  type (\ref {dq})--(\ref {ssolh}) were extended to the Maxwell-Lorentz translation-invariant 
system \eqref {ML} without external fields. In this case, the Hamiltonian coincides with \eqref {HAMml} where $ V(x)\equiv 0 $.
The extension of  methods \cite {KS1998} to this case required a new detailed analysis of the corresponding Hamiltonian structure 
which is necessary for the canonical transformation.
Now the key role in applying Huygens' strong principle is played by new estimates 
of long-time decay for oscillations of energy and total momentum solutions for perturbed Maxwell-Lorentz system
(estimates (4.24)--(4.25) in \cite {IKM2004}).
%%%%%%%%%%%%%%%%%%%%%%%%%%%%%%%%%%%%%%%%%%%%%%%%%%%%%%%%%%%%%
%%%%%%%%%%%%%%%%%%%%%%%%%%%%%%%%%%%%%%%%%%%%%%%%%%%%%%%%%%%%%
\subsection {The case of weak interaction} \la {19.3}
%%%%%%%%%%%%%%%%%%%%%%%%%%%%%%%%%%%%%%%%%%%%%%%%%%%%%%%%%%%%%

Soliton asymptotic of the type (\ref {dq})--(\ref {ssolh}) for the system \eqref {w3}--\eqref {q3}  was proved in a stronger 
form for the case of a weak coupling
\begin {equation} \label {rosm}
\Vert \rho \Vert_ {L ^ 2 (\mathbb R ^ 3)} \ll 1.
\end {equation}
Namely, in \cite {IKS2004a} initial fields are considered with decay $ | x | ^ {- 5 / 2- \ve} $,  
 $ \forall\ve> 0 $
(condition (2.2) in \cite {IKS2004a}) provided that $ \nabla V (q) = 0 $ for $ | q |> \const $.
Under these assumptions, more strong decay holds,
\begin {equation} \label {dqsm}
| \ddot q (t) | \le C (1+ | t |) ^ {- 1- \ve}, \qquad t \in \R,\qquad\forall\ve>0
\end {equation}
for “outgoing” solutions that satisfy the condition
\begin {equation} \label {sfin}
| q (t) | \to \infty, \qquad t \to \pm \infty.
\end {equation}
%In particular, all solutions are outgoing in the case $ V (q) \equiv 0 $.
With these assumptions asymptotics (\ref {dq})--(\ref {ssolh}) can be significantly strengthen: now
$$
\dot q (t) \to v_ \pm, \qquad (\psi (x, t), \pi (x, t))
\! = \!
(\psi_ {v_ \pm} (x \! - \! q (t)), \pi_ {v_ \pm} (x \! - \! q (t))) + W (t) \Phi_ \pm + (r_ \pm (x, t), s_ \pm (x, t)),
$$
where ``dispersion waves'' $ ~ W (t) \Phi_ \pm $ are solutions of a free wave equation, and the remainder
converges to zero in  {\it global energy norm}:
$$
\Vert \nabla r_ \pm (q (t), t) \Vert + \Vert r_ \pm (q (t), t) \Vert + \Vert s_ \pm (q (t), t) \Vert \to 0 , \qquad t \to \pm \infty.
$$
This progress compared with local decay \eqref {ssolh} is due to the fact that we  identified a dispersion wave 
$ W (t) \Phi_ \pm $ under the condition of smallness \eqref {rosm}. This identification is possible due to the rapid decay
\eqref {dqsm}, in difference with \eqref {rel}.

All solitons propagate with velocities  $v <1 $, and therefore they are spatially separated  for large time
from the dispersion waves $ W (t) \Phi_ \pm $, which propagate with unit velocity (Fig. \,\ref {fig-32}).

The proofs rely on  integral Duhamel  representation and on rapid dispersion decay of  solutions to  free
wave equation. Similar result was obtained in \cite {IKM2003} for a system of  type \eqref {w3}--\eqref {q3} 
with the Klein--Gordon equation and in \cite {IKS2002} for the Maxwell-Lorentz system \eqref {ML} 
with the same smallness condition (\ref {sfin})  under  assumption that 
$ E ^ {\rm ext} (x) = B ^ {\rm ext} (x) = 0 $ for $ | x |> \const$.
In \cite {IKS2004b}, this result was extended to the Maxwell-Lorentz system of type \eqref {ML} with a rotating charge.
%%%%%%%%%%%%%%%%%%%%%
\begin {remark} \label {rGC}
{\rm
The results of \cite {IKS2004a, IKS2004b} imply Soffer's ``Grand Conjecture'' \cite [p. 460] {soffer2006}
in a moving frame for translation-invariant  systems under the condition of smallness \eqref {rosm}.
}
\end {remark}

\bigskip
%%%%%%%%%%%%%%%%%%%%%%%%%%%%%%%%%%%%%%%%%%%%%
\begin{figure}[htbp]
\begin{center}
\includegraphics[width=1.0\columnwidth]{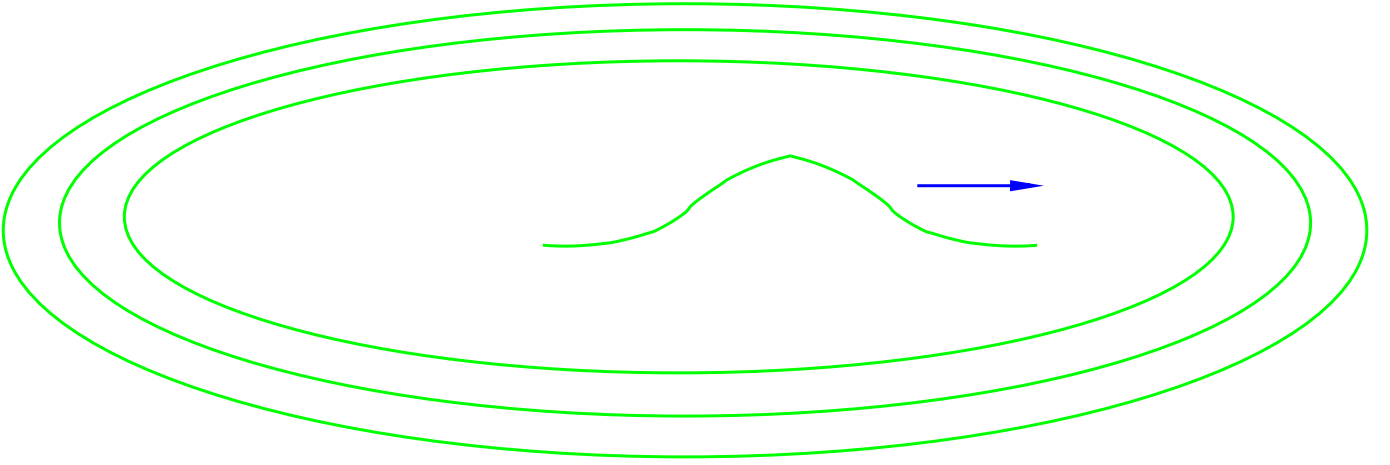}
\caption{Soliton and dispersion waves}
\label{fig-32}
\end{center}
\end{figure}

%%%%%%%%%%%%%%%%%%%%%%%%%%%%%%%%%%%%%%%%%%%%%%
%%%%%%%%%%%%%%%%%%%%%%%%%%%%%%%%%%%%%%%%%%%%%%
\section {Adiabatic effective dynamics of solitons}
%%%%%%%%%%%%%%%%%%%%%%%%%%%%%%%%%%%%%%%%%%%%%%
The existence of solitons and the global attraction to solitons (\ref {att}) are typical features of translation-invariant systems.
However, if the deviation of a system from translational invariance is in some sense small, the system can admit solutions
which  are close forever to solitons with time-dependent parameters (velocity, etc.). Moreover, in some cases
it is possible to identify an ``effective dynamics'' which describes  the evolution of these parameters.

%%%%%%%%%%%%%%%%%%%%%%%%%%%%%%%%%%%%%%%%%%%%%%
%%%%%%%%%%%%%%%%%%%%%%%%%%%%%%%%%%%%%%%%%%%%%%
\subsection {Wave-particle system with a slowly varying external potential} \la {s3}
%%%%%%%%%%%%%%%%%%%%%%%%%%%%%%%%%%%%%%%%%%%%%%

The solitons (\ref {solit}) are solutions to  the system \eqref {wq3}--\eqref {q3}  with zero external potential $V(x)\equiv 0$.
However, even for the  system \eqref {w3}--\eqref {q3} 
with nonzero external potential {\it soliton-like} solutions of the form
\begin {equation} \label {asol}
\psi (x, t) \approx \psi_ {v (t)} (x-q (t))
\end {equation}
may exist if the potential is slowly changing:
\begin {equation} \label {asolV}
| \nabla V (q) | \le \varepsilon \ll 1.
\end {equation}
In this case, the total momentum (\ref {P}) is generally not conserved, but its slow evolution and  slow evolution 
of solutions (\ref {asol}) can be described in terms of some finite-dimensional Hamiltonian dynamics.

Namely, let $ P = P_v $ be  total momentum of the soliton $ S_ {v, Q} $ in the notation (\ref {cS}).
It is important that the map $ \cP: v \mapsto P_v $ is an isomorphism of the ball $ | v |< 1 $ on $ R ^ 3 $.
Therefore, we can consider $ Q, P $ as global coordinates on the soliton manifold $ \cS $. 
We define effective Hamilton functional
\begin {equation} \label {Heff}
\cH _ {\rm eff} (Q, P_v) \equiv \cH_0 (S_ {v, Q}), \qquad (Q, P_v) \in \cS,
\end {equation}
where $ \cH_0 $ is {\it unperturbed} Hamiltonian (\ref {Ham}) with $ V = 0 $.
This functional allows the splitting
$ \cH _ {\rm eff} (Q, \Pi) = E (\Pi) + V (Q) $
since the 
first  integral in (\ref {Ham}) does not depend on $Q$ while the 
last  integral  vanishes on the solitons.
 Hence,  the corresponding Hamilton equations read
\begin {equation} \label {dyn}
\dot Q (t) = \nabla E (\Pi (t)), \qquad \dot \Pi (t) = - \nabla V (Q (t)).
\end {equation}
The main result of \cite {KKS1999} is the following theorem.

%%%%%%%%%%%%%%%%%%%%%%%%%%
\begin {theorem} \label {t7}
Let condition (\ref {asolV}) hold, and the initial state $ S_0=(\psi_0, \pi_0, q_0, p_0) \in \cS$
is a soliton with a full momentum $ P_0 $. Then the corresponding solution
$ \psi (x, t),  \pi (x, t), q (t), p (t) $ to 
 the system \eqref {w3}--\eqref {q3}  admits the following “adiabatic asymptotics”
\begin {eqnarray}
&& | q (t) -Q (t) | \le C_0, \quad | P (t) - \Pi (t) | \le C_1 \varepsilon \quad \mbox {\rm for} \quad | t | \le C \varepsilon ^ {- 1}, \label {effd1}
\\
\nonumber \\
&& \sup_ {t \in \mathbb R} \Big [\Vert \nabla [\psi (q (t) + x, t) - \psi_ {v (t)} (x)] \Vert_R 
+ \Vert \pi (q (t) + x, t) - \pi_ {v (t)} (x) \Vert_R \Big] \le C \varepsilon,
\label {effd2}
\end {eqnarray}
where
$ P (t) $ denotes  total momentum \eqref {P},  $ v (t) = \cP ^ {- 1} (\Pi (t)) $, and
$( Q (t),  \Pi (t)) $ is the solution to the effective Hamilton  equations \eqref {dyn} with initial conditions
$$
Q (0) = q (0), \qquad \Pi (0) = P (0).
$$
\end {theorem}
%%%%%%%%%%%%%%%%%%%%%%%%%%%%%
Note that such relevance of effective dynamics (\ref {dyn}) is due to the consistency of Hamiltonian structures:
\medskip \\
1) The effective Hamiltonian (\ref {Heff}) is a restriction of the Hamiltonian functional (\ref {Ham}) with $ V = 0 $
onto the soliton manifold  $ \cS $.
\medskip \\
2) As shown in \cite {KKS1999}, the canonical form of the Hamiltonian system
(\ref {dyn}) is also a restriction onto $ \cS $ of  canonical form of the   system \eqref {w3}--\eqref {q3}: formally
$$
P \, dQ = \Bigl [p \, dq + \int \psi (x) \, d \pi (x) \5 dx \Bigr] \Big | _ \cS.
$$
Therefore, the total momentum $ P $ is canonically conjugate to the variable $ Q $ on the soliton manifold $ \cS $.
This fact justifies definition (\ref {Heff}) of the effective Hamiltonian as a function of the total momentum  $ P_v $,
and not of the particle momentum $ p_v $.
\medskip

One of the important results of \cite {KKS1999} is  the following ``effective dispersion relation'':
\begin {equation} \label {EP}
E (\Pi) \sim \frac {\Pi ^ 2} {2 (1 + m_e)} + {\rm const}, \qquad | \Pi | \ll 1.
\end {equation}
It means that non-relativistic mass of  a slow soliton increases due to an interaction with the field by the amount
\begin {equation} \label {me}
m_e = - \frac 13 \langle \rho, \Delta ^ {- 1} \rho \rangle.
\end {equation}
This increment is proportional to the field energy of a soliton in rest
$$
\cH(\Delta ^ {- 1} \rho,0,0,0)=-\frac 12 \langle \rho, \Delta ^ {- 1} \rho 
\rangle,
$$
 which
agrees with the Einstein mass-energy equivalence principle (see below).

%%%%%%%%%%%%%%%%%%%%%%%%%%%%%
\begin {remark} \label {rad}
{\rm
The relation (\ref {EP}) gives only a hint that $ m_e $ is an increment of the effective mass.
The true {\it dynamical
justification} for such an interpretation is given by  the asymptotics (\ref {effd1})--(\ref {effd2})
which demonstrate the relevance of the effective dynamics (\ref{dyn}).
}
\end {remark}
%%%%%%%%%%%%%%%%%%%%%%%%%%%%%%%%%%%%%%%%%%%%%%
%%%%%%%%%%%%%%%%%%%%%%%%%%%%%%%%%%%%%%%%%%%%%%
\hspace{-6 mm}{\bf Generalizations.}
%%%%%%%%%%%%%%%%%%%%%%%%%%%%%%%%%%%%%%%%%%%%%% 
After the paper
  \cite {KKS1999} adiabatic effective asymptotics of  type (\ref {effd1}), (\ref {effd2})
were obtained in \cite {FTY2002, FGJS2004} for nonlinear Hartree 
and Schr\"odinger equations with slowly varying external potentials, and in \cite {LS2009,S2010} - for nonlinear equations of
Einstein's--Dirac, Chern--Simon--Schr\"odinger and Klein--Gordon--Maxwell
with small external fields.

Recently, similar adiabatic effective dynamics were established in \cite {BCFJS} for an electron in  second-quantized 
Maxwell field in the presence of a slowly changing external potential.

%%%%%%%%%%%%%%%%%%%%%%%%%%%%%%%%%%%%%%%%%%%%%%
%%%%%%%%%%%%%%%%%%%%%%%%%%%%%%%%%%%%%%%%%%%%%%
\subsection {Mass--Energy equivalence} \la {s4}
%%%%%%%%%%%%%%%%%%%%%%%%%%%%%%%%%%%%%%%%%%%%%%
In  \cite {KS2000ad}, asymptotics (\ref {effd1}), (\ref {effd2}) were extended to solitons of the Maxwell--Lorentz 
equations \eqref {ML} with small external fields. In this case the increment of  nonrelativistic mass 
also turns out to be proportional to the energy of the static soliton's own field.

Such equivalence of the self-energy of a  particle  with its mass was first discovered in 1902 by Abraham: 
he obtained by direct calculation that electromagnetic self-energy $ E _ {\rm own} $
of an electron at rest adds $ m_e = \displaystyle \frac43 E _ {\rm own} / c ^ 2 $ to its non-relativistic mass
(see \cite {A1902, A1905}, and also \cite [p. 216--217] {K2013}).
It is easy to see that this self-energy is infinite for a point electron 
at the origin
with a charge density $ \delta (x) $, 
because in this case, the Coulomb electrostatic field $ | E (x) | = C / | x | ^ 2 $ 
so the integral in (\ref {HAMml}) diverges around $x=0$.
This means that the field mass for a point electron is infinite, which contradicts experiment. 
That's why Abraham introduced the model  of electrodynamics  with  ``extended electron'' \eqref {ML}, 
whose self-energy is finite.

At the same time, Abraham conjectured that the {\it entire mass} of an electron is due to its own electromagnetic energy; 
that is, $ m = m_e $: ``... \textit {matter disappeared,
only energy remains} ... '', as philosophical-minded contemporaries  wrote \cite [p. ~ 63, 87, 88] {H1908}  (smile :) \,)

This conjecture was  justified in 1905 by Einstein, who discovered the famous universal relation $ E = m_0 c ^ 2 $, 
suggested by Special Theory of Relativity \cite {Einstein-1905}.
Additional factor $ \frac43 $ in the Abraham formula is due to nonrelativistic character of the system  \eqref {ML}.
According to modern view, about 80\% of the electron mass is of electromagnetic origin ~ \cite {F1966}.
 
%%%%%%%%%%%%%%%%%%%%%%%%%%%%%%%%%%%%%%%%%%%%%%
%%%%%%%%%%%%%%%%%%%%%%%%%%%%%%%%%%%%%%%%%%%%%%
\setcounter{equation}{0}
\section {Global attraction to stationary orbits} \la {s5}
%%%%%%%%%%%%%%%%%%%%%%%%%%%%%%%%%%%%%%%%%%%%%%
Global attraction to stationary orbits (\ref {atU}) was first established in
\cite {K2003, KK2006, KK2007} for the Klein--Gordon equation coupled to  nonlinear oscillator
\begin {equation} \label {KG1i}
\ddot \psi (x, t) = \psi '' (x, t) -m ^ 2 \psi (x, t) + \delta (x) F (\psi (0, t)), \qquad x \in \mathbb R.
\end {equation}
We consider complex solutions, identifying complex values
$ \psi \in \Co $ with real vectors $ (\psi_1, \psi_2) \in \mathbb R ^ 2 $,
where $ \psi_1 = \rRe \psi $ and $ \psi_2 = \rIm \psi $.
Suppose that $ F \in C ^ 1 (\mathbb R ^ 2, \mathbb R ^ 2) $ and
\begin {equation} \label {C1}
F (\psi) = - \nabla _ {\overline \psi} U (\psi), \qquad \psi \in \Co,
\end {equation}
where $ U $ is  real function and $ \nabla _ {\overline \psi}: = (\partial_1, \partial_2) $.
In this case, the equation \eqref {KG1} is formally equivalent to Hamiltonian system 
(\ref{hasys})
in the Hilbert
phase space $ \cE: = H ^ 1 (\mathbb R) \oplus L ^ 2 (\mathbb R) $, where 
 $H ^ 1  := H ^ 1 (\mathbb R) $ and $L ^ 2  := L ^ 2 (\mathbb R) $.
The  Hamilton functional reads
\begin {equation} \label {U}
\cH (\psi, \pi) = \frac12 \int \Big [
| \pi (x) | ^ 2 + | \psi '(x) | ^ 2
+ m ^ 2 | \psi (x) | ^ 2 \Big] \, dx + U (\psi (0)), \qquad (\psi, \pi) \in \cE.
\end {equation}
Let us write \eqref {KG1i} in the vector form as 
\be\la{KG1}
\dot Y(t)=\cF(Y(t)),\qquad t\in\R,
\ee
where $ Y (t) = (\psi (t), \dot \psi (t))$.
We assume that
\begin {equation} \label {C2}
\inf _ {\psi \in \Co} U (\psi)> - \infty.
\end {equation}
In this case,  finite energy solution  $ Y (t)  \in C (\mathbb R, \cE) $
exists and is unique for any initial state $ Y (0) \in \cE $.
A priori bound
\begin {equation} \label {apri}
\sup_ {t \in \mathbb R} [\Vert \dot \psi (t) \Vert_ {L ^ 2 (\mathbb R)} + \Vert \psi (t) \Vert_ {H ^ 1 (\mathbb R )}] <\infty
\end {equation}
holds due to  conservation of energy (\ref {U}). Note that the condition \eqref {conf} is no longer necessary,
since conservation of energy (\ref {U})  with $ m> 0 $ provides the boundedness of   solutions.

Further, we assume $ U (1) $ -invariance of the potential:
\begin {equation} \label {C3}
U (\psi) = u (| \psi |), \qquad \psi \in \Co.
\end {equation}
Then differentiation (\ref {C1}) gives
\begin {equation} \label {Fap}
F (\psi) = a (| \psi |) \psi, \qquad \psi \in \Co,
\end {equation}
and therefore
\be\la{U1}
F (e ^ {i \theta} \psi) = e ^ {i \theta} F (\psi), \qquad \theta \in \mathbb R.
\ee
By ``stationary orbits''  we mean solutions of the form
\be\la{storb}
\psi (x, t) = \psi_\om (x) e^{- i \omega t}
\ee
 with $\omega \in \R$ and $\psi_\omega \in H^1(\R)$. 
Each stationary orbit corresponds to some solution 
to the equation
$$
- \omega ^ 2 \psi_ \omega (x) = \psi '' _ \omega (x) -m ^ 2 \psi_ \omega (x) + \delta (x)
F (\psi_ \omega (0)), \qquad x \in \mathbb R,
$$
which is the  \textit {nonlinear eigenvalue problem}.
Solutions $ \psi_ \omega \in H ^ 1 (\R) $ of this equation have the form
$ \psi_ \omega (x) \! = \! Ce ^ {- \kappa | x |} $, where $ \kappa \!: = \! \sqrt {m ^ 2 \! - \! \omega ^ 2 } \!> \! 0 $, 
and the constant $ C $ satisfies the nonlinear algebraic equation
$
2 \kappa C = F (C).
$
The solutions $\psi_\omega$ exist for $\omega$ from some set $\Omega \subset \R$, lying in the \textit {spectral gap} $[- m, m]$. 
Denote the corresponding \textit {solitary manifold}
\begin {equation} \label {Sol}
{\mathcal S} = \{(e ^ {i \theta} \psi_ \omega, -i \omega e ^ {i \theta} \psi_ \omega) \in \cE: \omega \in \Omega, ~ \theta \in [0,2 \pi] \}.
\end {equation}
Finally, suppose the equation \eqref {KG1} be \textit {strictly nonlinear:}
\begin {equation} \label {C4}
U (\psi) = u (| \psi | ^ 2) = \sum_0 ^ {N} u_ {j} | \psi | ^ {2j}, \quad u_N> 0, \quad N \ge 2.
\end {equation}
For example, well known  \textit {Ginzburg--Landau} potential 
$ U (\psi) = | \psi | ^ 4 / 4- | \psi | ^ 2/2 $ satisfies all the conditions \eqref {C2}, \eqref {C3} and \eqref {C4}.

\bd i) ${\cal E}_F\subset H^1_{loc}(\R^3)\oplus L^2_{loc}(\R^3)$ is the space $\cE$
endowed with the seminorms
\be\la{semin}
\Vert Y\Vert_{\cE,R}:=\Vert Y\Vert_{H^1(-R,R)}+\Vert Y\Vert_{L^2(-R,R)},
\qquad R=1,2,\dots
\ee
\smallskip\\
ii) The convergence in ${\cal E}_F$ is equivalent to the convergence in every
seminorm  (\ref{semin}).
\ed
The convergence in  ${\cal E}_F$ is equivalent to the convergence in the metric of type (\ref{metr}),
\begin {equation} \label {metrU}
{\rm dist} [Y_1, Y_2] =
\sum_1 ^ \infty 2 ^ {- R} \displaystyle \frac {\Vert Y_1-Y_2 \Vert _ {\cE,R}}
{1+ \Vert Y_1-Y_2 \Vert _ {\cE,R}}, \qquad Y_1, Y_2 \in \cE.
\end {equation}

%%%%%%%%%%%%%%%%%%%%%%%%
\begin {theorem} \label {t5}
Let the conditions \eqref {C1}, \eqref {C2}, \eqref {C3} and \eqref {C4}  hold. 
Then any finite energy solution  $ Y (t) = (\psi (t), \dot \psi (t))\in C(\R,\cE) $ to equation \eqref {KG1} 
attracts to the solitary manifold {\rm (see Fig. \ref {fig-3})}:
\begin {equation} \label {ga}
Y (t) \toEF {\mathcal S}, \qquad t \to \pm \infty,
\end {equation}
where the attraction holds in the sense  \eqref {conv}.

\end {theorem}
%%%%%%%%%%%%%%%%%%%%%%%%%%%%%%%%%%%%%%%%%%%%%
%%%%%%%%%%%%%%%%%%%%%%%%%%%%%%%%%%%%%%%%%%%%%%%%%%%%%%%
\begin{figure}[htbp]
\begin{center}
\includegraphics[width=0.8\columnwidth]{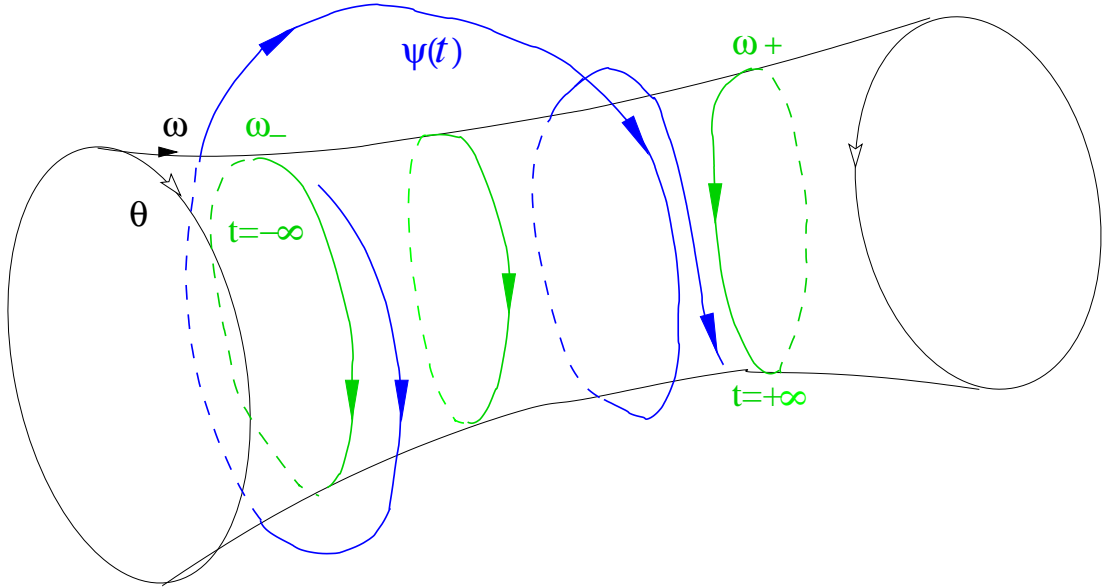}
\caption{Convergence to stationary orbits}
\label{fig-3}
\end{center}
\end{figure}
%%%%%%%%%%%%%%%%%%%%%%%%%%%%%%%%%%%%%%%%%%%%%%

\textbf {Generalizations:}
The attraction (\ref {ga}) was extended in \cite {KK2010a} to   1D Klein--Gordon  equation with $ N $ nonlinear oscillators
\be\la{KGN}
\ddot \psi (x, t) = \psi '' (x, t) -m ^ 2 \psi +
\sum \limits_ {k = 1} ^ N \delta (x-x_k) F_k (\psi (x_k, t)), ~ x \in \mathbb R,
\ee
and in \cite {C2012, KK2009, KK2010b} - to the Klein--Gordon and Dirac equations
in $ \R ^ n $
with $ n \ge 3 $ and ~ non-local interaction
\begin {eqnarray}
\ddot \psi (x, t) & = & \Delta \psi (x, t) -m ^ 2 \psi +
\sum_ {k = 1} ^ N \rho (x \! - \! x_k) F_k (\langle \psi (\cdot, t), \rho (\cdot -x_k) \rangle), \label {KGn} \\
\nonumber \\
i \dot \psi (x, t)
& = &
\big (-i \alpha \cdot \nabla + \beta m \big) \psi
+ \rho (x) F (\langle \psi (\cdot, t), \rho \rangle), \la{Dn}
\end {eqnarray}
under Wiener's condition \eqref {W1}, where $ \alpha = (\alpha_1, \dotsc, \alpha_n) $ and $ \beta = \alpha_0 $ are Dirac matrices.
\\
Recently, the attraction (\ref {ga}) was extended  in \cite{KK2019} to 1D Dirac equation coupled to  nonlinear oscillator,
and in \cite{K2017, K2018, KK2019b} to 3D wave and Klein-Gordon equations with concentrated nonlinearities.
\smallskip

In addition, attraction (\ref {ga}) was extended in \cite {C2013} to non-linear discrete in space and time
Hamiltonian
equations that are discrete approximations of equations of the type (\ref {KGn}), i.e. corresponding difference schemes.
The proof relies on a novel version of the Titchmarsh theorem for distributions on a circle, obtained in ~ \cite {KK2013t}.
\medskip

\textbf {Open questions:}
\medskip \\
I. Global attraction (\ref {atU}) to orbits with fixed frequencies $ \omega_ \pm $ is not proved yet.
\medskip \\
II. Global attraction to stationary orbits for nonlinear Schr\"odinger equations
also is not proved. In particular, such  attraction is not proved 
for the 1D Schr\"odinger equation associated with a nonlinear oscillator
\be\la{S1}
i \dot \psi (x, t) = - \psi '' (x, t) + \delta (x) F (\psi (0, t)), \qquad x \in \mathbb R.
\ee
The main difficulty is the infinite ``spectral gap'' $(-\infty,0)$ 
(see Remark \ref {rST}).
\medskip \\
III. Global attraction to solitons \eqref {att} for nonlinear
{\it relativistically invariant}
 Klein--Gordon equations is an open problem.
In particular, for one-dimensional equations
\be\la{NWEn}
\ddot \psi (x, \! t) = \psi '' (x, t) -m ^ 2 \psi (x, t) + F (\psi (x, t)).
\ee
The main difficulty is the presence of nonlinear interaction in every point $x\in\R$.
Asymptotic stability of solitons
(that is, {\it local attraction} to them) for such equations was first proved in \ci{KopK2011a, KopK2011b}, see Section 6.3 below.

\subsection{Method of omega-limit trajectories} 

The proof of Theorem \ref{t5} relies on a general strategy
 of {\it omega-limit trajectories}
 introduced first
in \ci{K2003}
and developed further in 
\cite {KK2006,KK2007,KK2010a,KK2008,KK2009,KK2010b,KK2013t,
K2017, KK2019b,KK2019,C2012,C2013}.

\bd\la{domt}
An omega-limit trajectory for a given 
$Y(t)\in C(\R,\cE) $ is any limit function $Z(t)$ such that
\be\la{omt}
Y(t+s_j)\toEF   Z(t)
%=(\beta(t),\dot\beta(t))
,\qquad t\in\R,
\ee
where $s_j\to\infty$.

\ed

%% ?? Let us note that any omega-limit trajectory lies on the attractor.

\bd
A function $Y (t)\in C(\R,\cE) $ is omega-compact if for any sequence $s_j\to\infty$
there exists such a subsequence $s_{j'}\to\infty$ that (\ref{omt}) holds.

\ed

These concepts are useful due to the following lemma
which lies in the basis our approach.

\bl\la{lomtr}
Let any solution $Y (t)\in C(\R,\cE) $ to (\ref{KG1}) be omega-compact,
and any omega-limit trajectory is a stationary orbit
\be\la{gab4}
Z(x,t)=(\psi_\om(x)e^{-i\om t},-i\om \psi_\om(x)e^{-i\om t}),
\ee
where $\om\in\R$. Then the attraction (\ref{ga}) holds for 
each solution $Y (t)\in C(\R,\cE) $ to (\ref{KG1}).
\el
\begin{proof}
We need to show that  
$$
\lim_{t\to\infty}{\rm dist}(Y(t),{\cal S})=0.
$$
Assume by contradiction that there exists a sequence $s_j\to \infty$ such that
\begin{equation}\label{cal-C}
{\rm dist}(Y(s_j),{\cal S})\ge\delta>0,\quad\forall j\in\N.
\end{equation}
  According to the omega-compactness of the solution $Y$,  the convergence
  (\ref{omt}) holds for some
   subsequence $s_{j'}\to\infty$, 
and some stationary orbit   (\ref{gab4}):
\be\la{Ysj}
Y(t+s_j)\toEF Z(t),\qquad t\in\R.
\ee
 But 
this convergence with $t=0$  contradicts \eqref{cal-C}
since $Z(0)\in\cS$ by definition (\ref{Sol}).
\end{proof}
Now for the proof of  Theorem \ref {t5} is suffices to check the conditions
of Lemma \ref{lomtr}:
\medskip 

I. Each solution $Y (t)\in C(\R,\cE) $ to (\ref{KG1}) is omega-compact.
\medskip 

II. Any omega-limit trajectory is a stationary orbit (\ref{gab4}).
\medskip 

We check these conditions analysing the Fourier transform in time of solutions.
The main steps of the proof are as follows:
\medskip \\
(1) Spectral representation  for solutions to  nonlinear equation (\ref {KG1}):
\be\la{FL}
\psi(t)=\fr 1{2\pi}\int e^{-i\om t}\ti\psi(\om)d\om.
\ee
We call \textit {spectrum} of a solution $\psi(t):=\psi(\cdot,t)$ the support of its spectral density $\ti\psi(\cdot)$ which is
a tempered distribution of $\om\in\R$ with the values in $H^1$.
\smallskip \\
(2) {\it Absolute continuity} of the spectral density $\ti\psi(\om)$
{\it on the continuous spectrum} $(-\infty,-m)\cup(m,\infty)$ of the free Klein--Gordon equation, which is an analogue of the Kato theorem on the absence of embedded eigenvalues.
\smallskip \\
(3) {\it Omega-limit compactness} of each solution.
\smallskip \\
(4) Reduction of spectrum {\it of each omega-limit trajectory} to a subset of 
the {\it spectral gap} $[-m,m]$.
\smallskip \\
(5) Reduction of this spectrum to a {\it single point} using the {\it Titchmarsh convolution theorem}.
\smallskip

Below we follow this program,
referring at some points to \cite {K2003, KK2007} for 
technically important 
properties of quasi-measures. 

%%%%%%%%%%%%%%%%%%%%%%%%%%%%%%%%%%%%%%%%%%%%%%%%%%%%%%%%
%%%%%%%%%%%%%%%%%%%%%%%%%%%%%%%%%%%%%%%%%%%%%%%%%%%%%%%%
\subsection {Spectral representation and limiting absorption principle}
%%%%%%%%%%%%%%%%%%%%%%%%%%%%%%%%%%%%%%%%%%%%%%%%%%%%%%%%
It suffices to prove attraction (\ref {ga}) only for positive times.
For the simplicity of exposition
we consider 
the solution  $\psi(x,t)$ to equation (\ref{KG1i})
corresponding to zero initial data only:
\begin {equation} \label {ini0}
\psi (x, 0) = 0, \qquad \dot \psi (x, 0) = 0.
\end {equation}
 General case of nonzero initial data can be reduced
to this case by a
trivial subtraction of the 
   solution to the free Klein--Gordon equation  with these initial data
   which is a dispersion wave
    \cite {K2003,KK2007}.
We extend $ \psi (x, t) $ and $ f (t): = F (\psi (0, t)) $ by zero for $ t <0 $ and denote
\begin {equation} \label {gap2}
\psi _ + (x, t): =
\left \{
\begin {array} {ll}
\psi (x, t), & t> 0, \\
0, & t <0,
\end {array} \right.
\qquad f _ + (t): =
\left \{
\begin {array} {ll}
f (t), & t> 0, \\
0, & t <0.
\end {array} \right.
\end {equation}
From \eqref {KG1} and (\ref {ini0}) it follows that these functions satisfy the equation 
\begin {equation} \label {KG2}
\ddot \psi _ + (x, t) = \psi _ + '' (x, t) -m ^ 2 \psi _ + (x, t) + \delta (x) f _ + (t), \qquad (x, t) \in \mathbb R ^ 2
\end {equation}
in the sense of distributions.
\medskip \\
%%%%%%%%%%%%%%%%%%%%%%%%%%%%%%%%%%%%%%%%%%%%%%%%
{\bf Fourier-Laplace transform in time.}
%%%%%%%%%%%%%%%%%%%%%%%%%%%%%%%%%%%%%%%%%%%%%%%%
For tempered  distributions $ g (t) $
we denote by $ \tilde g (\omega) $ their Fourier transform,
which  is defined for $ g \in C_0 ^ \infty (\mathbb R) $ as
$$
\tilde g (\omega) = \int_ \mathbb R e ^ {i \omega t} g (t) \, dt, \qquad \omega \in \mathbb R.
$$
A priori estimates \eqref {apri} imply that
 $ \psi _ + (x, t) $ and $ f _ + (t) $ are bounded functions of $ t \in \mathbb R $
with values
in the Sobolev space $ H ^ 1 (\mathbb R) $ and in $ \Co $, respectively.
Therefore, their Fourier transforms are  (by definition) {\it quasi-measures}
with values
in $ H ^ 1 (\mathbb R) $ and in $ \Co $, respectively \ci {Gaudry1966}.
Moreover,
these Fourier transforms allow an extension from the real axis to analytic functions
in the upper complex half-plane $\Co ^ +: = \{\omega \in \Co: ~ \rIm \omega> 0 \} $ with values in $ H ^ 1 (\mathbb R) $ and in $ \Co $ respectively:
$$
\tilde \psi _ + (x, \omega) = \int_0 ^ \infty e ^ {i \omega t} \psi (x, t) \, dt, \qquad \tilde f _ + (\omega) = \int_0 ^ \infty e ^ {i \omega t} f (t) \, dt, \qquad \omega \in \Co ^ +.
$$
Further, 
we have the following convergence of tempered distributions 
with values in $H^1$ and $\Co$ respectively,
$$
e^{-\ve t}\psi _ + (x, t)\to \psi _ + (x, t),\qquad e^{-\ve t}f _ + (t)\to f_ + (t),
\qquad \ve\to 0+.
$$
Hence, also their Fourier transforms converge in the same sense, 
\be\la{bv}
\ti\psi _ + (x, \om+i\ve)\to \ti\psi _ + (x, \om),\qquad \ti f _ + (\om+i\ve)\to \ti f_ + (\om),
\qquad \ve\to 0+.
\ee
The analytic functions 
$\ti\psi_+(x,\om)$ and $\ti f(\om)$
grow (in the norm) not faster than $ | \rIm \omega | ^ {- 1} $ as
$ \rIm \omega \to 0 + $
in view of (\ref {apri}). Hence,
their boundary values at $\omega \in \mathbb R $ are tempered distributions of a small singularity:
they are second order derivatives of continuous functions, as in the case of $ \tilde f _ + (\omega) = i /( \omega- \omega_0) $ with $ \omega_0 \in \R $,
which corresponds to
$ f _ + (t) = \theta (t) e ^ {- i \omega_0 t} $.
\medskip \\
{\bf Limiting Absorption Principle.} By \eqref {ini0} the equation
(\ref {KG2}) in the Fourier transform becomes  stationary Helmholtz equation
\be\la{Helm}
- \omega ^ 2 \tilde \psi _ + (x, \omega) = \tilde \psi _ + '' (x, \omega) -m ^ 2 \tilde \psi _ + (x, \omega) + \delta (x ) \tilde f _ + (\omega), \qquad x \in \mathbb R.
\ee
This equation has two linearly independent solutions, but only one of these solutions is analytic and bounded in 
$ \rIm \omega> 0 $ with values in $ H ^ 1 (\mathbb R) $:
\begin {equation} \label {KG4}
\tilde \psi _ + (x, \omega) = - \tilde f _ + (\omega) \frac {e ^ {ik (\omega) | x |}} {2ik (\omega)}, \qquad \rIm \omega> 0.
\end {equation}
Here $ k (\omega): = \sqrt {\omega ^ 2-m ^ 2} $, where the branch has 
a positive imaginary part for $ \rIm \omega> 0 $.
For other branch, this function {\it grows exponentially} as $ | x | \to \infty $.
Such an argument in the selection of  solutions
to stationary Helmholtz equations
 is known as the ``limiting absorption principle''
in the diffraction theory \cite {KopK2012, KM2019}.
\medskip \\
{\bf Spectral representation.}
We rewrite (\ref {KG4}) in the form
\begin {equation} \label {KG5}
\tilde \psi _ + (x, \omega) = \tilde \alpha (\omega) e ^ {ik (\omega) | x |}, \qquad \rIm \omega> 0;
\qquad  \alpha (t): = \psi _ + (0, t).
\end {equation}
A nontrivial fact is that the identity of
analytic functions (\ref {KG5})
keeps its structure for their restrictions
onto the real axis:
\begin {equation} \label {KG6}
\tilde \psi _ + (x, \omega + i0) = \tilde \alpha (\omega + i0) e ^ {ik (\omega + i0) | x |}, \qquad \omega \in \mathbb R,
\end {equation}
where $ \tilde \psi _ + (\cdot, \omega + i0) $ and $ \tilde \alpha (\omega + i0) $ are the corresponding
quasi-measures with values in $ H ^ 1 (\mathbb R) $ and $ \Co $, respectively.
The problem is that the factor $ M_x (\omega): = e ^ {ik (\omega + i0) | x |} $ is not smooth
in $ \omega $
at the points $ \omega = \pm m $. Respectively, the identity (\ref {KG6}) requires a justification, 
based on the quasi-measure theory
\ci {KK2007}.
\medskip

Finally, the inversion of the Fourier transform can be written as
\begin {equation} \label {FLi}
\psi _ + (x, t) = \frac1 {2 \pi} \langle \tilde \psi _ + (x, \omega + i0), e ^ {- i \omega t} \rangle
= \frac1 {2 \pi} \langle
\tilde \alpha (\omega + i0) e ^ {ik (\omega + i0) | x |}, e ^ {- i \omega t} \rangle, \quad x, t \in \R,
\end {equation}
where$ \langle \, \cdot {,} \, \cdot \rangle $ is a bilinear duality between distributions and smooth bounded functions. 
The right hand side  exists by the Theorem \ref{Kato}, see below.

%%%%%%%%%%%%%%%%%%%%%%%%%%%%%%%%%%%%%%%
%%%%%%%%%%%%%%%%%%%%%%%%%%%%%%%%%%%%%%%%
\subsection {Nonlinear analogue of Kato's theorem}\la{sKT}
%%%%%%%%%%%%%%%%%%%%%%%%%%%%%%%%%%%%%%%%
It turns out that the properties of the quasimeasures $ \tilde \alpha (\omega + i0) $ with $ | \omega | <m $ and with $ | \omega |> m $
significantly differ. This is due to the fact that the set $ \{i \omega: | \omega | \ge m \} $ is the 
continuous spectrum of the generator
$$
A = \left (\begin {array} {cc} 0 & 1 \\
\frac {d ^ 2} {dx ^ 2} -m ^ 2 & 0
\end {array}\right),
$$
which is the generator of the linearisation  of  equation \eqref {KG1}. The following theorem plays a key role in the proof of Theorem \ref{t5}.
It is a nonlinear analogue of Kato's theorem on the absence of embedded eigenvalues in the continuous spectrum, see Remark \ref{rK} below.
Denote $ \Sigma: = \{\omega \in \mathbb R: | \omega |> m \} $, and we will write below $ \tilde \alpha (\omega) $
and $ k (\omega) $ instead of $ \tilde \alpha (\omega + i0) $ and $ k (\omega + i0) $ for $ \omega \in \R $.

%%%%%%%%%%%%%%%%%%%%%%%%%%%%%%%
\bt \label {lKato} {\rm (\cite [Proposition 3.2] {KK2007})}
Let conditions \eqref {C1}, \eqref {C2} and \eqref {C3} hold, and $\psi (t)\in C (\R, \cE)$ is any finite energy solution to the equation \eqref {KG1}. 
Then the corresponding tempered distribution $ \tilde \alpha (\omega) $ is absolutely continuous on $ \Sigma $. 
Moreover, $ \al \in L ^ 1 (\Si) $ and
\begin {equation} \label {Kato}
\int_ \Sigma | \tilde \alpha (\omega) | ^ 2 \5 | \omega \5 k (\omega) | \5 d \omega <\infty.
\end {equation}
\et
%%%%%%%%%%%%%%%%%%%%%%%%%%%%%%%%%
\begin{proof}
First,
let us first explain the main idea of the proof.
By (\ref{FLi}), the function $\psi_+(x,t)$ formally is a ``linear combination'' of the functions $e\sp{ik|x|}$ with the
{\it amplitudes} $\hat z(\omega)$:
$$
\psi_+(x,t)=\frac 1{2\pi}\int\sb\R \hat z(\omega)e\sp{ik(\omega)|x|}e\sp{-i\omega t}\,d\omega,\qquad x\in\R.
$$
For $\omega\in\Si $, the functions $e\sp{ik(\omega)|x|}$ are of infinite $L^2$-norm, while $\psi_+(\cdot,t)$ is of finite $L^2$-norm.
This is possible only if the amplitude is absolutely continuous in $\Si $.
This idea is  suggested by the Fourier integral $f(x)=\ds\int_\R e^{-ikx}g(k)dk$ which belongs to $L^2(\R)$ if and only if $g\in L^2(\R)$.
For example, if one took 
$\hat z(\omega)=\delta(\omega-\omega\sb 0)$ with $\omega\sb 0\in\Si $, then $\psi_+(\cdot,t)$ would be of infinite $L^2$-norm.
\smallskip

 The rigorous proof relies on estimates of the Paley-Wiener type. Namely, the Parseval identity and (\ref{apri}) imply that
\begin{equation}\label{PW}
\int\limits\sb\R\norm{\tilde\psi_+(\cdot,\omega+i\varepsilon)}\sb{H\sp 1}^ 2\,d\omega=2\pi\int\limits\sb 0\sp\infty e\sp{-2\varepsilon t}
\norm{\psi_+(\cdot,t)}\sb{H\sp 1}^ 2\,dt\le\frac{\const}{\varepsilon},\quad \varepsilon>0.
\end{equation}
On the other hand, we can estimate exactly the integral on the left-hand side of (\ref{PW}). Indeed, according to (\ref{FLi}),
\[
\tilde\psi_+(\cdot,\omega+i\varepsilon)=\ti\al (\omega+i\varepsilon)e^{ik(\omega+i\varepsilon)|x|}.
\]
Hence, (\ref{PW}) results in
\begin{equation}\label{PW-1}
\varepsilon\int\sb\R|\ti\al (\omega+i\varepsilon)|^2\norm{e^{ik(\omega+i\varepsilon)|x|}}\sb{H\sp 1}\sp 2\,d\omega\le\const,
\qquad \varepsilon>0.
\end{equation}

Here is a crucial observation about the  asymptotics of the norm of 
$e^{ik(\omega+i\varepsilon)|x|}$ as $\ve\to 0+$.
%%%%%%%%%%%%%%%%%%%%%%%%%%%
\begin{lemma}
\begin{enumerate}
\item
For $\omega\in\R$,
\begin{equation}\label{wei}
\lim\sb{\varepsilon\to 0+}\varepsilon \norm{e^{ik(\omega+i\varepsilon)|x|}}\sb{H\sp 1}\sp 2=n(\omega):=
\left\{\begin{array}{ll}\omega k(\omega),&|\omega|>m\\0,&|\omega|<m
\end{array}\right.,
\end{equation}
where the norm in $H\sp 1$ is chosen to be
$\norm{\psi}\sb{H\sp 1}=\left(\norm{\psi'}\sb{L\sp 2}^2+m^2\norm{\psi}\sb{L\sp 2}^2\right)\sp{1/2}.$
\item
For any $\delta>0$ there exists such $\varepsilon\sb\delta>0$  that for $|\omega|>m+\delta$ and $\varepsilon\in(0,\varepsilon\sb\delta)$,
\begin{equation}\label{n-half}
\varepsilon \norm{e^{ik(\omega+i\varepsilon)|x|}}\sb{H\sp 1}\sp 2\ge n(\omega)/2.
\end{equation}
\end{enumerate}
\end{lemma}
%%%%%%%%%%%%%%%%%%%%%%%%%%%%%%%%
\begin{proof}
Let us compute the $H\sp 1$-norm using the Fourier space representation. Setting $k\sb\varepsilon=k(\omega+i\varepsilon)$,
so that $\rIm k\sb\varepsilon>0$, we get
$F\sb{x\to k}\left[e^{ik\sb \varepsilon|x|}\right]=2ik\sb\varepsilon/(k\sb \varepsilon^2-k^2)$ for $k\in\R$.
Hence, by the Cauchy theorem
$$
\norm{e^{ik\sb \varepsilon|x|}}\sb{H\sp 1}\sp 2=\frac{2|k\sb \varepsilon|^2}\pi\int\sb{\R}\frac{(k^2+m^2)dk}{|k\sb \varepsilon^2-k^2|^2}
=-4\,\rIm\!\left[\frac{(k\sb \varepsilon^2+m^2)\Bar{k\sb \varepsilon}}{k\sb \varepsilon^2-\Bar{k\sb \varepsilon}^2}\right].
$$
Substituting here  $k\sb\varepsilon^2=(\omega+i\varepsilon)^2-m^2$, we get
$$
\norm{e^{ik(\omega+i\varepsilon)|x|}}\sb{H\sp 1}\sp 2=\frac{1}{\varepsilon}
\rRe\left[\frac{(\omega+i\varepsilon)^2\overline{k(\omega+i\varepsilon)}}{\omega}\right],
\quad\varepsilon>0, \quad \omega\in\R,\quad \omega\ne 0.
$$
Now the limits (\ref{wei}) follow since the function $k(\omega)$ is real for $|\omega|>m$, but is purely imaginary for  $|\omega|<m$.
Hence, the second statement of the Lemma also follows since $n(\omega)>0$ for $|\omega|>m$, and
$n(\omega)\sim|\omega|^2$ for $|\omega|\to\infty$.
\end{proof}

%%%%%%%%%%%%%%%%%%%%%%%%%%%%%%
\begin{remark} 
Obviously,
$n(\omega)\equiv 0$  for $|\omega|<m$ without any calculations, since in that case the function
$e^{ik(\omega)|x|}$ decays exponentially in $x$, and hence, the $H\sp 1$-norm of $e^{ik(\omega+i\varepsilon)|x|}$
remains finite when $\varepsilon\to 0+$.
\end{remark}

Substituting (\ref{n-half}) into (\ref{PW-1}), we get:
\begin{equation}\label{fin}
\int\sb{\Si_\de }|\ti\al (\omega+i\varepsilon)|\sp 2\omega k(\omega)\,d\omega\le 2C,
\qquad 0<\varepsilon<\varepsilon\sb\delta,
\end{equation}
with the same $C$ as in (\ref{PW-1}), and the region $\Si_\de:= \{\om\in\R:|\om|>m+\de\}$.
We conclude that for each $\delta>0$ the set of functions
\[
g_\ve(\omega)=\ti\al (\omega+i\varepsilon)|\omega k(\omega)|^{1/2},
\qquad\varepsilon\in(0,\varepsilon\sb\delta),
\]
is bounded in the Hilbert space $L\sp 2(\Si_\de )$, and, by the Banach Theorem, is weakly compact.
Hence,
the convergence of the distributions (\ref{bv}) implies the  weak convergence in the Hilbert space $L\sp 2(\Si_\de )$:
$$
g_\ve \rightharpoonup g,\qquad \varepsilon\to 0+,
$$
where the limit function $g(\omega)$ coincides with the distribution
$\hat z(\omega)|\omega k(\omega)|^{1/2}$ restricted onto ${\Si_\de }$.
It remains to note that the norms of  $g$ 
in $L\sp 2(\Si_\de )$ with all $\de>0$ are
 bounded  by (\ref{fin}),
which implies (\ref{Kato}). Finally, $\ti \al(\omega)\in L\sp 1(\Bar{\Si })$ by (\ref{Kato}) and the Cauchy-Schwarz inequality.
\end{proof}

\br\la{rK}
Theorem \ref{lKato} is a nonlinear analogue of Kato's theorem on the absence of embedded eigenvalues in the continuous spectrum. Indeed,  solutions of type $\psi_*(x)e^{-i\om_* t}$ become
%$\psi(x)\de(\om-\om_*)$ in the Fourier transform (and 
$\psi_*(x)[\pi i\de(\om-\om_*)+v.p. \fr 1{i(\om-\om_*)}]$
in the Fourier--Laplace transform
that is forbidden for $|\om_*|>m$ by Theorem \ref{lKato}.
\er
%%%%%%%%%%%%%%%%%%%%%%%%%%%%%%%%%%%%%%%
%%%%%%%%%%%%%%%%%%%%%%%%%%%%%%%%%%%%%%%%
\subsection {Splitting onto dispersion and bound components}
%%%%%%%%%%%%%%%%%%%%%%%%%%%%%%%%%%%%%%%%
Theorem \ref {lKato} suggests a splitting of the  solutions (\ref {FLi}) onto a ``dispersion'' and a ``bound'' components
\begin {eqnarray}
\psi _ + (x, t) & = & \frac1 {2 \pi} \int_ \Sigma (1- \zeta (\omega)) \tilde \alpha (\omega) e ^ {ik (\omega) | x |} e ^ {- i \omega t} d \omega +
\frac1 {2 \pi} \langle \zeta (\omega) \tilde \alpha (\omega) e ^ {ik (\omega) | x |}, e ^ {- i \omega t} \rangle
\nonumber \\
\nonumber \\
& = & \psi_d (x, t) + \psi_b (x, t), \qquad t> 0, \qquad x \in \mathbb R, \la{pdb}
\end {eqnarray}
where
$$
\zeta (\omega) \in C_0 ^ \infty (\mathbb R), \qquad \zeta (\omega) = 1 \quad \mbox {\rm when} \quad \omega \in [-m-1, m +1].
$$
Note that $ \psi_d (x, t) $ is a dispersion wave, because
$$
\psi_d (x, t): = \frac1 {2 \pi} \int_ \Sigma (1- \zeta (\omega)) e ^ {- i \omega t} \tilde \alpha (\omega) e ^ { ik (\omega) | x |} d \omega \to 0, \qquad t \to \infty
$$
according to the Riemann--Lebesgue theorem, since $ \al \in L ^ 1 (\Si) $ by Theorem \ref {lKato}. Moreover, it is easy to prove that
\begin {equation} \label {psid}
(\psi_d (\cdot, t), \dot \psi_d (\cdot, t)) \to 0, \qquad t \to \infty
\end {equation}
in the seminorms \eqref {cER}. Therefore, it remains to prove the attraction
(\ref {ga}) for $ Y_b (t): = (\psi_b (\cdot, t), \dot \psi_b (\cdot, t)) $ instead of $ Y (t) $:
\begin {equation} \label {gapb}
Y_b (t) \tocEF {\mathcal S}, \qquad t \to \infty.
\end {equation}
%%%%%%%%%%%%%%%%%%%%%%%%%%%%%%%%%%%%%%%%%%%%%
%%%%%%%%%%%%%%%%%%%%%%%%%%%%%%%%%%%%%%%%%%%%%
\subsection {Omega-compactness}
%%%%%%%%%%%%%%%%%%%%%%%%%%%%%%%%%%%%%%%%%%%%%%%

Here we establish  the omega-compactness of the trajectrory $Y_b(t)$ that
is necessary for the application of Lemma \ref{lomtr}.
First, we note that the bound 
component $ \psi_b (x, t) $ is a smooth function for $ x \ne 0 $, and
\begin {equation} \label {gab2}
\partial_x ^ j \partial_t ^ l \psi_b (x, t) = \frac1 {2 \pi} \langle \zeta (\omega) (ik (\omega) \operatorname {sgn} \5 x) ^ j \tilde \alpha (\omega) e ^ {ik (\omega) | x |},
(-i \omega) ^ l e ^ {- i \omega t} \rangle,
\qquad t> 0, \quad x \ne 0
\end {equation}
for any $j,l=0,1,\dots$.
These formulas should be justified
since the function $ k (\omega) $ is not smooth at the points $ \om = \pm m $.
The needed justification is done in  \ci {K2003, KK2007} by a suitable development of 
the theory of  quasimeasures.
These formulas imply the boundedness of each derivative:
%%%%%%%%%%%%%%%%%%%%%%%
\begin {lemma} \label {lbound} {\rm (\cite [Proposition 4.1] {KK2007})}
For all $ j, l = 0,1,2, \dotsc $ and $ R> 0 $
\begin {equation} \label {gab22}
\sup_ {x \ne 0} \, \, \sup_ {t \in \mathbb R} | \partial_x ^ j \partial_t ^ l \psi_b (x, t) | <\infty.
\end {equation}
\end {lemma}
%%%%%%%%%%%%%%%%%%%%%%%%%%%%%%%%
\begin {proof}
Note that the distribution $ \tilde \alpha (\omega) $ generally is not a finite measure, since we only know that $ \al (t): = \psi _ + (0, t) $   is a bounded
function by (\ref{KG5})
and (\ref {apri}). To prove the lemma, it suffices to check that
$$
\zeta (\omega) (ik (\omega) \operatorname {sgn} \5 x) ^ je ^ {ik (\omega) | x |} (- i \omega) ^ l = \tilde g_x (\omega),
$$
where the  function $g_x(\cdot)$ belongs to a bounded subset of $L^ 1(\mathbb R)$ for $ x \ne 0 $ and $ t \in \R $.
This implies the lemma, since the right-hand side of  (\ref {gab2}) equals, by the Parseval identity, to convolution
$$
\langle \al (t-s), g_x(s) \rangle,
$$
 where  $ \al (t) $ is a bounded function.
\end {proof}
%%%%%%%%%%%%%%%%%%%%%%
\br
 All needed properties of quasimeasures  that we use, are justified in \ci {K2003, KK2007}
by 
similar arguments relying on  the Parseval identity.
\er
%%%%%%%%%%%%%%%%%%%%%%%

Now, by the Ascoli--Arzella theorem, for any sequence $ s_j \to \infty $ there is such a subsequence $ s_ {j '} \to \infty $, that
\begin {equation} \label {gab3}
\partial_x ^ j \partial_t ^ l \psi_b (x, s_ {j '} + t) \to \partial_x ^ j \partial_t ^ l \beta (x, t), \qquad x \ne 0, \, \, t \in \R
\end {equation}
for any $ j, l = 0, \dots, $ and this convergence is uniform on $|x|+|t|\le R$. 
%We call any such function $ \beta (x, t) $ \textit {omega-limit trajectory}
%of the bound component $ \psi_b (x, t) $. 
Estimates (\ref {gab22}) imply that
\be\la{betab}
\sup _ {(x, t) \in \mathbb R ^ 2} | \partial_x ^ j \partial_t ^ l \beta (x, t) | <\infty.
\ee
%%%%%%%%%%%%%--Lemma
\bc
Each solution $Y (t)\in C(\R,\cE) $ to (\ref{KG1}) is omega-compact. This follows from 
(\ref{pdb}), (\ref{psid}) and (\ref{gab3}).

\ec
%%%%%%%%%%%%%%%%%%%%%%%%%%%%%%%%%%%%%%%%%%%%%%%%%%%%%%%%
%%%%%%%%%%%%%%%%%%%%%%%%%%%%%%%%%%%%%%%%%%%%%%%%%%%%%%%%%
\subsection {Reduction of spectrum  of omega-limit trajectories to spectral gap}
%%%%%%%%%%%%%%%%%%%%%%%%%%%%%%%%%%%%%%%%%%%%%%%%%%%%%%%%%
The convergence of functions (\ref {gab3}) implies the convergence of their Fourier transforms
\begin {equation} \label {gab5}
\tilde \psi_b (x, \omega) e ^ {- i \omega s_ {j '}} \to \tilde \beta (x, \omega), \qquad \forall x \in \mathbb R
\end {equation}
in the sense of  temperate  distributions of $ \om \in \R $.
%%%%%%%%%%%%%%%
\bl\la{lmms}
For any $x\in\R$
\be\la{mms}
\ti\beta(x,\om)=0,\qquad |\om|>m.
\ee
\el
%%%%%%%%%%%%%%%%%%
\begin{proof}
Convergence (\ref{gab5}) and  representation (\ref {gab2}) with $j=l=0$ imply that
\begin {equation} \label {gab6}
\zeta (\omega) \tilde \alpha (\omega) e ^ {ik (\omega) | x |} e ^ {- i \omega s_ {j '}} \to \tilde \beta (x, \omega ), \qquad \forall x \in \mathbb R
\end {equation}
in the sense of  temperate distributions of $ \om \in \R $.
Moreover, this convergence takes place in a stronger {\it Ascoli--Arzella topology} in the space of quasimeasures \cite {KK2007}.
In addition, $ e ^ {- ik (\omega) | x |} $ is a multiplier in the space of quasimeasures with this topology by Lemma B.3 of \cite {KK2007}). 
Therefore, \eqref {gab6} implies that
\begin {equation} \label {gab7}
\zeta (\omega) \tilde \alpha (\omega) e ^ {- i \omega s_ {j '}} \to \tilde \gamma (\omega): = \tilde \beta (x, \omega) e ^ {- ik (\omega) | x |}, \qquad \forall x \in \mathbb R
\end {equation}
in the same topology of quasimeasures. Applying the same lemma again, we obtain
\begin {equation} \label {gab8}
\beta (x, t) = \frac1 {2 \pi} \langle \tilde \gamma (\omega) e ^ {ik (\omega) | x |}, e ^ {- i \omega t} \rangle, \qquad (x, t) \in \mathbb R ^ 2.
\end {equation}
 Note that
\begin {equation} \label {gab9}
\beta (0, t) = \gamma (t).
\end {equation}
Finally, the key observation is that \eqref {gab7} and the theorem \ref {lKato} imply
\begin {equation} \label {gab10}
\supp \tilde \gamma \subset [-m, m]
\end {equation}
by the Riemann – Lebesgue theorem.
\end{proof}

%%%%%%%%%%%%%%%%%%%%%%%%%%%%%%%%%%%%%%%%%%%%%%%%%
%%%%%%%%%%%%%%%%%%%%%%%%%%%%%%%%%%%%%%%%%%%%%%%%%
\subsection{Reduction of spectrum  of omega-limit trajectories to a single point}

\subsubsection {Equation for omega-limit trajectories and spectral inclusion}
%%%%%%%%%%%%%%%%%%%%%%%%%%%%%%%%%%%%%%%%%%%%%%%%%
Now the question arises about   available means for the  proof of representation \eqref {gab4}
for omega-limit trajectories.
We  have no formulas for solutions to  equation \eqref {KG1}, and so the only hope is to use the 
nonlinear equation itself.
 The  key observation, albeit simple, is that $\beta(x,t)$ is a solution to this 
nonlinear
equation {\it for all $ t \in \R $}, 
despite the fact that $ \psi _ + (x, t) $ is a solution to the equation \eqref {KG1} only for $ t> 0 $ due to (\ref {gap2}).
%%%%%%%%%%%%%%%%%%%%%
\begin {lemma} \label {leq}
The function $\beta(x,t)$ 
 satisfies the original equation
\eqref {KG1}:
\begin {equation} \label {KG11}
\ddot \beta (x, t) = \beta '(x, t) -m ^ 2 \beta (x, t) + \delta (x) F (\beta (0, t)), \qquad ( x, t) \in \mathbb R ^ 2.
\end {equation}
\end {lemma}
%%%%%%%%%%%%%%%%%%%%%%%
\begin{proof}
This lemma follows 
by \eqref {psid} and \eqref {gab3}
in the limit  $ s_ {j '} \to \infty $ 
from the equation (\ref {KG1}) for
$ \psi _ + (x, s_ {j '} + t) = \psi_d (x, s_ {j'} + t) + \psi_b (x, s_ {j '} + t) $ with  $ s_ {j '} + t> 0 $.
\end{proof}
%%%%%%%%%%%%%%%%%%%%%%%

Now applying the Fourier transform to the equation (\ref {KG11}), we get the corresponding  ``nonlinear stationary  Helmholtz equation''
\begin {equation} \label {Fap1}
- \omega^2 \tilde\beta(x, \omega) = \tilde \beta '' (x, \omega) -m ^ 2 \tilde \beta (x, \omega) + \delta(x)\tilde f(\omega), \qquad (x,\omega) \in\mathbb R^2,
\end {equation}
where we denote  $ f (t): = F (\beta (0, t)) = F (\gamma (t)) $ in accordance with (\ref {gab9}).
From (\ref {Fap}), we get
$$
f (t) = a (| \gamma (t) |) \gamma (t) = A (t) \gamma (t), \qquad A (t): = a (| \gamma (t) |), \qquad t \in \mathbb R.
$$
Finally, in the Fourier transform we get the convolution $ \tilde f = \tilde A * \tilde \gamma $, which exists by (\ref {gab10}). Respectively,
(\ref {Fap1}) now reads
$$
- \omega ^ 2 \tilde \beta (x, \omega) = \tilde \beta '' (x, \omega) -m ^ 2 \tilde \beta (x, \omega) + \delta (x) [\tilde A * \tilde \gamma] (\omega), \qquad (x, \omega) \in \mathbb R ^ 2.
$$
This identity implies the key \textit {\bf spectral inclusion}
\begin {equation} \label {si}
\supp \tilde A * \tilde \gamma \subset \supp \tilde \gamma,
\end {equation}
because $ \supp \tilde \beta (x, \cdot) \subset \supp \tilde \gamma $ and $ \supp \tilde \beta '(x, \cdot) \subset \supp \tilde \gamma $ by the representation ( \ref {gab8}).
From this inclusion, we will derive below \eqref {gab4}, using the fundamental result of Harmonic Analysis - Titchmarsh convolution theorem.

%%%%%%%%%%%%%%%%%%%%%%%%%%%%%%%%%%%%%%%%%%%%%%%%%
%%%%%%%%%%%%%%%%%%%%%%%%%%%%%%%%%%%%%%%%%%%%%%%%%
\subsubsection {Titchmarsh convolution theorem}
%%%%%%%%%%%%%%%%%%%%%%%%%%%%%%%%%%%%%%%%%%%%%%%%%%%%%%
In 1926, Titchmarsh proved a theorem on the distribution of zeros of entire functions  \cite [p.119] {Lev96}, \cite {Tit26}, which implies,
in particular, the following corollary \cite [Theorem 4.3.3] {Hor90}:
\smallskip

\noindent
\textbf {Theorem.} \textit {Let $ f (\omega) $ and $ g (\omega) $ be distributions of  $ \omega \in \mathbb R $ with bounded supports. Then
$$
[\supp \5 f \! * \! g] = [\supp f] + [\supp g],
$$
where $ [X] $ denotes {\bf convex hull}  of a set $ X \subset \mathbb R $.}
\smallskip

Note, that in our situation, $ \supp \tilde \gamma $ is bounded by (\ref {gab10}). Consequently, $ \supp \tilde A $ is also bounded, since
$ A (t): = a (| \gamma (t) |) $ is a polynomial in $ | \gamma (t) | ^ 2 $ according to \eqref {C4}.
Now the spectral inclusion (\ref {si}) and Titchmarsh theorem imply that
$$
[\supp \tilde A] + [\supp \tilde \gamma] \subset [\supp \tilde \gamma],
$$
whence it immediately follows that $ [\supp \tilde A] = \{0 \} $. Besides, $ A (t): = a (| \gamma (t) |) $ is a bounded function due to
(\ref {betab}), because $ \gamma (t) = \beta (0, t) $. Therefore, $ \tilde A (\omega) = C \delta (\omega) $. Hence,
$$
a (| \gamma (t) |) = C_1, \qquad t \in \mathbb R.
$$
Now,  strict nonlinearity condition \eqref {C4} implies that
$$
| \gamma (t) | = C_2, \qquad t \in \mathbb R.
$$
This implies immediately that $ \supp \tilde \gamma = \{\omega_ + \} $ by the same Titchmarsh theorem
for the convolution $ \ti \ga * \ov {\ti \ga} = C_3 \de (\om) $.
Therefore, $ \tilde \gamma (\omega) = C_4 \5 \delta (\omega- \omega _ +) $, and now \eqref {gab4} follows  from  \eqref {gab8}.

%%%%%%%%%%%%%%%%%%%%%%%%%%%%%
\begin {remark} \label {rST}
{\rm
In the case of the Schr\"odinger equation \eqref {S1}, the Titchmarsh theorem does not work.
The fact is that the continuous spectrum of the operator $ -d ^ 2 / dx ^ 2 $ is the half-line $ [0, \infty) $,
so now the role of the ``spectral gap'' plays unbounded interval $ (- \infty, 0) $. Respectively, in this case the spectral inclusion \eqref {I}
gives only that  $ \supp \tilde \beta (x, \cdot) \subset (- \infty, 0) $, while the Titchmarsh theorem applies only to
distribution with bounded supports.
}
\end {remark}

%%%%%%%%%%%%%%%%%%%%%%%%%%%%%%%%%%%%%%%%%%%%%%%%%
%%%%%%%%%%%%%%%%%%%%%%%%%%%%%%%%%%%%%%%%%%%%%%%%%
\subsection {Remarks on dispersion radiation and nonlinear energy transfer} \label {Sec-4.8}
%%%%%%%%%%%%%%%%%%%%%%%%%%%%%%%%%%%%%%%%%%%%%%%%%
Let us explain informal arguments for the attraction to stationary orbits behind formal proof of Theorem \ref {t5}.
Main part of the proof concerns the study of the spectrum of omega-limit trajectories
$$
\beta (x, t) = \lim_ {s_ {j '} \to \infty} \psi (x, s_ {j'} + t).
$$
Theorem \ref {lKato} implies the spectral inclusion (\ref {gab10}), which leads to
\begin {equation} \label {I}
\supp \tilde \beta (x, \cdot) \subset [-m, m], \qquad x \in \mathbb R.
\end {equation}
 Then the Titchmarsh theorem allows us to conclude that
\begin {equation} \label {II}
\supp \tilde \beta (x, \cdot) = \{\omega _ + \}.
\end {equation}
These two inclusions are suggested by the following two informal arguments:
\medskip \\
\textbf {A.} \textit {Dispersion radiation in the continuous spectrum.}
\medskip \\
\textbf {B.} \textit {Nonlinear spreading of the spectrum and the energy transfer  from lower 
to higher harmonics.}
\bigskip \\
\textbf {A. Dispersion radiation.} Inclusion (\ref {I}) is due to the dispersion mechanism, which can be illustrated by
energy radiation in a wave field with harmonic excitation with a frequency lying in the continuous spectrum.
Namely, let us consider  one-dimensional linear Klein--Gordon equation  with a \textit {harmonic source}
\be \la {welap}
\displaystyle \ddot \psi (x, t) = \psi '' (x, t) -m ^ 2 \psi (x, t) + b (x) e ^ {- i \omega_0 t}, \qquad x \in \mathbb R,
\ee
where $ b \in L ^ 2 (\mathbb R) $ and real frequency $ \om_0 \ne \pm m $. Then the \textit {limiting amplitude principle} holds
\cite{Lad1957, Mor1962,KopK2012}:
\begin {equation} \label {lap}
\psi (x, t) \sim a (x) e ^ {- i \omega_0 t}, \qquad t \to \infty.
\end {equation}
For the equation (\ref {welap}), this  follows directly from the Fourier--Laplace transform in time
\begin {equation} \label {Foul}
\ti\psi (\om, t)= \int_0^\infty e ^ {i \omega t}\psi (x, t)dt, \qquad
x \in \mathbb R,\quad \rIm \om> 0.
\end {equation}
Namely, applying this transform to equation (\ref{welap}), we obtain
$$
- \om ^ 2 \ti \psi (x, \om) = \ti\psi '' (x, \om) -m ^ 2 \ti \psi (x, \om) + \fr {b (x)} { i (\om- \om_0)}, \qquad
x \in \mathbb R,\quad \rIm \om> 0, 
$$
where we assume zero initial data for the simplicity of exposition. Hence,
\be \la {welap3}
\ti \psi (\cdot, \om) = \fr {R (\om) b} {i (\om- \om_0)} = \fr {R (\om_0 + i0) b} {i (\om - \om_0)}
+ \fr {R (\om) b-R (\om_0 + i0) b} {i (\om- \om_0)}, \qquad
\rIm \om> 0,
\ee
where  $ R (\om): = (H- \om ^ 2) ^ {- 1} $ is the resolvent of the Schr\"odinger operator 
$ H: = - d ^ 2 / dx ^ 2 + m ^ 2 $. 
This resolvent is an operator of convolution with  fundamental solution
$ -\fr {e ^ {ik (\om) | x |}} {2ik (\om)} $, where $ k (\om) = \sqrt {\om ^ 2-m ^ 2} \in \ov {\Co ^ +} $
for $ \om \in \ov {\Co ^ +} $, as in (\ref {KG4}). 
The last quotient of (\ref {welap3}) is regular at $ \om = \om_0 $, and therefore its contribution
is a dispersion wave which decays in local energy seminorms like (\ref{psid}). Hence, the long-time asymptotics of $\psi(x,t)$ is determined by the middle quotient of (\ref {welap3}). 
Therefore,
(\ref {lap}) holds  with the limiting amplitude $ a (x) = R (\om_0 + i0) b $.
The Fourier transform of this limiting  amplitude is equal to
$$
\hat a (k) = - \frac {\hat b (k)} {k ^ 2 + m ^ 2 - (\om_0 + i0) ^ 2}, \qquad k \in \mathbb R.
$$
This formula shows that the properties of the limiting amplitude differ significantly in the cases
$ | \om_0 | <m $ and $ | \om_0 | \ge m $:
 $a (x) \in H ^ 2 (\mathbb R) $
for $ | \om_0 | <m $, however,
\begin {equation} \label {lapa}
a (x) \not \in L ^ 2 (\mathbb R) \quad \mbox {for} \quad | \om_0 | \ge m,
\end {equation}
if $ | \hat b (k) | \ge \ve> 0 $ in the neighborhood of the ``sphere'' $ | k | ^ 2 + m ^ 2 = \om_0 ^ 2 $
(which consists of two points in 1D case).
{\it This means the following}:
\smallskip

{\bf I. In the case $ | \om_0 | \ge m $}
 the energy of the solution $ \psi (x, t) $ tends to infinity for large times
according  to (\ref {lap}) and (\ref {lapa}).
This means that energy is transmitted from the harmonic source  to the wave field!
\smallskip

{\bf II. Contrary, for $ | \om_0 | < m $}
 the energy of the solution remains bounded, so there is no radiation.
\smallskip\\
Exactly this radiation in the case of $ | \om_0 | \ge m $ prohibits the presence of harmonics with such frequencies in
omega-limit trajectories.
Namely, any omega-limit trajectory cannot radiate at all since 
total energy is finite and bounded from below,
and hence the radiation cannot last forever.
These physical arguments  make the inclusion (\ref {I}) plausible, although its rigorous proof, as was seen above, requires special arguments.

Recall that the set $ i\Sigma: = \{i\om_0 \in \mathbb R $, $ | \om_0 | \ge m \} $ coincides
with the continuous spectrum of the generator of the free Klein--Gordon equation. Radiation in the continuous spectrum  is well known
 in the theory of waveguides. Namely, waveguides can transmit only signals with a frequency 
$|\om_0|>\mu$ 
where $\mu$ is a
{\it threshold frequency}, 
which is an edge point of the continuous spectrum \cite {Lewin}. In our case, the waveguide 
occupies the “entire space” $ x \in \R $ and
is described by the nonlinear Klein--Gordon equation (\ref{KG1i}) with the threshold frequency $m$. 
%%%%%%%%%%%%%%%%%%%%%%%%%%%%%%%%%%%%%%%%%%%%%%%%%%%%%%%%%%%%%%%%
\smallskip \\
\textbf {B. Nonlinear inflation of  spectrum and the energy transfer 
from lower
to higher harmonics.}
%%%%%%%%%%%%%%%%%%%%%%%%%%%%%%%%%%%%%%%%%%%%%%%%%%%%%%%%%%%%%%%%
Let us  show that the single spectrum (\ref {II}) is due to the inflation of  spectrum by nonlinear functions.
For example, let us consider the potential
$ U (\psi) = | \psi | ^ 4 \! $. Respectively, $ F (\psi) = - \nabla_ {\ov \psi} U (\psi) = - 4 | \psi | ^ 2 \psi $.
Consider the sum of two harmonics $ \psi (t) = e ^ {i \omega_1t} + e ^ {i \omega_2t} $, which spectrum  is shown on Fig. \ref {fig-11}:

%%%%%%%%%%%%%%%%%%%%%%%%%%%%%%%%%%%%%%%%
\begin {figure} [htbp]
\begin {center}
\includegraphics [width = 0.9 \columnwidth] {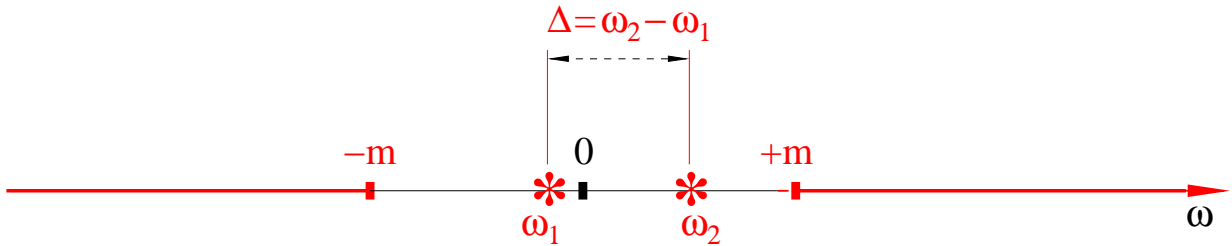}
\caption {Two-point spectrum}
\label {fig-11}
\end {center}
\end {figure}
%%%%%%%%%%%%%%%%%%%%%%%%%%%%%%%%%%%%%%%%
We substitute this sum into the nonlinearity:
$$
F (\psi (t)) \sim \psi (t) \overline {\psi (t)} \psi (t) = e ^ {i \omega_2t} e ^ {- i \omega_1t} e ^ {i \omega_2t} + \dotsc
= e ^ {i (\omega_2 + \Delta) t} + \dotsc, \qquad \Delta: = \omega_2- \omega_1.
$$
The spectrum of this expression contains harmonics with new frequencies $ \omega_1- \Delta $ and $ \omega_2 + \Delta $.
As a result, all frequencies  $ \omega_1- \Delta $, $ \omega_1-2 \Delta, \dots $ and $ \omega_2 + \Delta $, $ \omega_2 +2 \Delta $, $ \dots $
also will appear in the nonlinear dynamics (\ref{KG1i})
(see Fig.  \ref {fig-2}).
Therefore, these frequencies will appear also in the nonlinear term with $\de$-function.
\smallskip

%%%%%%%%%%%%%%%%%%%%%%%%%%%%%%%%%%%%%%%%%%%%%%%%%%%%%% %%%%%
\begin {figure} [htbp]
\begin {center}
\includegraphics [width = 1.0 \columnwidth] {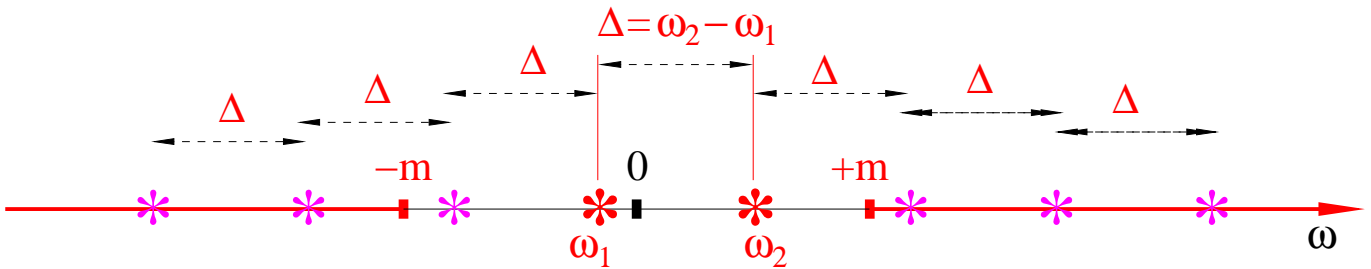}
\caption {Nonlinear inflation of spectrum}
\label {fig-2}
\end {center}
\end {figure}
~ \\
%%%%%%%%%%%%%%%%%%%%%%%%%%%%%%%%%%%%%%%%%%%%%%%%%%%%%%%%%%%%
As we already know, these  
 frequencies lying in the continuous spectrum $ | \omega | > m $ will surely 
cause  energy radiation. This radiation will continue until the spectrum of the solution contains at least two different frequencies. 
Exactly this fact prohibits the presence of two different frequencies in omega-limit trajectories because total energy is finite, 
so the radiation cannot continue forever.

Let us emphasize that an exact meaning of the inflation of spectrum by nonlinearity  is established by the Titchmarsh convolution theorem.

%%%%%%%%%%%%%%%%%%%%%%%%%%%%%%%%%
\begin {remark} \label {phys}
{\rm
The above arguments physically mean the following two-step {\it nonlinear radiation mechanism}:
\smallskip\\
i) Nonlinearity inflates the spectrum, which means energy transfer from lower to higher harmonics;
\smallskip \\
ii) The dispersion radiation  transfers energy to infinity.
\smallskip \\
We have rigorously justified 
such  nonlinear radiation mechanism   for the first time for nonlinear $ U (1) $ - invariant
Klein--Gordon and Dirac  equations \eqref {KG1} and \eqref {KGN}--\eqref {Dn}.
Our numerical experiments demonstrate similar radiation mechanism for nonlinear {\it relativistic} wave equations, 
see Remark \ref {rrm}. However, a  rigorous proof is still missing.
}
\end {remark}
%%%%%%%%%%%%%%%%%%%%%%%%%%%%%%
\begin {remark} \label {rg}
{\rm Let us comment on the term {\it generic equation} in our conjecture (\ref {at10}).
\medskip \\
i) Asymptotics (\ref {WP9}), (\ref {WP10}) hold under the 
Wiener condition (\ref {W1}), which defines some ``open dense set'' of functions $ \rho $.
This asymptotics may break down if the Wiener condition fails. For example, if $ \rho (x) \equiv 0 $, then the particle dynamics is independent from the fields, and hence, the attraction
to stationary states can fail. 
\medskip \\
ii)
Similarly, asymptotics (\ref {ga}) is valid for an open  set of $ U (1) $-invariant equations 
corresponding to polynomials (\ref {C4}) with $ N \ge 2 $. 
However, this asymptotics may break down  for ``exceptional'' $ U (1) $ - invariant equations. 
In particular, for linear equations, corresponding to polynomials (\ref {C4}) with $ N = 1 $.
The corresponding  examples are constructed in \cite {KK2007}.
\medskip \\
iii) General situation is the following. Let a Lie group $ g $ be a (proper) subgroup of some  larger Lie group $ G $. 
Then $ G $-invariant equations form an ``exceptional subset'' among all $ g $-invariant equations, and
the corresponding asymptotics (\ref {at10}) may be completely different.
For example, the trivial group $ \{e \} $ is a subgroup in $ U (1) $ and in $ \R ^ n $, and asymptotics
(\ref {att}) and (\ref {atU}) may differ significantly  from (\ref {ate}).
}
\end {remark}

%%%%%%%%%%%%%%%%%%%%%%%%%%%%%%%%%%%%%%%%%%%%%%%%%%%%%%%%%%
%%%%%%%%%%%%%%%%%%%%%%%%%%%%%%%%%%%%%%%%%%%%%%%%%%%%%%%%%%
\setcounter {equation} {0}
\section {Asymptotic stability of stationary orbits and solitons}
%%%%%%%%%%%%%%%%%%%%%%%%%%%%%%%%%%%%%%%%%%%%%%%%%%%%%%%%%%
Asymptotic stability of solitary manifolds means a local attraction, i.e. for states sufficiently close to the manifold.
The main feature of this attraction is the instability of the dynamics of \textit {along the manifold}.
This follows directly from the fact that solitons move with different speeds and therefore run away for large times.

Analytically, this instability is caused by  the presence of the eigenvalue $ \lambda = 0 $ in  spectrum of  the generator of linearized dynamics. 
Namely, the tangent vectors to soliton manifolds are eigenvectors  and associated vectors of  the generator.  
They correspond to zero eigenvalue. Respectively, Lyapunov's theory is not applicable to this case.

In a series of articles  of Weinstein, Soffer and Buslaev, Perelman and Sulem 1985--2003 
an original  strategy was developed
for proving asymptotic stability of solitary manifolds.
This strategy relies on i)  special projection of a trajectory onto the solitary manifold,
ii) modulation equations for  parameters of the projection, and iii) time-decay of  transversal component.
This approach is a far-reaching development of the Lyapunov stability theory.

%%%%%%%%%%%%%%%%%%%%%%%%%%%%%%%%%%%%%%%%%%%%%%%%%%%%%%%%%%
%%%%%%%%%%%%%%%%%%%%%%%%%%%%%%%%%%%%%%%%%%%%%%%%%%%%%%%%%%
\subsection {Asymptotic stability of stationary orbits. Orthogonal projection} \la {s6}
%%%%%%%%%%%%%%%%%%%%%%%%%%%%%%%%%%%%%%%%%%%%%%%%%%%%%%%%%%

This strategy arose in 1985--1992 in the pioneering work of Soffer and Weinstein
\cite {SW1990, SW1992,W1985}, see the review \cite {soffer2006}.
The results concern nonlinear $ U (1) $ - invariant Schr\"odinger equations with real potential $ V (x) $
\begin {equation} \label {Su}
i \dot \psi (x, t) = - \Delta \psi (x, t) + V (x) \psi (x, t) + \lambda | \psi (x, t) | ^ p \psi ( x, t), \qquad x \in \mathbb R ^ n,
\end {equation}
where $ \lambda \in \mathbb R $, $ p = 3 $ or $ 4 $, $ n = 2 $ or $ n = 3 $,
and $ \psi (x, t) \in \Co $.
The corresponding Hamilton functional reads
$$
\cH = \int \big[\frac12 | \nabla \psi | ^ 2 + \frac12V (x) | \psi (x) | ^ 2 + \frac \lambda p | \psi (x) | ^ p\big] \, dx.
$$
For $ \lambda = 0 $, the equation (\ref {Su}) is linear.
It is assumed that the discrete spectrum
of the short range Schr\"odinger operator $H:=-\De+V(x)$
 is a single point $ \omega_ * <0 $, and the point zero is neither an eigenvalue nor a resonance
 for $H$.
Let $ \phi _ * (x) $ denote the corresponding ground state:
\be\la{grst}
H\phi _ * (x)=\omega_ * \phi _ * (x).
\ee
Then $ C \phi _ * (x) e ^ {- i \omega_ * t} $
are periodic solutions for all complex constants $ C $. Corresponding phase curves are circles, filling the 
complex plane.

For nonlinear equations (\ref {Su}) with a small real $ \lambda \ne 0 $, it turns out that a wonderful {\it bifurcation} occurs:
small neighborhood of the zero of the complex plane
turns into an analytic invariant soliton manifold
$ \cS $
which is still filled with invariant circles which are trajectories 
of {\it stationary orbits}  of type (\ref{storb}),
\be\la{storb2}
 \psi(x,t)=\psi_ \omega (x) e ^ {- i \omega t} 
 \ee
whose frequencies $ \omega $ are close to $ \omega _ * $.
\br\la{rgs}
Now all these solutions $ \psi_\omega (x) e ^{- i \omega t} $ are called as {\it ground states}.
\er
The  main result of \cite {SW1990, SW1992} (see also \cite {PW1997}) is  long-time attraction to one of these ground states
 for any solution with sufficiently small initial data:
\begin {equation} \label {soli}
\psi (x, t) = \psi _ {\pm} (x) e ^ {- i \omega _ {\pm} t} + r_ \pm (x, t),
\end {equation}
where the remainder decay in weighted norms: for $ \sigma> 2 $
$$
\Vert \langle x \rangle ^ {- \sigma} r_ \pm (\cdot, t) \Vert_ {L ^ 2 (\mathbb R ^ n)} \to 0, \qquad t \to \pm \infty,
$$
where $ \langle x \rangle: = (1+ | x |) ^ {1/2} $.
The proof relies on  linearization of the dynamics and decomposition of solutions into two components
$$
\psi (t) = e ^ {- i \Theta (t)} (\psi _ {\omega (t)} + \phi (t)),
$$
with the  orthogonality condition \cite [(3.2) and (3.4)] {SW1990}:
\begin {equation} \label {or}
\langle \psi _ {\omega (t)}, \phi (t) \rangle = 0.
\end {equation}
This orthogonality and dynamics (\ref {Su}) imply the {\it modulation equations} for $ \omega (t) $ and $ \gamma (t) $, where
$ \gamma (t): = \Theta (t) - \displaystyle \int_0 ^ t \omega (s) ds $ (see (3.2) and (3.9a)--(3.9b) from \cite {SW1990}).
The orthogonality (\ref {or}) implies that the component $ \phi (t) $ lies in the continuous spectral space of the Schr\"odinger operator
$H (\omega_0): = - \Delta + V + \lambda |\psi _ {\omega_0}|^p$, which leads to time-decay of  $\phi (t)$ (see \cite [(4.2a) and (4.2b)] {SW1990}).
Finally, this decay implies the convergence $ \omega (t) \to \omega_ \pm $ and the asymptotics (\ref {soli}).
\smallskip

These results and methods were further developed in the numerous works  for nonlinear  Schr\"odinger, wave and Klein--Gordon  equations with potentials under various spectral assumptions on linearized dynamics, 
\cite {BKKS2008,B2006,KKopSt2011,SW1990,SW1992,PW1997,SW1999,SW2004,
soffer2006,W1985}.
\medskip

%%%%%%%%%%%%%%%%%%%%%%%%%%%%%%%%%%%%%%%%%%%%%%%%%%%%%%%%%%
%%%%%%%%%%%%%%%%%%%%%%%%%%%%%%%%%%%%%%%%%%%%%%%%%%%%%%%%%%
\subsection {Asymptotic stability of solitons. Symplectic projection} \la {s7}
%%%%%%%%%%%%%%%%%%%%%%%%%%%%%%%%%%%%%%%%%%%%%%%%%%%%%%%%%%
Genuine breakthrough  in the theory of asymptotic stability was achieved in 1990-2003 
by Buslaev, Perelman and Sulem \cite {BP1993,BP1995,BS2003},
who first generalised  asymptotics of the type (\ref {soli}) for translation-invariant 1D Schr\"odinger equations 
\begin {equation} \label {BPS}
i \dot \psi (x, t) = - \psi '' (x, t) -F (\psi (x, t)), \qquad x \in \mathbb R,
\end {equation}
which are also assumed to be $ U (1) $-invariant. The latter means that the nonlinear function $ F (\psi) = - \na _ {\ov \psi} U (\psi) $ satisfies
identities (\ref {C3})--(\ref {U1}). Also the following condition is assumed
\be\la{F10}
U(\psi)=\cO(|\psi|^{10}),\qquad \psi\to 0,
\ee
which is caused probably by a failure of suitable technique.
Under some simple additional conditions on the potential $U$ (see below), 
there exist 
{\it stationary orbits} which are
finite energy solutions of  the form
\be \la {sol0}
 \psi (x, t) = \psi_0 (x) e ^ { i \omega_0 t},
 \ee
 with $ \om_0>0 $. The amplitude $ \psi_0 (x) $ satisfies the corresponding stationary equation
 \be\la{cse}
 -\om_0 \psi_0 (x) = - \psi_0 '' (x) -F (\psi_0 (x)), \qquad x \in \R,
 \ee
which implies the ``conservation law''
 \be \la {solC}
 \fr {| \psi_0 '(x) | ^ 2} 2 + U_e (\psi_0 (x)) = E,
 \ee
where the ``effective potential''
$U_e (\psi) = U (\psi) + \om_0 \fr {| \psi | ^ 2}2 \sim \om_0 \fr {| \psi | ^ 2}2$ as  $\psi\to 0$
by (\ref{F10}).
 For the existence of finite energy solution (\ref{sol0}),  the graph of the effective potential $ U_e (\psi) $ 
should be similar to Fig. \ref {Ue}. 
The finite energy solution
 is defined by (\ref {solC})  with the constant $ E = U_e (0) $
 since for other $E$ the solutions to
 (\ref{solC})
 do not converge to zero as $|x|\to\infty$. 
 This equation with  $ E = U_e (0) $ implies that
 \be \la {solC2}
 \fr {| \psi_0 '(x) | ^ 2} 2 = U_e (0) -U_e (\psi_0 (x))
\sim \fr{\om_0}  2\psi_0^2 (x).
 \ee
 Hence, for finite energy solutions
 \be \la {rU}
 \psi_0 (x) \sim e ^ {- \sqrt{\om_0} | x |},\qquad |x|\to\infty.
 \ee
%%%%%%%%%%%%%%%%%%%%%%%%

\begin {figure} [htbp]
\begin {center}
\includegraphics [width = 0.8 \columnwidth] {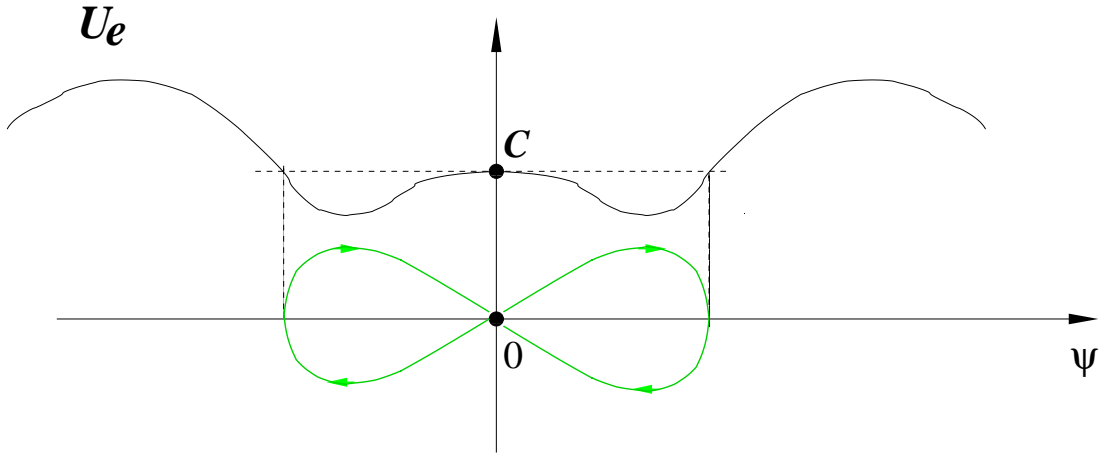}
\caption {Reduced potential and soliton.}
\label {Ue}
\end {center}
\end {figure}
%%%%%%%%%%%%%%%%%%%%%%%%
It is easy to verify that  the  following functions are also solutions  ({\it moving solitons})
\be \la {solv}
 \psi _ {\om, v, a, \theta} (x, t) = \psi_ \om (x-vt-a) e ^ { i (\omega t + kx + \theta)}, \qquad
 \om = \om_0 - v ^ 2/4, \qquad k = v / 2.
 \ee
The set of all such solitons with parameters $ \om, v, a, \theta $ forms a 4-dimensional smooth submanifold $ \cS $ in the
Hilbert phase space $ \cX: = L ^ 2 (\R) $. Moving solitons (\ref {solv}) are obtained from 
standing  (\ref {sol0}) by the Galilean transformation
\be \la {Gavt}
G (a, v, \theta): \psi (x, t) \mapsto \vp (x, t) =\psi (x-vt-a, t) e ^ { i (-\fr {v ^ 2} 4t +\fr v2 x + \theta)}.
\ee
It is easy to verify that the Schr\"odinger equation (\ref {BPS}) is invariant with respect to this group of transformations.
\smallskip

Linearization of the Schr\"odinger equation (\ref {BPS}) at the stationary orbit
 (\ref {sol0}) is obtained by substitution
$ \psi (x, t) = (\psi_0 (x) + \chi (x)) e ^ {- i \om_0t} $ and retaining terms of the first order in $ \chi $. 
This linearised equation contains $ \chi $ and $ \ov \chi $, and hence,
it
 is not linear over the field of complex numbers. 
This follows from the fact that the nonlinearity of $ F (\psi) $ is not complex-analytic due to the $ U (1) $-invariance 
(\ref {C3}).  Complexification of this linearized equation reads
\be \la {cLS}
 \dot \Psi (x, t) = C_0 \Psi (x, t),\qquad C_0=-jH_0,
\ee
where $ j $ is a real $ 2 \times 2 $ matrix, representing the multiplier $ i $, $ \Psi (x, t) \in \Co ^ 2 $,
and  $ H_0 = - {d ^ 2} / {dx ^ 2} +\om_0 + V (x) $, where $ V (x) $ is a real matrix potential, 
which decreases exponentially as $ | x | \to \infty $ due to (\ref {rU}). 
Note that the operator $C_0=C_{\om_0, 0,0,0} $ corresponds to the linearization   at the soliton
(\ref{solv}) with parameters $\om=\om_0$, and $a=v=\theta=0$.
Similar operators $C_ {\om, a, v, \theta} $,  corresponding to linearization at solitons (\ref{solv}) with various parameters $ \om, a, v, \theta $, 
are also connected by the linear Galilean transformation (\ref {Gavt}). 
Therefore, their spectral properties completely coincide. In particular, their continuous spectrum coincides with $(-i\infty,-i \om_0]\cup [i\om_0, i\infty) $.

Main results of \cite {BP1993,BP1995,BS2003} are asymptotics of type (\ref {soli})
for solutions with initial data close to the solitary manifold $ \cS $:
\begin {equation} \label {sdw}
\psi (x, t) = \psi_ \pm (x-v_ \pm t) e ^ {- i (\omega_ \pm t + k_ \pm x)} + W (t) \Phi_ \pm + r_ \pm (x, t),\qquad \pm t>0,
\end {equation}
where $ W (t) $ is the dynamical group of the free Schr\"odinger equation, $ \Phi _ {\pm} $ are some scattering states of finite energy, and
$ r _ {\pm} $ are  remainder  terms which decay to zero in a global norm:
\begin {equation} \label {glob2u}
\Vert r _ {\pm} (\cdot, t) \Vert_ {L ^ 2 (\mathbb R)} \to 0, \qquad t \to \pm \infty.
\end {equation}
These asymptotics were obtained under following assumptions on the spectrum of the generator $B_0$:
\smallskip

U1. The discrete spectrum of the operator $ C_0 $ consists of exactly three eigenvalues $ 0 $ and $ \pm i\lam  $, and
\be \la {fot}
 \lam< \om_0<2 \lam.
\ee
This condition means that the discrete mode can interact with 
the continuous spectrum already  in the first order of perturbation theory.

U2. The edge points $\pm i\om_0$ of the continuous spectrum  are neither  eigenvalues, nor resonances of $ C_0 $.

U3. Furthermore, it is assumed the condition \cite [(1.0.12)] {BS2003}, which means a strong coupling of discrete and continuous spectral components, 
providing energy radiation, similarly to the Wiener condition  (\ref{W1}).
The condition \cite [(1.0.12)] {BS2003} ensures that the interaction of  discrete component with  continuous spectrum does not vanish  in the first order of perturbation theory.
This condition is a nonlinear version of the Fermi Golden Rule \cite {RS4}, which was introduced by Sigal 
in the context of nonlinear PDEs \cite {Sig1993}.
\smallskip

Examples of potentials satisfying all these conditions are constructed in \ci {KKK2013}.  
\smallskip

In 2001, Cuccagna extended results of  \cite {BP1993,BP1995,BS2003} to nD translation-invariant Schr\"odinger equations 
in the dimensions $ n \ge 2 $, \cite {C2001}.
%%%%%%%%%%%%%%%%%%%%%%%%%%%%%%%%%%%%%%%%%%%%%%%%
\medskip \\
{\bf Method of symplectic projection in the Hilbert  phase space.}
%%%%%%%%%%%%%%%%%%%%%%%%%%%%%%%%%%%%%%%%%%%%%%%%
Novel approach \cite {BP1993,BP1995,BS2003} relies on {\it symplectic projection} 
 of solutions onto  the solitary manifold. This means that 
$$
Z: = \psi-S \quad \mbox {symplectic-orthogonal to the tangent space} \quad \cT: = T_ {S} \cS
$$
for the projection $ S: = P \psi $. This projection is correctly defined in a small neighborhood of $ \cS $ because $ \cS $ is a {\it symplectic manifold}, 
i.e. the corresponding symplectic form is non-degenerate on the tangent spaces $T_ {S}\cS$.
In particular, 
the  approach \cite {BP1993,BP1995,BS2003} 
does not require a smallness of 
 initial data.

Thus a solution $ \psi(t)$ for each $ t> 0 $ decomposes as  $\psi(t)= S(t)+ Z(t)$, 
where $ S (t): = P \psi (t) $, and the dynamics is linearized on the soliton $ S (t)$.  
Similarly, for each $t\in\R$
the total 
Hilbert phase space $ \cX: = L ^ 2 (\R) $ is splitted as $ \cX = \cT (t) \oplus \cZ (t) $, where
$ {\mathcal Z} (t) $ is {\bf symplectic-orthogonal} complement to the tangent space $ \cT (t): = T_ {S (t)} \cS $.
The corresponding equation for the {\it transversal component} $ Z (t) $ reads
$$
\dot Z (t) = A (t) Z (t) + N (t),
$$
where $ A (t) Z (t) $ is the linear part, and $ N (t) = \cO (\Vert Z (t) \Vert ^ 2) $ is the corresponding nonlinear part.

The main difficulties in studying this equation are as follows i) it is {\it non-autonomous},
and ii) the generators $ A (t) $ {\it  are not self-adjoint} (see Appendix in \cite {KKsp2014}).
It is important that $ A (t) $ are {\it Hamiltonian operators}, for which the existence of spectral  decomposition 
is provided by the Krein-Langer theory of $ J $ -selfadjoint operators \ci {KL1963, L1981}.
In \ci{KKsp2014, KKsp2015} we have developed a special version of this theory  providing
the corresponding eigenfunction expansion which is necessary for the justification 
 of the  approach \cite {BP1993,BP1995,BS2003}.
The main steps of this strategy are as follows.
\smallskip \\
$\bullet$ {\bf Modulation equations.}
The parameters of the soliton $ S (t) $ satisfy  {\bf modulation equations}: 
for example, for the  speed  $v (t)$ we have
$$ 
\dot v (t) = M (\psi (t)), 
$$ 
where $ M (\psi) = \cO (\Vert Z \Vert ^ 2) $ for small norms $ \Vert Z \Vert $. 
This means that the parameters change ``superslow'' near the soliton manifold, like adiabatic invariants.
\smallskip \\
$\bullet$ {\bf Tangent and transversal components.}
The {\it transversal component} $ Z (t) $ in the splitting $ \psi (t) = S (t) + Z (t) $ belongs to the 
{\it transversal subspace} $ \cZ (t) $.
The tangent space $ \cT (t) $ is the root space of the generator $ A (t) $ 
and corresponds to the “unstable” spectral point $ \lambda = 0 $. The key observation is that

i) the transversal subspace $ {\mathcal Z} (t) $ is {\bf invariant} with respect to the generator $ A (t) $, since
the subspace $ \cT (t) $ is invariant, and $ A (t) $  is the  Hamiltonian operator;

ii) moreover, the transversal subspace
$ {\mathcal Z} (t) $
does not contain ``unstable'' tangent vectors.
\smallskip \\
$\bullet$ {\bf Continuous and discrete components.}
The transversal component allows further splitting $ Z (t) = z (t) + f (t) $, where $ z (t) $ and $ f (t) $ 
belong, respectively, to discrete and continuous spectral subspaces $ \cZ_d (t) $ and $ \cZ_c (t) $ of $A(t)$
in the space $ \cZ (t) = \cZ_d (t) + \cZ_c (t) $.
 \smallskip \\
$\bullet$ {\bf  Poincare normal forms and  Fermi Golden Rule.}
The component $z (t) $ satisfies a nonlinear equation, which is reduced to Poincare normal form 
up to higher order terms \cite [Equations (4.3.20)]{BS2003}.
The normal form allowed  to obtain some ``conditional decay'' for $z(t) $ using the Fermi Golden Rule
\cite [(1.0.12)] {BS2003}.
For the relativistic-invariant Ginzburg-Landau equation, a similar 
reduction done in \cite [Equations (5.18)]{KopK2011b}. 
\smallskip \\
$\bullet$ {\bf  Method of majorants.}
A skillful combination of the conditional decay for $ z (t) $ with the superslow evolution of the soliton parameters
allows us to prove the decay for $ f (t) $ and $ z (t) $ by the method of  majorants.
Finally, this decay implies the asymptotics (\ref {sdw})--(\ref {glob2u}).

%%%%%%%%%%%%%%%%%%%%%%%%%%%%%%%%%%%%%%%%%%%%%%%%%%%%
%%%%%%%%%%%%%%%%%%%%%%%%%%%%%%%%%%%%%%%%%%%%%%%%%%%%
\subsection {Generalizations and Applications} \la {s8}
%%%%%%%%%%%%%%%%%%%%%%%%%%%%%%%%%%%%%%%%%%%%%%%%%%%%
{\bf $N$-soliton solutions.} The methods and results of \ci{BS2003} were developed in \cite {MMT2002,MSig1999,MW1996,PW1994,P2004,RSS2003,RSS2005} 
for $ N $-soliton solutions for translation-invariant nonlinear Schr\"odinger equations.
\smallskip\\
{\bf  Multiphoton  radiation.}
In  \cite {CM2008} Cuccagna and Mizumachi  extended methods  and results  of \ci{BS2003}  to the case 
when the inequality (\ref {fot}) is changed to
$$
N  \lam< \om_0< (N+1) \lam,
$$
with some natural $N>1$,  and the corresponding analogue of condition U3 holds.
It means, that the interac\-ti\-on of discrete modes with a continuous spectrum occurs only in the $ N $-th order of 
perturbation theory. The decay rate of the remainder term (\ref {glob2u}) worsens
 with growing  $ N $.
\smallskip\\
{\bf Linear equations coupled to nonlinear oscillators and particles.}
The methods and results of \ci{BS2003} were extended
i) in \cite {BKKS2008, KKopSt2011} to the Schr\"odinger equation coupled to a nonlinear $ U (1) $-invariant oscillator, 
ii) in \cite {IKS2011, IKV2011}  to systems  \eqref {wq3} and \eqref {ML} with zero external fields, 
and iii) in \cite {IKV2006, KKop2006, KKopS2011}  to  similar translation-invariant systems
of the Klein--Gordon, Schr\"odinger and Dirac equations coupled to a particle.
The survey of these results  can be found in \cite {Im2013}.

For example, article  \cite {IKV2011} concerns solutions to the system \eqref {wq3} with initial data 
close to a soliton manifold (\ref {solit}) in weighted norm
$$
\Vert \psi \Vert_ \sigma ^ 2 = \int \langle x \rangle ^ {2 \sigma} | \psi (x) | ^ 2dx
$$
with sufficiently large $ \si> 0 $. Namely, the initial state is close to soliton (\ref {solit}) with some parameters $ v_0, a_0 $:
$$
\begin {gathered}
\Vert \nabla \psi (x, 0) - \nabla \psi_ {v_0} (x-a_0) \Vert_ \sigma
+
\Vert \psi (x, 0) - \psi_ {v_0} (x-a_0) \Vert_ \sigma
+ \Vert \pi (x, 0) - \pi_ {v_0} (x-a_0) \Vert_ \sigma
\\
+
| q (0) -a_0 | +
| \dot q (0) -v_0 | \le \varepsilon,
\end {gathered}
$$
where $ \sigma> 5 $, and $ \varepsilon> 0 $ is sufficiently small.
Moreover, the Wiener condition \eqref {W1} is assumed  for $ k \ne 0 $.
Additionally, let 
$$
\partial ^ \alpha \hat \rho (0) = 0, \quad | \alpha | \le 5,
$$
that is equivalent to equalities
$$
\int x ^ \alpha \rho (x) \, dx = 0, \quad | \alpha | \le 5.
$$
Under these conditions, the main results of \cite {IKV2011} are the asymptotics
$$
\ddot q(t)\to 0,\quad
\dot q (t) \to v_ \pm, \quad q (t) \sim v_ \pm t + a_ \pm, \qquad t \to \pm \infty
$$
(cf. (\ref {dq}) and (\ref {tsol})) and the 
attraction to solitons (\ref {ssol}),
where the  remainder now decays in {\it global weighted norms} in the comoving frame 
(cf. (\ref {ssolh})):
$$
\Vert\nabla r_\pm (q (t) + x, t)\Vert _{-\sigma} + \Vert r_\pm (q (t) + x, t)\Vert _ {-\sigma} + \Vert s_\pm (q (t) + x, t) \Vert _ {-\sigma} \to 0, \qquad t \to \pm \infty.
$$
{\bf Relativistic equations.}
In \cite {K2002, BC2012, Kumn2013, KopK2011a, KopK2011b} methods and results \cite {BS2003} were extended for the first time 
to {\it relativistic-invariant} nonlinear equations.  Namely, in \cite {K2002} and \cite {Kumn2013, KopK2011a,KopK2011b} asymptotics of the type (\ref {sdw}) 
were obtained for 1D relativistic-invariant nonlinear wave equations (\ref {NWEn}) with potentials of the Ginzburg--Landau type,
and in \cite {BC2012} for relativistic-invariant nonlinear Dirac equations.
In  \cite {KKK2013} we have constructed examples of potentials providing all spectral properties of the linearised dynamics
imposed in \cite {Kumn2013, KopK2011a,KopK2011b}.

In \cite {KKsp2014,KKsp2015} we have justified the eigenfunction expansions  for nonselfadjoint Hamiltonian operators 
which were used in \cite {Kumn2013,KopK2011a,KopK2011b}. For the justification we have developed a
special version of the Krein--Langer  theory of $ J $-selfadjoint operators \ci{KL1963, L1981}.
\medskip \\
{\bf Vavilov-Cherenkov radiation.}
The article \cite {FG2014} concerns a system of
type \eqref {wq3} with the Schr\"odinger equation instead of the wave equation  (system (1.9)--(1.10) in \cite {FG2014}). 
This system  is considered as a model of the Cherenkov  radiation. 
The main result of \cite {FG2014} is long-time convergence to a soliton with the sonic speed for initial solitons with a supersonic speed 
in the case of a weak interaction (“Bogolubov limit”) and small initial field.
Asymptotic stability of solitons for  similar system was established in \cite {KKop2006}.

%%%%%%%%%%%%%%%%%%%%%%%%%%%%%%%%%%%%%%%%%%%%%%%%%%%%%%%
%%%%%%%%%%%%%%%%%%%%%%%%%%%%%%%%%%%%%%%%%%%%%%%%%%%%%%%
\subsection {Further generalizations}
%%%%%%%%%%%%%%%%%%%%%%%%%%%%%%%%%%%%%%%%%%%%%%%%%%%%%%%
The results on asymptotic stability of solitons were developed in different directions.
\smallskip \\
{\bf Systems with several bound states.}
Papers \cite {BC2011,C2011,T2003,T2002-1,TY2002} concern asymptotic stability of stationary orbits (\ref{storb2}) 
for the nonlinear Schr\"odinger, Klein--Gordon  and wave equations
 in the
case of several simple eigenvalues of the linearization.
The typical assumptions are as follows:
\smallskip

i) the endpoint of  continuous spectrum is neither an eigenvalue nor a resonance for linearized equation;

ii) the eigenvalues of the linearised equation satisfy several non-resonance conditions;

iii) a new version of the Fermi Golden Rule.
\smallskip

One typical difficulty is  possible long stay of solutions near metastable
tori which correspond to approximate resonances.
Great efforts are being made to show that the role of metastable tori decreases as
$ t ^ {- 1/2} $ as $ t \to \infty $.
The typical result is the long-time asymptotics “ground state + dispersion wave”
in the norm $ H ^ 1 (\R ^ 3) $ for solutions close to the ground state.
\smallskip \\
{\bf General theory of relativity.}
The article \cite {HM2004} concerns so-called ``kink instability'' of  self-similar and spherically symmetric
solutions of the equations of the general theory of relativity with a scalar field, as well as with a ``hard fluid'' as sources. 
The authors constructed examples of self-similar solutions that are unstable to the kink perturbations.

The article \cite {DR2011} examines linear stability of slowly rotating Kerr solutions for the Einstein equations in vacuum.
In \cite {T2013} a pointwise damping of solutions to the wave equation is investigated for the case 
of stationary asymptotically flat space-time in the three-dimensional case.

In \cite {AB2015} the Maxwell equations are considered outside  slowly rotating Kerr black hole.
The main results are: i) ~ boundedness of a positive definite energy on each hypersurface $ t = \const $ 
and ii) convergence of each solution to a stationary Coulomb field.

In \cite {DSS2012} the pointwise decay was proved for linear waves against the Schwarzschild black hole.
\smallskip \\
{\bf Method of concentration compactness.}
In \cite {KM2006} the concentration compactness method was used for the first time to prove global well-posedness, 
scattering and blow-up of solutions to critical focusing nonlinear Schr\"odinger equation
$$
i \dot \psi (x, t) = - \Delta \psi (x, t) - | \psi (x, t) | ^ {\frac4 {n-2}} \psi (x, t), \qquad x \in \R ^ n
$$
in the radial case. Later on, these methods were extended in \cite {KM2012, DKM2016, KM2008, KNS2015}
to general non-radial solutions and to nonlinear wave equations of the type
$$
\ddot \psi (x, t) = \Delta \psi (x, t) + | \psi (x, t) | ^ {\frac4 {n-2}} \psi (x, t), \qquad x \in \R ^ n.
$$
One of the main results is splitting of the set of initial states, close to the critical energy level,
into three subsets with certain long-term asymptotics:
either a blow-up in a finite time, or an asymptotically free wave,
or the sum of the ground state and an asymptotically free wave.
All three alternatives are possible; all nine combinations with $ t \to \pm \infty $ are also possible. 
Lectures \cite {NS2011} give excellent introduction to this area.
The articles \cite {DKM2014, KM2011} concern super-critical nonlinear wave equations.

Recently, these methods and results were extended to critical wave mappings  \cite {KLLS2015,KNS2015,KS2012}. 
The ``decay onto solitons'' is proved: every $1$-equivariant finite-energy wave mapping
of exterior of a ball with Dirichlet boundary conditions into  three-dimensional sphere
exists globally in time and dissipates into a single stationary solution of its own topological class.
\smallskip \\
{\bf Weak convergence to equilibrium distributions in nonlinear Hamilton systems.}
The papers \ci{K2001}--\ci{KS2006} concern the weak convergence 
to an equilibrium distribution in the
Liouville, Vlasov and Schr\"odinger  equations. In \ci{KS2006} the authors introduced the 
quantum Poincar\'e model.

%%%%%%%%%%%%%%%%%%%%%%%%%%%%%%%%%%%%%%%%%%%%%%%%%%%%%
%%%%%%%%%%%%%%%%%%%%%%%%%%%%%%%%%%%%%%%%%%%%%%%%%%%%%
\subsection {Linear dispersion}
%%%%%%%%%%%%%%%%%%%%%%%%%%%%%%%%%%%%%%%%%%%%%%%%%%%%% 
The key role in all results on long-time asymptotic  for nonlinear Hamiltonian PDEs is played by
dispersion decay of solutions of the corresponding linearized equations.
A huge number of publications 
concern this decay, so we choose only most important or recent.
\medskip \\
{\bf Dispersion decay in weighted Sobolev norms.} 
Dispersion decay for wave equations was first proved in  linear scattering theory \cite {LMP1963}.

A powerful systematic  approach to dispersion decay 
for the Schr\"odinger equation with potential  was proposed by 
Agmon, Jensen and Kato \cite {Agmon, JK}. This theory was extended by many authors
to  wave,  Klein--Gordon and Dirac equations and to the corresponding discrete equations, see
\cite {
BG2012,BG2014,A2015,AFVV2010,EGG2014,GG2015,GG2017,EKMT,EKT}
and 
\cite{JSS1991,KopK2010,KopK2010-1,KopK2012,KM2019,KopK2013,
KKm2015,
KKK2006,KKV2008,Kumn2010,K2010,K2018d,KT2016}
   and references therein. 
\medskip \\
{\bf $ L ^ 1 - L ^ \infty $ decay} 
\begin {equation} \label {l1i}
\Vert P_c \psi (t) \Vert_ {L ^ \infty (\R ^ n)} \le Ct ^ {- n / 2} \Vert \psi (0) \Vert_ {L ^ 1 (\R ^ n )}, \qquad t> 0
\end {equation}
for solutions of  linear Schr\"odinger equation
\begin {equation} \label {LSE}
i \dot \psi (x, t) = H \psi (x, t): = (- \Delta + V(x)) \psi (x, t), \qquad x \in \R ^ n
\end {equation}
with $ n \ge 3 $
was proved for the first time by Journet, Soffer and Sogge \cite {JSS1991}
 provided that $ \lambda = 0 $ is neither an eigenvalue nor resonance for $ H $.
The potential $V (x) $ is sufficiently smooth and rapidly decays as $ | x | \to \infty $.
Here $ P_c $ is an orthogonal projection onto continuous  spectral space of the operator $ H $.
This result was  generalised later by many authors, see below.
\smallskip

In \cite {RS2004} a decay of type  (\ref {l1i})  and Strichartz estimates were established
for 3D Schr\"odinger equations  (\ref {LSE}) with ``rough'' and time-dependent potentials $ V = V (x, t) $ (in
stationary case $ V (x) $ belongs to both the Rollnik class and the Kato class).
Similar estimates were received in \cite {BG2012} for 3D Schr\"odinger and  wave equations with (stationary) Kato class potentials.
\smallskip

In \cite {EGG2014} the 4D Schr\"odinger equations (\ref {LSE}) are considered
for the case when there is a resonance or an eigenvalue at  zero energy.
In particular, in the case of  an eigenvalue at  zero energy,  there is a time-dependent  operator $ F_t $ of rank $ 1 $,
such that $ \Vert F_t \Vert_ {L ^ 1 \to L ^ \infty} \le 1 / \log \5 t $ for $ t> 2 $, and
$$
\Vert e ^ {itH} P_c-F_t \Vert_ {L ^ 1 \to L ^ \infty} \le Ct ^ {- 1}, \qquad t> 2.
$$
Similar dispersion estimates were proved also for solutions to 4D wave equation with a potential.
\smallskip

In \cite {GG2015,GG2017} the Schr\"odinger equation  (\ref {LSE}) is considered in $ \R ^ n $ with  $ n \ge 5 $
when there is an eigenvalue at the zero point of the spectrum. It is shown, in particular, that
there is a time-dependent rank one operator $ F_t $  such that
$ \Vert F_t \Vert_ {L ^ 1 \to L ^ \infty} \le C | t | ^ {2-n / 2} $ for $ | t |> 1 $, and
$$
\Vert e ^ {itH} P_c-F_t \Vert_ {L ^ 1 \to L ^ \infty} \le C | t | ^ {1-n / 2}, \qquad | t |> 1.
$$
With a stronger decay of the potential, the evolution admits an operator-valued expansion
$$
e ^ {itH} P_c (H) = | t | ^ {2-n / 2} A _ {- 2} + | t | ^ {1-n / 2} A _ {- 1} + | t | ^ {- n / 2} A_0,
$$
where $ A _ {- 2} $ and $ A _ {- 1} $ are  finite rank operators $ L ^ 1 (\R ^ n) \to L ^ \infty (\R ^ n) $, while $ A_0 $
maps weighted  $ L ^ 1 $ spaces to  weighted  $ L ^ \infty $ spaces . Main members $ A _ {- 2} $ and $ A _ {- 1} $
equal to zero under certain conditions of the orthogonality of the potential $ V $ to eigenfunction with zero energy.
Under the same orthogonality conditions, the remainder term $ | t | ^ {- n / 2} A_0 $ also maps $ L^1 (\R ^ n) $ to $ L ^ \infty (\R ^ n) $, 
and therefore, the group $ e ^ {itH} P_c (H) $ has the same dispersion decay as free evolution, despite its eigenvalue at zero.
\medskip \\
{\bf $ L ^ p - L ^ q $ decay}  was first established in \cite {MWS1980} for solutions
of the free Klein--Gordon equation $ \ddot \psi = \Delta \psi- \psi $ with initial state $ \psi (0) = 0 $:
\begin {equation} \label {LpLqKG}
\Vert \psi (t) \Vert_ {L ^ q} \le Ct ^ {-d} \Vert \dot \psi (0) \Vert_ {L ^ p}, \qquad t> 1,
\end {equation}
where $ 1 \le p \le 2 $, $ 1 / p + 1 / q = 1 $, and $ d \ge 0 $ is a piecewise-linear function of $ (1 / p, 1 / q) $.
The proofs use the Riesz interpolation theorem.
\smallskip

In \cite {BS1993}, the estimates (\ref {LpLqKG}) were extended to solutions of perturbed Klein--Gordon equation
$$
\ddot \psi = \Delta \psi- \psi + V (x) \psi
$$
with $ \psi (0) = 0 $. The authors show that (\ref {LpLqKG}) holds for $ 0 \le 1 / p-1/2 \le 1 / (n + 1) $. 
The smallest value of $ p $ and the fastest decay rate $ d $ occurs when $ 1 / p = 1/2 + 1 / (n + 1) $, $ d = (n-1) / (n + 1) $. 
The result is proved under the assumption that the potential $ V $ is  smooth and small in a suitable sense. 
For example, the result true when $ | V (x) | \le c (1+ | x | ^ 2) ^ {- \sigma} $, where $ c>0 $ is sufficiently small.
Here  $ \sigma> 2 $ for $ n = 3 $,  $ \sigma> n / 2 $ for  odd $ n\ge 5 $, and $ \sigma> (2n ^ 2 + 3n + 3) / 4 (n + 1) $ for even $ n \ge 4 $. 
The results also apply to the case when $ \psi (0) \ne 0 $.
\smallskip

The seminal article \cite {JSS1991} concerns $ L ^p-L ^q $  decay of solutions to the Schr\"odinger equation (\ref {LSE}).
It is assumed that $ (1+ | x | ^ 2) ^ \alpha V (x) $ is a multiplier in the Sobolev spaces $ H ^ \eta $ for some $ \eta> 0 $ and $ \alpha> n + 4 $,
and the Fourier transform of $ V $ belongs to $ L ^ 1 (\R^n)$. 
Under this conditions,  the main result of \cite {JSS1991} is the following theorem: if $\lambda = 0$ is neither an eigenvalue nor a resonance for $ H$, 
then
\begin {equation} \label {LpLq}
\Vert P_c \psi (t) \Vert_ {L ^ q} \le Ct ^ {- n (1 / p-1/2)} \Vert \psi (0) \Vert_ {L ^ p}, \qquad t> 1,
\end {equation}
where $ 1 \le p \le 2 $ and $ 1 / p + 1 / q = 1 $.
Proofs are based on $ L ^ 1-L ^ \infty $  decay (\ref {l1i}) and the Riesz interpolation theorem.
\smallskip

In \cite {Y2005} estimates (\ref {LpLq}) were proved for  all $ 1 \le p \le 2 $ under suitable conditions on decay of  $ V (x) $
if $ \lam=0 $ is  neither an eigenvalue nor a resonance for $ H $, and for all $ 3/2 <p \le 2 $ otherwise.
\medskip \\
{\bf The Strichartz estimates} were extended i) in \cite {AFVV2010} to the Schr\"odinger magnetic equations in $ R ^ n $ with $ n \ge 3 $,
ii) in \cite {A2015} -  
to wave equations with a magnetic potential in $ R ^ n $ for $ n \ge 3 $, and
iii) in \cite {BG2014} - 
to wave equation in $ \R ^ 3 $
with potentials of the  Kato class.

%%%%%%%%%%%%%%%%%%%%%%%%%%%%%%%%%%%%%%%%%%%%%%%%%%
%%%%%%%%%%%%%%%%%%%%%%%%%%%%%%%%%%%%%%%%%%%%%%%%%%
\setcounter {equation} {0}
\section {Numerical Simulation of Soliton Asymptotics} \la {s9}
%%%%%%%%%%%%%%%%%%%%%%%%%%%%%%%%%%%%%%%%%%%%%%%%%%
Here we describe the results of joint work with Arkady Vinnichenko (1945-2009) on numerical simulation of 
i) global attraction to solitons \eqref {att} and \eqref {attN},
and ii) adiabatic effective dynamics of solitons (\ref {effd2}) for relativistic-invariant one-dimensional
nonlinear wave equations.
In 
 \cite {KMV2004} can be found an additional information.

%%%%%%%%%%%%%%%%%%%%%%%%%%%%%%%%%%%%%%%%%%%%%%%%%%
%%%%%%%%%%%%%%%%%%%%%%%%%%%%%%%%%%%%%%%%%%%%%%%%%%
\subsection {Kinks of relativistic-invariant Ginzburg--Landau equations}
%%%%%%%%%%%%%%%%%%%%%%%%%%%%%%%%%%%%%%%%%%%%%%%%%%, 
First, let us describe  numerical simulations of 
 solutions to relativistic-invariant 1D nonlinear wave equations with polynomial nonlinearity
\begin {equation} \label {GL}
\ddot \psi (x, t) = \psi '' (x, t) + F (\psi (x, t)), \qquad \mbox {where} \quad F (\psi): = - \psi ^ 3 + \psi.
\end {equation}
Since $F (\psi) = 0$ for $\psi = 0, \pm1$, there are three  equilibrium state : $S (x)\equiv  0, + 1, -1$.
This equation  formally is equivalent to a Hamiltonian system (\ref{hasys}) with the Hamiltonian
\begin {equation} \label {hamGL}
\cH (\psi, \pi) = \int [\frac12 | \pi (x) | ^ 2 + \frac12 | \psi '(x) | ^ 2 + U (\psi (x))] \, dx,
\end {equation}
where the potential $ U (\psi) = \frac {\psi ^ 4} {4} - \frac {\psi ^ 2} {2} +\frac 14$.
This Hamiltonian is finite  for functions $ (\psi, \pi)$ from the space  $\cE_c $, defined in (\ref {hamdal})-- (\ref {lim}) with $C_\pm=\pm 1$, 
for which the convergence
$$
\psi (x) \to \pm 1, \qquad | x | \to \pm \infty
$$
is sufficiently fast.

The corresponding  potential $ U (\psi) = \frac {\psi ^ 4} {4} - \frac {\psi ^ 2} {2} +\frac 14$ has minima at $ \psi = \pm 1 $ 
and a maximum at $ \psi = 0 $.
Respectively,
two finite energy solutions  $\psi={\pm 1}$ are stable, and the solution $\psi=0$ with infinite energy  is unstable.
Such potentials with two wells are called  potentials of Ginzburg-Landau type.

Besides the constant stationary solutions $ S (x) \equiv 0, + 1, -1 $, there is also a non-constant one $ S (x) = \tanh \5 {x} / {\sqrt2} $, which is called  ``kink''.
Its shifts and reflections $ \pm S (\pm x-a) $ are also stationary solutions as well as their Lorentz transforms
$$ 
\pm S (\gamma (\pm x-a-vt)),\qquad  \gamma = 1 / {\sqrt {1-v ^ 2}}, \quad | v | <1. 
$$
These are uniformly moving ``travelling waves'' (i.e. solitons). The kink is strongly compressed when the velocity $ v $ is close to $ \pm 1 $. This compession is known as the ``Lorentz contraction''.
\medskip\\
{\bf Numerical Simulation.}
Our numerical experiments show a decay  of finite energy solutions to a finite set
of kinks and dispersion waves
outside  the kinks, that corresponds to the asymptotics of \eqref {attN}.
One of the experiments is shown on Fig. \ref {fig-4}: a finite energy solution to the  equation (\ref {GL}) decays to three kinks.
Here the vertical line is the time axis, and the horizontal line is the   space axis.
The spatial scale redoubles at $ t = $ 20 and $ t = $ 60.

Red color corresponds to $ \psi> 1+ \varepsilon $ values, blue color to $ \psi <-1 - \varepsilon $ values,
and the yellow one,  to intermediate values $ -1 + \varepsilon <\psi <1 - \varepsilon $, where $ \ve> 0 $ is sufficiently small.
Thus, the yellow stripes represent  the kinks, while the blue and red zones outside the yellow stripes are filled with dispersion waves.

For $ t = 0 $, the solution begins with a rather chaotic behavior, when there are no visible kinks.
After 20 seconds, three separate kinks appear, which subsequently move almost uniformly.

%%%%%%%%%%%%%%%%%%%%%%%%%%%%%%%%%%%%
\begin{figure}[htbp]
\begin{center}
\includegraphics[width=1.00\columnwidth]{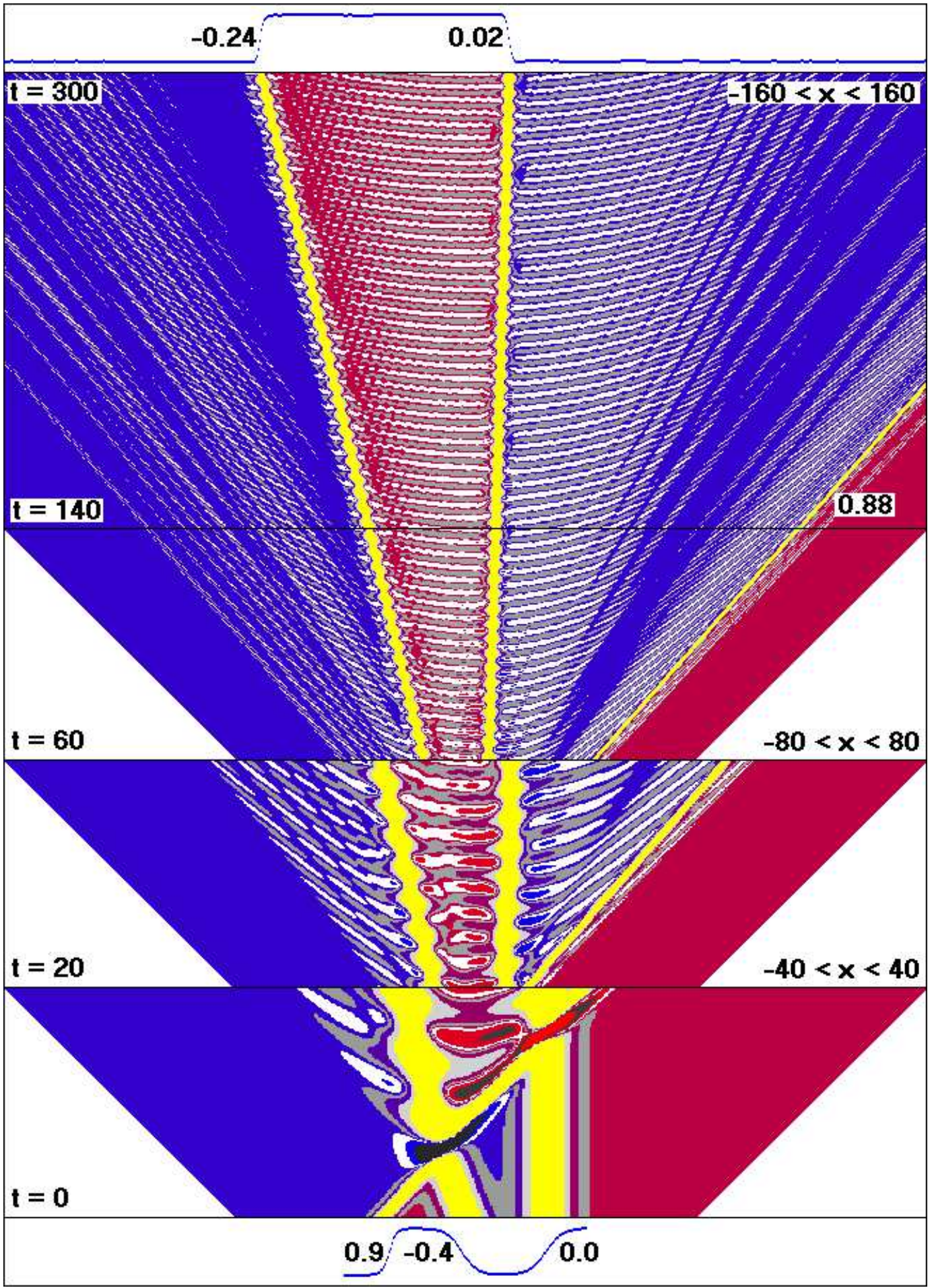}
\caption{Decay to three kinks}
\label{fig-4}
\end{center}
\end{figure}
%%%%%%%%%%%%%%%%%%%%%%%%%%%%%%%%%%%%
\noindent {\bf The Lorentz contraction.}
The left kink moves to the left at a low speed $ v_1 \approx $ 0.24, the central kink is almost standing, 
because its velocity  $ v_2 \approx 0.02 $ is very small, and the right kink moves very fast with the speed $ v_3 \approx $ 0.88.
The Lorentz spatial contraction $ \sqrt {1-v_k ^ 2} $ is clearly visible in this picture:
the central kink is wide, the left is a bit narrower, and the right one is quite narrow.
\smallskip \\
{\bf The Einstein time delay.}
Also, the Einstein time delay is also very pronounced. 
Namely, all three kinks oscillate due to  the presence of  nonzero eigenvalue in the linearized equation at the kink. Indeed,
substituting $ \psi (x, t) = S (x) + \varepsilon \varphi (x, t) $ in (\ref {GL}), we get in the first order approximation the linearized equation
\be\la{lineq}
\ddot \varphi (x, t) = \varphi '' (x, t) -2 \varphi (x, t) -V (x) \varphi (x, t),
\ee
 where the potential $V (x) = 3S ^ 2 (x) -3 = - \frac {3} {\cosh ^ 2 {x} / {\sqrt2}}$ decays exponentially for large $ | x | $.
It is a great joy that for this potential the spectrum of the corresponding
\textit {Schr\"odinger operator} $ H: = - \frac {d ^ 2} {dx ^ 2} + 2 + V (x) $ is well known  \cite {Lamb80}.
Namely, the operator $ H $ is non-negative, and its continuous spectrum coincides with $ [2, \infty) $.
It turns out that $ H $ still has a two-point discrete spectrum: the points $ \lambda = 0 $ and $ \lambda = \frac32 $.
Exactly this nonzero eigenvalue  is responsible  for the pulsations that we observe for the central slow kink, with  the frequency
$ \omega_2 \approx \sqrt {\frac32} $ and period $ T_2 \approx 2 \pi / \sqrt {\frac32} \approx 5 $. 
On the other hand, for  fast kinks, the ripples are much slower, i.e., the corresponding period is longer. 
This  time delay agrees numerically with the Lorentz formulas, that confirms the relevance of 
these results
of numerical simulation. 
\smallskip \\
{\bf Dispersion waves.}
An analysis of dispersion waves provides additional confirmation. Namely, the space outside the kinks in Fig. \ref {fig-4} is 
filled with dispersion waves, whose values are very close to $ \pm 1 $, with an accuracy $ 0.01 $.
These waves satisfy with high accuracy,  the linear Klein--Gordon equation, which is obtained by linearization
of the Ginzburg--Landau equation (\ref {GL}) at the stationary solutions $ \psi_\pm \equiv \pm 1 $:
$$
\ddot \varphi (x, t) = \varphi '' (x, t) +2 \varphi (x, t).
$$
The corresponding dispersion relation $ \omega ^ 2 = k ^ 2 + 2 $ determines the group velocities of high-frequency wave packets:
\begin {equation} \label {ev}
\omega '(k)= \frac k {\sqrt {k ^ 2 +2}} = \pm \frac {\sqrt {\omega ^ 2 -2}} \omega.
\end {equation}
These wave packets are clearly visible in Fig. \ref {fig-4} as straight lines, whose propagation speeds converge to $ \pm 1 $.
This convergence is explained by the high-frequency limit  $ \omega '(k)\to \pm1 $  as $ \omega \to \pm \infty $.
For example, for dispersion waves emitted by central kink, the frequencies $ \omega = \pm n \omega_2 \to \pm \infty $
are generated by the polynomial nonlinearity in (\ref {GL}) in accordance with Fig. \ref {fig-2}.

%%%%%%%%%%%%%%%%%%%%%%%%
\begin {remark} \label {rrm}
{\rm
These observations of dispersion waves agree with the radiation mechanism  from Remark \ref {phys}.}
\end {remark}
%%%%%%%%%%%%%%%%%%%%%%%%%%
The nonlinearity in (\ref {GL}) is chosen exactly because  of well-known spectrum of the linearized equation (\ref{lineq}). 
In numerical experiments \cite {KMV2004}  also more general nonlinearities of the Ginzburg - Landau type were considered . 
The results were qualitatively the same: for “any” initial data, the solution decays for large times to a sum of kinks and dispersion waves.
Numerically, this is clearly visible, but rigorous justification remains an open problem.

%%%%%%%%%%%%%%%%%%%%%%%%%%%%%%%%%%%%%%%%%%%%%%%%%%%%%%%%%%%%%%%%%%% 
%%%%%%%%%%%%%%%%%%%%%%%%%%%%%%%%%%%%%%%%%%%%%%%%%%%%%%%%%%%%%%%%%%% 
\subsection{Numerical observation of soliton asymptotics}
%%%%%%%%%%%%%%%%%%%%%%%%%%%%%%%%%%%%%%%%%%%%%%%%%%%%%%%%%%%%%%%%%%% 
Besides the kinks the numerical experiments \cite{KMV2004} also resulted in the soliton-type asymptotics (\ref{attN}) 
and adiabatic effective dynamics of type (\ref{effd2}) for complex solutions to the 1D relativistically-invariant nonlinear wave equations
(\ref{NWEn}). Polynomial potentials of the form
\begin{equation}\label{pop2}
  U(\psi)=a|\psi|^{2m}-b|\psi|^{2n},
\end{equation}
were considered with
 $a,b>0$ and $m>n=2,3,\dotsc$. Respectively,
\begin{equation}\label{pop2f}
   F(\psi)=2am|\psi|^{2m-2}\psi-2bn|\psi|^{2n-2}\psi.
\end{equation}
The parameters $a,b,m,n$ were taken as follows:
$$
\begin{array}{rrrlr}
N~~~~~&a~~~~~&~~~~m&~~~~~~b&~~~~n\\
1~~~~~&1~~~~&~~~~3&~~~~~~0.61&~~~~2\\
2~~~~~&10~~~~&~~~~4&~~~~~~2.1&~~~~2\\
3~~~~~&10~~~~&~~~~6&~~~~~~8.75&~~~~5
\end{array}
$$
 Various ``smooth''  initial functions $ \psi (x, 0), \dot\psi (x, 0) $ with  supports on the interval $[-20,20]$ were considered.
The second order finite-difference scheme with $ \Delta x, \Delta t \sim 0.01, 0.001$ was employed. 
In all cases the asymptotics of type (\ref{attN}) were obsereved with the numbers of solitons $0,1,3,5$ for $t> 100$.

%%%%%%%%%%%%%%%%%%%%%%%%%%%%%%%%%%%%%%%%%%%%%%%%%%%%%%%%%%%%%%%%%%% 
%%%%%%%%%%%%%%%%%%%%%%%%%%%%%%%%%%%%%%%%%%%%%%%%%%%%%%%%%%%%%%%%%%%
\subsection{Adiabatic effective dynamics of relativistic solitons}
%%%%%%%%%%%%%%%%%%%%%%%%%%%%%%%%%%%%%%%%%%%%%%%%%%%%%%%%%%%%%%%%%%%
\begin{figure}[htbp]
\begin{center}
\includegraphics[width=0.9\columnwidth]{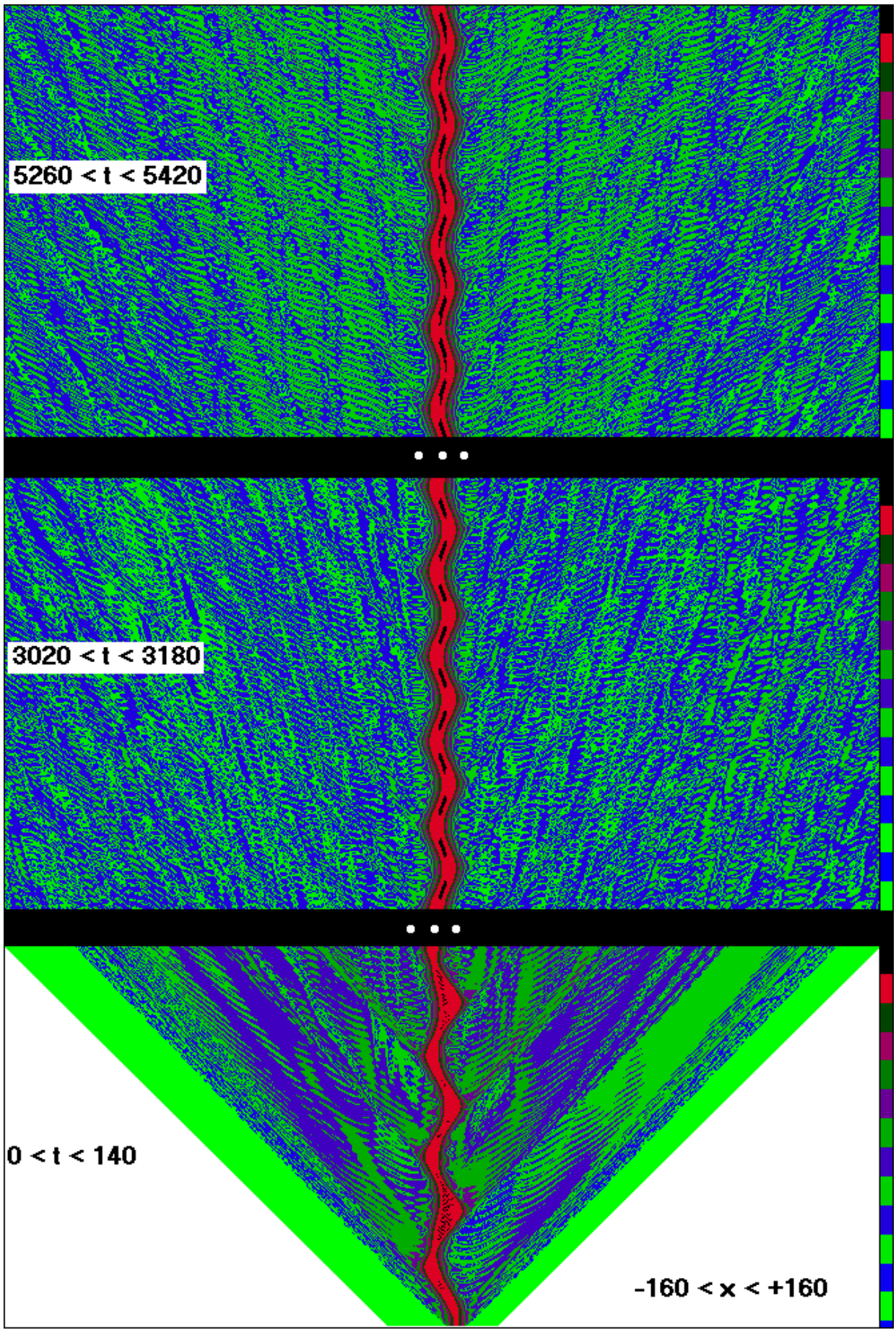}
\caption{Adiabatic effective dynamics of relativistic solitons}
\label{fig-7}
\end{center}
\end{figure}

In the numerical experiments \cite{KMV2004} was also observed the adiabatic effective dynamics of type (\ref{effd2}) for soliton-like solutions
of type (\ref{asol})
to the 1D equations (\ref{NWEn}) with a~slowly varying external potential (\ref{asolV}):
\begin{equation}\label{EQp}
  \ddot \psi(x,t)=\psi''(x,t)-\psi(x,t)+F(\psi(x,t))-V(x)\psi(x,t), \qquad x\in\mathbb R.
\end{equation}
This equation formally is equivalent to the Hamilton system (\ref{hasys}) with the Hamilton functional
\begin{equation}\label{HAMp}
  \cH_V(\psi,\pi)=\int [\frac12 |\pi(x)|^2+ \frac12 |\psi^\prime (x)|^2+U(\psi(x))+\frac12 V(x) |\psi(x)|^2 ]\,dx.
\end{equation}
The soliton-like solutions are of the form (cf. (\ref{asol}))
\begin{equation}\label{swei}
  \psi(x,t)\approx e^{i\Theta(t)}\phi_{\omega(t)}(\gamma_{v(t)}(x-q(t)) ).
\end{equation}
Numerical experiments \ci{KMV2004} qualitatively confirm the adiabatic effective Hamilton dynamics
for the parameters $\Theta,\omega,q$, and $v$, but its rigorous justification is still not established. Figure~\ref{fig-7} represents solutions to equation (\ref{EQp}) with the potential (\ref{pop2}), where $a=10$, $m=6$ and $b=8.75$, $n = 5$.
The potential $V(x)=-0.2\cos (0.31 x)$ and the  initial conditions read
\begin{equation}\label{inco}
   \psi(x,0)= \phi_{\omega_0}(\gamma_{v_0}(x-q_0)), \qquad \dot \psi(x,0)=0,
\end{equation}
where $v_0=0$, $\omega_0=0.6$ and $q_0=5.0$.
Note, that the initial state does not belong to  solitary manifold. An effective width (half-amplitude) of the solitons is in the range $[4.4, 5.6 ]$.
It is quite small when compared with the spatial period of the potential $2\pi/0.31 \sim 20$.
The results of  numerical simulations are shown on Figure~\ref{fig-7}:
\medskip\\
$\bullet$ Blue and green colors represent a dispersion wave with values $|\psi(x,t)| <0.01$, while  red color represents a soliton with values $|\psi(x,t)|\in [0.4, 0.8]$.
\smallskip\\
$\bullet$ The soliton trajectory (``red snake'') corresponds to oscillations of a classical particle in the potential $\!V(x)\!$.
\smallskip\\
$\bullet$ For $0\!<t<\!140$, the solution is rather distant from the solitary manifold, and the radiation is intense.
\smallskip\\
$\bullet$ For $3020\!<\!t\!< \!3180 $, the solution approaches the solitary manifold, and the radiation
weakens. The os\-cillation amplitude of the soliton is almost unchanged for a~long time, confirming  Hamilton type of effective dynamics.
\smallskip\\
$\bullet$ However, for $5260\! <\!t\!<\!5420$, the amplitude of the soliton oscillation is halved.
This suggests that at a~large time scale the deviation from the Hamilton effective dynamics becomes essential.
Consequently, the effective dynamics gives a~good approximation only on an adiabatic time scale of type $t\sim \varepsilon^{-1}$.
\smallskip\\
$\bullet$ The deviation from the Hamilton effective dynamics is due to radiation, which plays the role of dissipation.
\smallskip\\
$\bullet$
The radiation is realized as  dispersion waves which bring the energy to the infinity.
The dispersion waves combine into uniformly moving wave packets with discrete set of group velocities, as in Fig.~\ref{fig-4}.
The magnitude of the solution is of order $\sim 1$ on the trajectory of the soliton, while the values of the dispersion waves
is less than $0.01 $ for $t>200$, so that their energy density does not exceed $0.0001$.
The amplitude of the dispersion waves decays at large times.
\smallskip\\
$\bullet$ In the limit $t\to\pm\infty$, the soliton should converge to a static position corresponding to a local minimum
of the potential  $V(x)$. However, the numerical observation of this ``ultimate stage'' is hopeless since the rate of the convergence
decays with the decay of the radiation.

\appendix

\setcounter{section}{0}
\setcounter{equation}{0}
\protect\renewcommand{\thesection}{\Alph{section}}
\protect\renewcommand{\theequation}{\thesection.\arabic{equation}}
\protect\renewcommand{\thesubsection}{\thesection.\arabic{subsection}}
\protect\renewcommand{\thetheorem}{\Alph{section}.\arabic{theorem}}

\section{Attractors and Quantum Mechanics}

The foregoing results on attractors of  nonlinear Hamilton equations were suggested by fundamental postulates of quantum theory,
primarily Bohr's postulate on transitions between quantum stationary orbits. 
Namely, in 1913 Bohr suggested the following two postulates which gives the ``Columbus''  solution of the problem of stability and radiation of atoms and molecules \cite{Bohr1913}:
$$
\left|
\ba{c}
\mbox{\it {\bf B1.} Atoms and molecules are permanently on some stationary orbits $|E_n \rangle$ with}
\smallskip\\
\mbox{\it energies $E_n$, and sometimes make transitions between the orbits,}
\medskip\\
|E_n\rangle\mapsto |E_{n'}\rangle.
\ea
\right|
$$

$$
\left|
\ba{c}
\mbox{\it {\bf B2.} Such transition is followed by  radiation of an electromagnetic wave of frequency}
\smallskip\\
\mbox{\it $\om_{nn'}=\om_{n'}-\om_n$,}
\qquad
\mbox{\it where \quad $\om_k=E_k/\hbar$}
\ea
\right|
$$
Both these postulates should become theorems in  discovered later
quantum theory of Schr\"odinger and Heisenberg. However, this did not happen, and both postulates are still actively used in quantum theory. 
This lack of theoretical clarity hinders the progress in the theory (e.g., in superconductivity and  in nuclear reactions),
and in
numerical simulation of many engineering processes (e.g., of laser radiation and quantum amplifiers) since a computer can solve dynamical equations but cannot  take  postulates into account.
\smallskip

\subsection{On dynamical interpretation of quantum jumps}
The simplest dynamic interpretation of the postulate {\bf B1} is the attraction to stationary orbits (\ref{atU}) for any finite energy quantum trajectory $\psi(t)$.
This means that  stationary orbits form a global attractor of the corresponding quantum dynamics.
However, this attraction contradicts the Schr\"odinger linear equation due to the superposition principle.
Thus, Bohr's transitions {\bf B1} in the linear theory do not exist.
\medskip

It is natural to suggest that the attraction (\ref{atU}) holds for a~nonlinear modification of the linear Schr\"odinger theory. On the other hand,
it turns out that even the original Schr\"odinger theory is nonlinear, because it involves interaction with the Maxwell field.
The corresponding nonlinear Maxwell--Schr\"odinger system is contained essentially in the first Schr\"odinger's articles \ci{Schr1926}:
\begin{equation} \label{SM}
\!\!\left\{\!\!
\begin{array}{rcl}
i\dot\psi(x,t)\!\!&\!\!=\!\!&\!\!\displaystyle\frac12[-i\nabla +\bA(x,t)+\bA^{\rm ext}(x,t)]^2\psi+[A_0(x,t)+A_0^{\rm ext}(x)]\psi
		\\\\
 \Box A_{\alpha}(x,t)\!\!&\!\!=\!\!&\!\!4\pi J_{\alpha}(x,t),\qquad \alpha=0,1,2,3
\end{array}\!\!\right|, \qquad x\in\mathbb R^3,
\end{equation}
where  the units are chosen so that $\hbar=e=m=c=1$.
Maxwell's equations are written here in the 4-dimensional form, where $A=(A_0,\bA)=(A_0,A_1, A_2,$ $A_3)$
denotes  4-dimensional potential of the Maxwell field with the Lorentz gauge 
$\dot A_0+\nabla\cdot\bA=0$.  Further, $A^{\rm ext} = (A_0^{\rm ext}, \bA^{\rm ext})$ is an external 4-potential, 
and $J=(\rho,j_1, j_2 , j_3)$ is the 4-dimensional current.
To make these equations a~closed system, we must also express the density of charges and currents via the wave function:
\begin{equation} \label{roj}
J_0(x,t)=|\psi(x,t)|^2; J_k(x,t)=[(-i\nabla_k+A_k(x,t)+A_k^{\rm ext}(x,t))\psi(x,t)]\cdot\psi(x,t),
\end{equation}
where $k=1,2,3$, and ``$\cdot$'' denotes the scalar product of two-dimensional real vectors corresponding to complex numbers.
In particular, these expressions satisfy the continuity equation $\dot\rho+\dv j=0$ for any
solution of the Schr\"odinger equation with arbitrary potentials \cite[Section 3.4]{K2013}.
\medskip

System (\ref{SM}) is nonlinear in $(\psi, A)$ although the Schr\"odinger equation is formally linear in $\psi$. Now the question arises:
what should be ``stationary orbits'' for the nonlinear hyperbolic system (\ref{SM})?
It is natural to suggest that these are  solutions of type
\begin{equation}\label{ss}
(\psi(x)e^{-i\omega t},~A(x))
\end{equation}
in the case of  static external potentials $A^{\rm ext}(x,t)=A^{\rm ext}(x)$.

Indeed, in this case  functions (\ref{ss}) give stationary distributions of charges and currents (\ref{roj}).
Moreover, these functions are the  trajectories of one-parameter subgroups of the
symmetry group $U(1)$ of the system (\ref{SM}). Namely, for any solution
$(\psi(x,t), A(x,t))$ and $\theta\in\R$ the functions
\begin{equation}\label{Ut}
U_\theta(\psi(x,t),~A(x,t)):=(\psi(x,t)e^{i\theta},~A(x,t))
\end{equation}
are also solutions.
The same remarks apply to the Maxwell---Dirac system introduced by Dirac in 1927:
\begin{equation}\label{DM}
\left\{
\begin{aligned}
& \sum_{\alpha=0}^3 \gamma^\alpha[i\displaystyle \nabla_\alpha -A_\alpha(x,t)-A_\alpha^{\rm ext}(x,t)]\psi(x,t) =m \psi(x,t)\\
& \Box\, A_\alpha(x,t) =J_\alpha(x,t) := \overline{\psi(x,t)}\gamma^0\gamma_\alpha\psi(x,t), \quad \alpha=0,\dotsc,3\\
\end{aligned}\right| \quad x\in\mathbb R^3,
\end{equation}
where $\nabla_0:=\partial_t$.
Thus, Bohr's transitions {\bf B1} for the systems (\ref{SM}) and (\ref{DM}) with a static external potential $A^{\rm ext}(x,t)=A^{\rm ext}(x)$
can be interpreted as the long-time asymptotics
\begin{equation}\label{as}
(\psi(x,t),~A(x,t))\sim (\psi_\pm(x)e^{-i\omega_\pm t},~A_\pm(x,t)), \qquad t\to\pm\infty
\end{equation}
for every finite energy solution, where the asymptotics hold in  local energy norms.
The maps $U_\theta$ form a group isomorphic to $U(1)$, and the functions (\ref{ss}) are the trajectories of its one-parametric subgroups.
Hence, the asymptotics (\ref{as}) correspond to our general conjecture (\ref{at10}) with the symmetry group $G=U(1)$.
\smallskip

Furthermore, in the case of zero external potentials these systems are translation-invariant.
Respectively, for their solutions one should expect the soliton asymptotics of type \eqref{attN} 
 in global energy norms
as $t\to\pm\infty$:
\begin{eqnarray}\label{SA}
\psi(x,t)&\sim&\displaystyle\sum\limits_{k} \psi_\pm^k(x-v^k_\pm t) e^{i\Phi_\pm^k(x,t)}+\varphi_\pm(x,t),
\\
 A(x,t)&\sim&\displaystyle\sum\limits_{k} A_\pm^k(x-v^k_\pm t)+A_\pm(x,t).
\end{eqnarray}
Here $\Phi_\pm^k(x,t) $ are suitable phase functions, and each soliton $(\psi_\pm^k(x-v^k_\pm t) e^{i\Phi_\pm^k(x,t)}, \,A_\pm^k(x-v^k_\pm t))$
is a solution to the corresponding nonlinear system, 
while $\varphi_\pm(x,t)$ and $A_\pm(x,t)$ represent some dispersion waves which are solutions to the free Schr\"odinger 
and Maxwell equations respectively. The existence of the solitons to the 
Maxwell--Schr\"odinger and 
Maxwell--Dirac systems was established in \cite{CG2004} and \ci{EGS} respectively.
\medskip

The asymptotics (\ref{as}) and (\ref{SA}) are not proved yet for the Maxwell--Schr\"odinger and Maxwell--Dirac equations
(\ref{SM}) and (\ref{DM}). One could expect that these asymptotics should follow by suitable modification of the arguments from Section \ref{s5}.
Namely, let the time spectrum of an omega-limit trajectory
$\psi(x,t)$ contain at least two different frequencies $\omega_1\ne\omega_2$:
for example, $\psi(x,t)=\psi_1(x)e^{-i\omega_1 t}+\psi_2(x)e^{-i\omega_2 t}$.
Then the currents $J_\alpha(x,t)$ in the systems (\ref{SM}) and (\ref{DM}) contains the terms with the harmonics 
$e^{i n\Delta t}$ with $n\in\Z$,  where $\Delta:=\omega_1-\omega_2\ne 0$.
Thus the nonlinearity inflates the spectrum as in $U(1)$-invariant equations, considered in Section \ref{s5}.

In it own turn, these  harmonics $e^{i n\Delta t}$ with $n\ne 0$ on the right hand side of the Maxwell equations induce the radiation of electromagnetic waves
with the frequencies $n \Delta$ according to the limiting amplitude principle (\ref{lap}) since the continuous spectrum 
of the Maxwell generator is  $\R\setminus 0$.
Finally, this radiation brings the energy to infinity which is impossible for omega-limit trajectories.
This contradiction suggests the validity of the one-frequency asymptotics (\ref{as}).

Methods of Section \ref{s5} give 
a rigorous justification of 
similar arguments for $U(1)$-invariant equations (\ref{KG1}) and (\ref{KGN})--(\ref{Dn}).
However,  a rigorous justification for the systems (\ref{SM}) and (\ref{DM}) is still an open problem.

\subsection{Bohr's postulates  by perturbation theory}
The remarkable success of the Schr\"odinger theory was the explanation
of the Bohr' postulates
in the case of {\it static external potentials}
by {\it perturbation theory} applied to the 
{\it coupled Maxwell--Scr\"odinger equations} (\ref {SM}).
Namely, as a first approximation, the time-dependent 
fields $ \bA (x, t) $ and $ A^0 (x, t) $ in the Schr\"odinger equation of the system
 (\ref {SM})  can be neglected:
 \be \la {Sc0}
 i \hbar \dot \psi (x, t) =H\psi (x, t):=
 \fr 1 {2m} [- i \hbar \na- \ds \frac ec \bA_\ext (x)] ^ 2 \psi (x, t) + eA_\ext^0 (x) \psi (x, t),
\ee
For 
``sufficiently good''
external potentials and initial conditions
any finite energy solution
can be expanded in  eigenfunctions
 \be \la {Sexp}
 \psi (x, t) = \sum_n C_n \psi_n (x) e ^ {- i \om_n t} + \psi_c (x, t), \qquad
 \psi_c (x, t) =
 \int C (\om) e ^ {- i \om t} d \om,
\ee
where integration is performed over the continuous spectrum of the Schr\"odinger operator $ H $, and the integral decays
as $ t \to \infty $
in each bounded domain $ | x | \le R $, see, for example, \ci [Theorem 21.1]{KopK2012}.
The
substitution of this expansion into the expression for currents (\ref {roj}) gives the series
 \be \la {jexp}
J (x, t) = \sum_ {nn '} J_ {nn'} (x) e ^ {- i \om_ {nn '} t} + c.c. + J_c (x, t),
\ee
where $ J_c (x, t) $ has a continuous frequency spectrum. 
Therefore, the currents on the right hand side of 
the
Maxwell equation from (\ref {SM}) contains,
besides the continuous spectrum, only
discrete frequencies $ \om_ {nn '} $.
Hence, the discrete spectrum of the corresponding 
Maxwell field also contains only these frequencies $ \om_ {nn '} $.
This proves the Bohr rule {\bf B2} 
{\it in the first order of perturbation theory}, since this calculation ignores the inverse effect of radiation onto the atom.

Moreover, these arguments also clarify  the asymptotics (\ref {as}). Namely, 
the currents (\ref{jexp}) on the right hand of 
the Maxwell equation from (\ref {SM}) produce the radiation when 
nonzero frequencies  $ \om_ {nn '} $ are present. However, this radiation
cannot last forever since the total energy is finite. Hence, in the long-time 
limit should remain only $ \om_ {nn '}=0 $ which means exactly one-frequency asymptotics (\ref{as}) and the limiting  stationary Maxwell field.

\subsection{Conclusion}

The  discussion above suggests that Bohr's postulates
cannot be explained by linear Schr\"odinger equation alone but 
admit a hypothetical explanation in the framework
of the coupled Maxwell--Schr\"odinger equation.

This fact was the cause
 of heated discussions by Einstein
with Bohr and other physicists \ci {B1949}.
In \ci{Heis1961,Heis1966}, Heisenberg began developing 
a nonlinear theory of elementary particles.

%%%%%%%%%%%%%%%%%%%%%%%%%%%%%%%%%%%%%%%%%%%%%%%%%%%%%%
%%%%%%%%%%%%%%%%%%%%%%%%%%%%%%%%%%%%%%%%%%%%%%%%%%%%%%
%%%%%%%%%%%%%%%%%%%%%%%%%%%%%%%%%%%%%%%%%%%%%%%%%%%%%

\bigskip

\noindent {\it  
Institute of Information Transmission Problems RAS, Moscow
 127994,  Russia.}\\
\noindent {\it e-mail}: alexander.komech@univie.ac.at,\qquad elena.kopylova@univie.ac.at
\vskip 0.5cm

\end{document}